\documentclass[a4paper,onecolumn,unpublished]{preamble/quantumarticle}

\pdfoutput=1
\usepackage{hyperref}
\usepackage{bookmark}

\usepackage{mathtools}
\usepackage{amsmath}
\usepackage{amssymb}
\usepackage{amsthm}
\usepackage{thmtools} 
\usepackage{graphicx}
\usepackage{subcaption} 
\usepackage{amsfonts}
\usepackage{bm}
\usepackage[capitalize]{cleveref}
\usepackage{enumitem}
\usepackage{calc}
\usepackage{stmaryrd}
\usepackage[most]{tcolorbox}

\usepackage{tikz}
\usepackage{tikzscale}  
\usepackage{preamble/tikzit} 

\usepackage[style=alphabetic]{biblatex}

\usepackage[left=`,right=',leftsub=``,rightsub='']{dirtytalk} 

\hypersetup{
    colorlinks = true,
    linkcolor = lime!50!black,
    anchorcolor = lime!50!black,
    citecolor = red!70!black,
    filecolor = lime!50!black,
    urlcolor = lime!50!black
}
\bookmarksetup{startatroot}

\newtheoremstyle{mytheoremstyle}  
{}                           
{}                           
{}                      
{}                              
{\bfseries}                     
{.}                             
{ }                             
{}                              

\theoremstyle{mytheoremstyle}
\newtheorem{theorem}{Theorem}[section]
\newtheorem*{theorem*}{Theorem}
\newtheorem{definition}[theorem]{Definition}
\newtheorem{proposition}[theorem]{Proposition}

\newtheorem{corollary}[theorem]{Corollary}

\newtheorem*{conjecture*}{Conjecture}
\newtheorem{remark}[theorem]{Remark}

\numberwithin{equation}{section}

\tcbset{
  mytheostyle/.style={
    enhanced,
    colback=white,          
    colframe=black,         
    coltitle=black,         
    boxrule=0.6pt,          
    arc=3pt,                
    outer arc=3pt,
    left=6pt, right=6pt, top=4pt, bottom=4pt,
    before skip=8pt, after skip=8pt,
    fonttitle=\bfseries,    
    attach boxed title to top left={yshift=-2pt, xshift=4pt},
    boxed title style={
      colback=white,        
      boxrule=0pt,          
      sharp corners,
    },
  }
}

\newtcbtheorem[use counter=theorem,number within=section]{example}{Example}{
  mytheostyle,
}{ex}

\renewcommand{\thesection}{\arabic{section}} 
\makeatletter
\renewcommand{\@seccntformat}[1]{\csname the#1\endcsname\quad} 
\makeatother

\newlist{proofparts}{description}{1}
\setlist[proofparts,1]{%
  font=\normalfont\textsf,
  itemindent=-10pt,
  topsep=2pt,
  itemsep=8pt,
  labelsep=0.75ex
}

\tikzstyle{Z dot}=[inner sep=0mm, minimum size=2mm, shape=circle, draw=black, fill=zxGreen, tikzit fill={rgb,255: red,216; green,248; blue,216}, outer sep=-0.5mm]
\tikzstyle{Z phase dot}=[draw=black, fill=zxGreen, shape=rectangle, minimum size=4.5mm, rounded corners=1.8mm, inner sep=0.5mm, outer sep=-0.5mm, scale=0.8, tikzit shape=circle, tikzit fill={rgb,255: red,216; green,248; blue,216}, font={\footnotesize\boldmath}]
\tikzstyle{Z tiny phase dot}=[Z dot, draw=black, tikzit fill={rgb,255: red,216; green,248; blue,216}, font={\footnotesize\boldmath\tiny}]
\tikzstyle{X dot}=[shape=circle, draw=black, fill=zxRed, tikzit fill={rgb,255: red,221; green,165; blue,165}, inner sep=0 mm, minimum size=2 mm, outer sep=-0.5mm]
\tikzstyle{X phase dot}=[Z phase dot, draw=black, fill=zxRed, tikzit fill={rgb,255: red,221; green,165; blue,165}]
\tikzstyle{X tiny phase dot}=[X dot, draw=black, tikzit fill={rgb,255: red,221; green,165; blue,165}, font={\footnotesize\boldmath\tiny}]
\tikzstyle{hadamard}=[fill=zxHad, draw=black, shape=rectangle, inner sep=0.6mm, minimum height=1.5mm, minimum width=1.5mm, tikzit fill=yellow, font={\footnotesize\boldmath}]
\tikzstyle{paulibox}=[draw=black, shape=rectangle, fill=white, minimum size=1em, inner sep=0.2em, scale=0.85, font={\scriptsize}, outer sep=-0.5mm]
\tikzstyle{vertex}=[inner sep=0mm, minimum size=1mm, shape=circle, draw=black, fill=black, tikzit category=misc]
\tikzstyle{vertex set}=[inner sep=0mm, minimum size=1mm, shape=circle, draw=black, fill=white, font={\footnotesize\boldmath}, tikzit category=misc]
\tikzstyle{small black dot}=[fill=black, draw=black, shape=circle, inner sep=0pt, minimum width=1.2mm, tikzit category=circuit]
\tikzstyle{scalar}=[shape=rectangle, text height=1.5ex, text depth=0.25ex, yshift=0.5mm, fill=white, draw=black, minimum height=5mm, yshift=-0.5mm, minimum width=5mm, font={\small}]
\tikzstyle{empty diagram}=[draw={gray!40!white}, dashed, shape=rectangle, minimum width=1cm, minimum height=1cm, tikzit category=misc]
\tikzstyle{sLabel}=[font={\scriptsize}, tikzit draw=black, auto]
\tikzstyle{fault-location}=[fill=white, draw=black, shape=circle, minimum size=2mm, inner sep=0mm, outer sep=-0.5 mm, regular polygon, regular polygon sides=8, font=\tiny]

\tikzstyle{hadamard edge}=[-, dashed, dash pattern=on 2pt off 0.5pt, thick, draw={rgb,255: red,68; green,136; blue,255}]
\tikzstyle{fault-free}=[-, draw={rgb,255: red,177; green,98; blue,255}, line width=1.2pt]
\tikzstyle{component}=[-, style=dashed, draw={rgb,255: red,0; green,128; blue,128}]
\tikzstyle{X Web}=[-, preaction={ultra thick, draw={rgb,255: red,208; green,0; blue,0}, opacity=0.4}, tikzit draw={rgb,255: red,208; green,0; blue,0}]
\tikzstyle{Z Web}=[-, preaction={ultra thick, draw={rgb,255: red,0; green,207; blue,0}, opacity=0.4}, tikzit draw={rgb,255: red,0; green,207; blue,0}]
\tikzstyle{X Web Overlay}=[-, ultra thick, draw={rgb,255: red,208; green,0; blue,0}, opacity=0.4, tikzit draw={rgb,255: red,208; green,0; blue,0}]
\tikzstyle{Z Web Overlay}=[-, ultra thick, draw={rgb,255: red,0; green,207; blue,0}, opacity=0.4, tikzit draw={rgb,255: red,0; green,207; blue,0}]
\tikzstyle{X Web Overlay WIDE}=[-, line width=3pt, draw={rgb,255: red,208; green,0; blue,0}, opacity=0.4, tikzit draw={rgb,255: red,208; green,0; blue,0}]
\tikzstyle{Z Web Overlay WIDE}=[-, line width=3pt, draw={rgb,255: red,0; green,207; blue,0}, opacity=0.4, tikzit draw={rgb,255: red,0; green,207; blue,0}]
\tikzstyle{XZ Web}=[-, preaction={line width=2pt, draw={rgb,255: red,208; green,0; blue,0}, opacity=0.4, offset=-1pt}, preaction={line width=2pt, draw={rgb,255: red,0; green,207; blue,0}, opacity=0.4, offset=1pt}, tikzit draw={rgb,255: red,128; green,128; blue,0}]
\tikzstyle{ZX Web}=[-, preaction={line width=2pt, draw={rgb,255: red,208; green,0; blue,0}, opacity=0.4, offset=1pt}, preaction={line width=2pt, draw={rgb,255: red,0; green,207; blue,0}, opacity=0.4, offset=-1pt}, tikzit draw={rgb,255: red,191; green,128; blue,64}]
\tikzstyle{XZ Web no line}=[-, line width=0pt, opacity=0, preaction={line width=2pt, draw={rgb,255: red,208; green,0; blue,0}, opacity=0.4, offset=-1pt}, preaction={line width=2pt, draw={rgb,255: red,0; green,207; blue,0}, opacity=0.4, offset=1pt}, tikzit draw={rgb,255: red,128; green,128; blue,0}]
\tikzstyle{ZX Web no line}=[-, line width=0pt, opacity=0, preaction={line width=2pt, draw={rgb,255: red,208; green,0; blue,0}, opacity=0.4, offset=1pt}, preaction={line width=2pt, draw={rgb,255: red,0; green,207; blue,0}, opacity=0.4, offset=-1pt}, tikzit draw={rgb,255: red,191; green,128; blue,64}]
\tikzstyle{arrow}=[draw=black, ->]

\newcommand{\eg}{e.g.\xspace}

\newcommand{\ie}{i.e.\xspace}

\newcommand{\wrt}{w.r.t.\xspace}

\newcommand{\LHS}{{\normalfont\textsf{LHS}}\xspace}
\newcommand{\RHS}{{\normalfont\textsf{RHS}}\xspace}

\newcommand{\paulifaults}[1]{\overline{\mathcal{P}^{#1}}}

\newcommand{\weightfunc}{\text{wt}}

\input{preamble/quantum.tikzdefs}

\undef{\cite}

\undef{\eqref}
\undef{\autoref}

\usepackage[colorinlistoftodos, tickmarkheight=0.1cm]{todonotes}  

\title{Completeness for Fault Equivalence of Clifford ZX Diagrams}
\author{Maximilian Rüsch}
\affiliation{Quantinuum Ltd., Cambridge, UK}
\thanks{The work was carried out while MR was affiliated with University of Oxford, Oxford, UK.}

\author{Aleks Kissinger}
\affiliation{University of Oxford, Oxford, UK}
\author{Benjamin Rodatz}
\affiliation{University of Oxford, Oxford, UK}

\begin{document}
    \maketitle

    \begin{abstract}
	    Two circuits are considered to be equivalent under noise if the effect of faults on one circuit is no worse than the effect of faults on the other circuit.
We call this relationship \textit{fault equivalence}.
Fault equivalence offers a way to transform circuits while provably preserving their fault-tolerant properties, enabling a framework for fault-tolerant circuit synthesis and optimisation that is correct by construction.
The ZX calculus offers a diagrammatic way to represent and reason about quantum circuits and is a useful tool for manipulating circuits while preserving fault equivalence.
For this, the usual set of ZX rewrites has to be restricted to not only preserve the underlying linear map represented by the diagram, but also fault equivalence.

In this work, we provide a set of ZX rewrites that are sound and complete for fault equivalence of Clifford ZX diagrams.
This means that any equivalence that can be derived using the proposed rules is certain to be correct, and any correct equivalence can be derived using only these rules.
For this, we utilise diagrammatic constructions called fault gadgets to reason about arbitrary, possibly correlated Pauli faults in ZX diagrams.
Fault gadgets allow us to separate the diagram into a fault-free part, which captures the noise-free behaviour of a diagram, and a noisy part that enumerates the effects of all possible faults.
Using this, we provide a unique normal form for ZX diagrams under noise and show that any diagram can be brought into this normal form using our proposed rule set.
Leveraging the constructive nature of our proof of completeness, we implemented a software library to check whether two ZX diagrams with given noise models are fault-equivalent.
    \end{abstract}

    \tableofcontents


    \section{Introduction}
Scaling quantum computation involves, among many other components, optimising quantum programs and suppressing noise.
Independently, both fields are making progress: entire frameworks, libraries, and calculi are built to both systematically and heuristically simplify quantum programs while preserving their semantics~\autocite{nam2018automated, qiskit2024, kissinger2020Pyzx,vandewetering2020zxcalculusworkingquantumcomputer}, while fault-tolerant quantum computing research offers a versatile range of tools to reason about, detect, and correct faults resulting from noise~\autocite{gottesmanIntroductionQuantumError2009,gottesman1997stabilizercodesquantumerror,bacon2017SparseQuantumCodes,delfosse2023spacetimecodescliffordcircuits}.
More recently,~\autocite{rodatzFloquetifyingStabiliserCodes2024,rodatz2025faulttoleranceconstruction} have suggested a way to combine both fields, providing a framework to rewrite quantum circuits while preserving their behaviour under noise and therefore their fault tolerance properties.
This allows the provably fault-tolerant manipulation of quantum circuits, which can, for example, be used for circuit synthesis and optimisation. 
This paper advances this line of research by providing a complete set of such rewrite rules, meaning that, using these rules, any two circuits can be transformed into each other if and only if they are equivalent under noise. 

The core notion that underlies this framework is \textit{fault equivalence}.
Fault equivalence is based on the observation that two quantum circuits that implement the same linear map may have very different behaviour under noise, \eg a single fault on one circuit may be much more detrimental than a single fault on the other.
To formalise this, fault equivalence requires that for every undetectable fault that can happen on one circuit, there exists an equivalent fault on the other circuit. 
Additionally, this equivalent fault must have at most the same weight as the original fault, meaning it is at least as likely as the original one. 
This relationship must hold in both directions.
This guarantees that in a larger protocol, fault-equivalent circuits can be interchanged without affecting the protocol's behaviour under noise. 
Therefore, we can perform fault-tolerant circuit synthesis and optimisation as long as we preserve fault equivalence. 

To efficiently perform fault-tolerant circuit manipulation,~\autocite{rodatz2025faulttoleranceconstruction} proposes using the Clifford ZX calculus.
The ZX calculus formalises quantum programs as undirected diagrams obtained from a few simple generators~\autocite{coeckeInteractingQuantumObservables2011}.
It has proven successful as a framework for circuit synthesis and optimisation~\autocite{duncanGraphtheoreticSimplification2020, cowtanPhaseGadget2020, staudacher2024multicontrolled}. 
However, rewrites on the vanilla ZX calculus only preserve semantic equivalence, \ie whether diagrams represent the same linear map.
They do not guarantee the preservation of the stronger notion of fault equivalence. 
Therefore,~\autocite{rodatz2025faulttoleranceconstruction} restricts the set of allowed rewrites to \textit{fault-equivalent rewrites}, requiring this stronger notion.

A major question concerning any formal rewrite system is whether it is sound and complete. 
That is, the calculus should only allow deriving equivalences between expressions that actually hold (soundness), and it should allow deriving every equivalence that does hold (completeness).
While previous work in this domain has proven the soundness of the proposed rewrite rules for fault equivalence, it remained an open question of whether all valid fault equivalences could be derived from these rules. 
In this work, we propose a set of fault-equivalent rewrites and show that it is sound and complete. 

To facilitate the proof, we generalise the previous notion of a \textit{noise model} from~\autocite{rodatz2025faulttoleranceconstruction} that described noise through a set of independent faults which may occur during computation, called \textit{atomic faults}.
In our revised notion, atomic faults may be assigned integer weights describing their individual likelihood, inducing a weight assignment for all faults generated under their multiplication.
By generalising fault equivalence to these noise models, we increase the expressiveness of the framework, enabling circuit synthesis for operations with varying noise levels. 

Further, we utilise \textit{fault gadgets}~\autocite{rodatz2025faulttoleranceconstruction} to represent and manipulate atomic faults in circuits.
This allows us to separate a diagram into a fault-free part with non-trivial semantics and a faulty part with trivial semantics. 
Therefore, to rewrite one diagram into another, we have to show that each part can be rewritten individually.
For the fault-free part of the diagram, we leverage existing completeness results~\autocite{backensZXcalculusCompleteStabilizer2014,backens2017AsimplifiedStabiliserZXCalculus}.
For the faulty part of the diagram, we provide a novel procedure to rewrite any diagram into another fault-equivalent diagram, through a unique normal form that any diagram can be rewritten into.

As part of the rewriting procedure, we construct a framework to manipulate faults in spacetime.
We recover \textit{inconsequential faults}~\autocite{MagdalenadelaFuente2025XZYruby} that leave the diagram unchanged and form a group under multiplication.
This yields an equivalence of faults in spacetime: Two faults change the diagram in the same manner if and only if multiplying one by some inconsequential fault yields the other.
We further show that the group of inconsequential faults is exactly generated by local stabilisers inside the diagram, allowing us to reproduce the spacetime equivalence with fault gadgets.
Finally, we augment the framework with \textit{flip operators}, a class of faults such that every detectable or undetectable atomic fault may be expressed as a combination of flip operators.
These tools for manipulating faults open avenues to explore the behaviour of circuits under noise beyond fault equivalence checking.

Our proof of completeness is constructive, so an algorithmic procedure for checking fault equivalence can be directly derived.
An implementation of this algorithm to automatically check fault equivalence of ZX diagrams is provided as part of an open source Python package~\autocite{paritea2026code}\footnote{Available at \url{https://github.com/paritea/paritea}}.
    \section{Background}\label{sec:background}

\subsection{ZX Calculus}
The ZX calculus is a graphical language for representing and reasoning about linear maps between qubits~\autocite{coeckeInteractingQuantumObservables2011}.
Throughout this work, we focus on the Clifford ZX calculus, a fragment for reasoning about Clifford circuits.
For a more thorough introduction to the entire calculus, we refer to~\autocite{vandewetering2020zxcalculusworkingquantumcomputer}.

The basic building blocks of the ZX calculus are the \textit{spiders} that represent linear maps from $(\mathbb{C}^2)^{\otimes m}$ to $(\mathbb{C}^2)^{\otimes n}$:
\begin{definition}[Z spider, X spider]
 \begin{align*}
 \textit{Z-spider:} \qquad \input{02-preliminaries/z-spider.tikz} \quad
            &\coloneqq \quad \ket{0}^{\otimes n}\! \bra{0}^{\otimes m} + e^{i k\frac{\pi}{2}} \ket{1}^{\otimes n}\! \bra{1}^{\otimes m} \\[8pt]
 \textit{X-spider:} \qquad \input{02-preliminaries/x-spider.tikz} \quad
            &\coloneqq \quad \ket{+}^{\otimes n}\! \bra{+}^{\otimes m} + e^{i k\frac{\pi}{2}} \ket{-}^{\otimes n}\! \bra{-}^{\otimes m}
 \end{align*}
\end{definition}
By convention, when the phase of a spider is $0$, we omit the phase parameter.

We also include \textit{identities}, \textit{swaps}, \textit{cups}, and \textit{caps} in ZX diagrams:
\[
    \begin{tikzpicture}
	\begin{pgfonlayer}{nodelayer}
		\node [style=none] (0) at (-1, 0) {};
		\node [style=none] (1) at (1, 0) {};
	\end{pgfonlayer}
	\begin{pgfonlayer}{edgelayer}
		\draw (0.center) to (1.center);
	\end{pgfonlayer}
\end{tikzpicture}
 \ \ \coloneqq\ \  \sum_i |i \rangle\langle i|
    \qquad \input{02-preliminaries/swap.tikz} \ \ \coloneqq\ \ \sum_{ij} |ij \rangle\langle ji|
    \qquad \begin{tikzpicture}
	\begin{pgfonlayer}{nodelayer}
		\node [style=none] (0) at (0, 0.5) {};
		\node [style=none] (1) at (0, -0.5) {};
	\end{pgfonlayer}
	\begin{pgfonlayer}{edgelayer}
		\draw [bend left=270, looseness=2.00] (0.center) to (1.center);
	\end{pgfonlayer}
\end{tikzpicture}
 \ \ \coloneqq \ \ \sum_i |ii \rangle
    \qquad \begin{tikzpicture}
	\begin{pgfonlayer}{nodelayer}
		\node [style=none] (0) at (0, 0.5) {};
		\node [style=none] (1) at (0, -0.5) {};
	\end{pgfonlayer}
	\begin{pgfonlayer}{edgelayer}
		\draw [bend left=90, looseness=2.00] (0.center) to (1.center);
	\end{pgfonlayer}
\end{tikzpicture}
 \ \ \coloneqq \ \ \sum_i \langle ii|
\]

We can compose these elementary building blocks to form larger ZX diagrams:
\begin{definition}[Composition of ZX diagrams]
 The \textit{sequential composition} of two diagrams $D_1: (\mathbb{C}^2)^{\otimes m} \rightarrow (\mathbb{C}^2)^{\otimes k}$ and $D_2: (\mathbb{C}^2)^{\otimes k} \rightarrow (\mathbb{C}^2)^{\otimes n}$ is the diagram $D_2 \circ D_1$, which corresponds graphically to:
    \[ D_2 \circ D_1 \qquad\rightsquigarrow\qquad \input{02-preliminaries/sequential-composition.tikz}: (\mathbb{C}^2)^{\otimes m} \rightarrow (\mathbb{C}^2)^{\otimes n} \]
 The \textit{parallel composition} of two diagrams $D_1: (\mathbb{C}^2)^{\otimes m} \rightarrow (\mathbb{C}^2)^{\otimes n}$ and $D_2: (\mathbb{C}^2)^{\otimes k} \rightarrow (\mathbb{C}^2)^{\otimes l}$ is the diagram $D_1 \otimes D_2$, which corresponds graphically to:
    \[ D_1 \otimes D_2 \qquad\rightsquigarrow\qquad \input{02-preliminaries/parallel-composition.tikz}: (\mathbb{C}^2)^{\otimes (m+k)} \rightarrow (\mathbb{C}^2)^{\otimes (n+l)} \]
\end{definition}

As a shorthand, we further introduce a special notation for the Hadamard gate $H$, one of the basic Clifford unitaries:
\[ \input{02-preliminaries/hadamard-gate.tikz}\]

Diagrams without inputs and outputs represent complex numbers.
We refer to them as global scalars.
If we have multiple disconnected components in a diagram, the global scalar is the product of the scalars represented by each component.
The global scalar 0 is represented by a single spider, either green or red, with phase $\pi$ and no legs.

Along with diagrams, the ZX calculus also features a set of axioms called \say{rewrite rules}, that enable showing equivalences between diagrams.
An essential rule that we most often assume implicitly is \say{Only Connectivity Matters}, commonly abbreviated \say{OCM}.
This rule states that as long as we keep the connectivity of the spiders internally and the order of inputs and outputs the same, we may move spiders and bend legs arbitrarily without changing the underlying linear map:
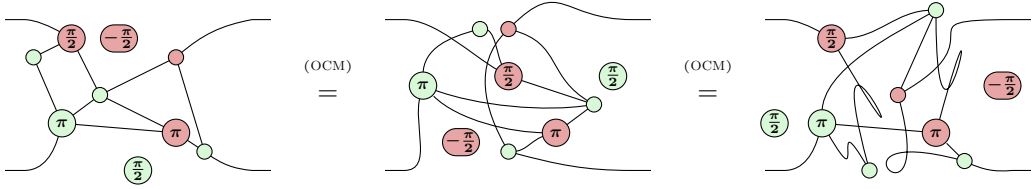
\[ \input{02-preliminaries/ocm.tikz} \]

Beyond OCM, we present six additional rewrite rules:
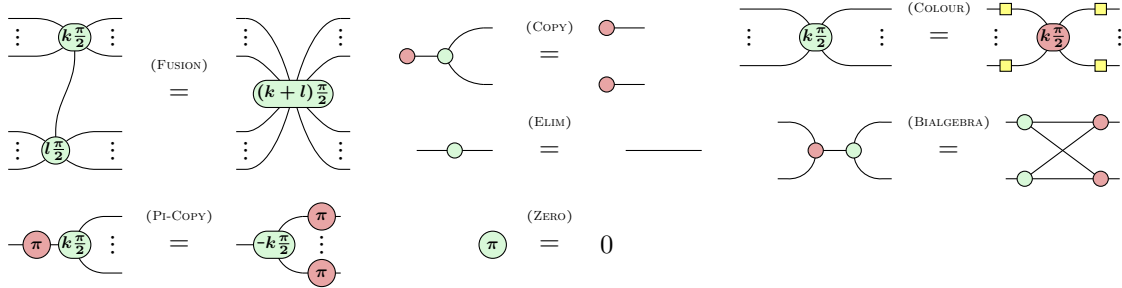
\[ \input{02-preliminaries/zx-axioms.tikz} \]
The axioms represent the most basic rewrites that transform one diagram into another semantically equivalent diagram.
All of these rules also hold in the colour inverse, meaning they are true if we make all green spiders red and all red spiders green.
We note that these rewrites are inherently bidirectional, \ie the transformation they apply has no particular direction in which it does not hold.

The presented rules only hold up to a non-zero global scalar.
This is sufficient for our purposes, as for comparing quantum circuits, we only care about the probability distribution of measurement outcomes, which is unaffected by non-zero global scalars.
Beyond the rules above, we consider any two diagrams with the same number of inputs and outputs and a global scalar of 0 to be equivalent, as they represent the same zero map.

A key characteristic of the ZX calculus is that it is sound and complete with respect to semantic equivalence:
\begin{theorem}[Completeness of the Clifford ZX calculus]
 We can rewrite two (non-zero) Clifford ZX diagrams $D_1, D_2$ into each other using the axioms above if and only if they represent the same linear map up to global scalar.
\end{theorem}
\begin{proof}
 See~\autocite{backensZXcalculusCompleteStabilizer2014,backens2017AsimplifiedStabiliserZXCalculus}.
\end{proof}
Finally, we note that through the Choi-Jamiołkowski isomorphism~\autocite{Nielsen_Chuang_2010} we can uniquely identify a linear map with a pure quantum state and vice versa.
In the ZX calculus, this isomorphism corresponds to bending around the input wires of a diagram into output wires, being careful to keep their relative order, or bending them back into input wires.
For the Clifford fragment, ZX diagrams uniquely correspond to stabiliser states following the stabiliser formalism.
Throughout this work, we will thus not distinguish between a Clifford ZX diagram representing a linear map and its corresponding stabiliser state.

\subsection{Noise on ZX Diagrams}
Following~\autocite{bacon2017SparseQuantumCodes, gottesman2022opportunitieschallengesfaulttolerantquantum, delfosse2023spacetimecodescliffordcircuits,rodatz2025faulttoleranceconstruction}, we consider Pauli faults in spacetime.
We define:
\begin{definition}[Pauli group]
 The \textit{Pauli group} $\mathcal{P}^1$ is defined as:
    \[ \mathcal{P}^1 = \{ \alpha P\ |\ \alpha\in\{ 1,-1,i,-i \},P \in \{ I,X,Y,Z \} \} \]
 where
    \[ I = \begin{pmatrix}1&0\\0&1\end{pmatrix} \quad X = \begin{pmatrix}0&1\\1&0\end{pmatrix} \quad Y = \begin{pmatrix}0&-i\\i&0\end{pmatrix} \quad Z = \begin{pmatrix}1&0\\0&-1\end{pmatrix}\,. \]
 The $n$-qubit Pauli group $\mathcal{P}^n$ is defined as the $n$-fold tensor product of $\mathcal{P}^1$.
 The Pauli group up to scalars, denoted $\paulifaults{n}$, is defined as the quotient of the Pauli group by its centre $\mathcal{P}^n / \{ \pm I, \pm i I \}$.
\end{definition}

We can represent faults in spacetime as elements of the Pauli group up to scalars:
\begin{definition}[Faults on ZX diagrams]
 Let $D$ be a ZX diagram with edges $E$.
 A \textit{fault} $F$ on $D$ is an element of $\paulifaults{|E|}$.
 The \textit{faulty diagram} $D^F$ is the diagram $D$ with the corresponding Pauli rotation applied to the respective edges as indicated by $F$.
\end{definition}

To indicate a specific fault $F$, we draw the effect of the fault on each edge in octagons on the edges, to clearly distinguish faults from Paulis that are part of the original diagram (see~\cref{fig:fault-examples-ZX} (a)).
Applying $F$ to $D$ then gives rise to the new faulty diagram $D^F$ (see~\cref{fig:fault-examples-ZX} (b)).
As we only care about diagrams up to non-zero global phase, we decompose the Pauli $Y$ rotation into a combination of Pauli $X$ and $Z$ rotations.
Since $Y = iXZ = -iZX$, up to global phase, this is well defined.

\begin{figure}
 \centering
 \begin{minipage}{0.32\textwidth}
 \centering
 \input{02-preliminaries/example-zx-fault.tikz}
 \caption*{(a)}
 \end{minipage}
 \begin{minipage}{0.32\textwidth}
 \centering
 \input{02-preliminaries/example-fault-diag.tikz}
 \caption*{(b)}
 \end{minipage}
 \begin{minipage}{0.32\textwidth}
 \centering
 \input{02-preliminaries/example-noise-model.tikz}
 \caption*{(c)}
 \end{minipage}
 \caption{An example ZX diagram from two qubits to three qubits.
 We have (a) A fault on the diagram as indicated by the Paulis in the octagons. (b) The faulty diagram obtained from applying the fault from (a). (c) An example annotation of an edge-flip noise model under which the fault from (a) has weight $8$.}
    \label{fig:fault-examples-ZX}
\end{figure}

We treat the ZX diagram $D$ as postselected on an expected measurement outcome.
From~\autocite{KissingerWetering2024Book} we know that if a diagram is zero, we will never see that particular measurement outcome.
Thus, if $D^F = 0$, the fault prevents us from ever seeing the expected measurement outcome of $D$.
In other words, we always see an unexpected outcome and therefore know that something went wrong --- the fault is detected.
We have:
\begin{definition}[Detectability of faults]
    Let $D$ be a diagram and $F \in \paulifaults{|E|}$ be a fault on $D$. We say $F$ is detectable if $D^F = 0$.
\end{definition}

Having defined faults in spacetime, we now provide a way to specify what faults can happen if something goes wrong:
\begin{definition}[Weighted noise model]
 Let $D$ be a ZX diagram with edges $E$.
 Let $\mathcal{F} \subseteq \paulifaults{|E|}$ be a set of faults that we treat as atomic faults.
 Then a \textit{weighted noise model} $\mathcal{N}$ is a map $\mathcal{N}: \mathcal{F} \to \mathbb{N}^+$ that assigns each atomic fault a weight, which we call \textit{atomic weight}.
 The set of all possible faults under the noise model $\mathcal{N}$ is the group $\langle \mathcal{F} \rangle$ generated by the atomic faults $\mathcal{F}$.
\end{definition}
In this definition, atomic faults make up the elementary faults that can occur.
This is analogous to the independent error mechanisms found in \texttt{stim}~\autocite{gidneyStimReadme2021}.
A fault is a combination of atomic faults.
The atomic weight of an atomic fault indicates how likely it is --- faults with a lower weight are more likely to occur than faults with a higher weight.
A weight can also be provided for combinations of atomic faults, through finding the combination with the lowest atomic weight sum:
\begin{definition}(Induced weight function)
    \label{def:induced-weight}
    A noise model $\mathcal{N}: \mathcal{F} \rightarrow \mathbb{N}^+$ induces a \textit{weight} function $\weightfunc_{\mathcal{N}}: \langle \mathcal{F} \rangle \rightarrow \mathbb{N}$, defined for a fault $F \in \langle \mathcal{F} \rangle$ as the minimum weight of any product of atomic faults that yields $F$:
    \[\weightfunc_{\mathcal{N}}(F) = \min_{\mathcal{F}_\text{sub} \subseteq \mathcal{F}} \sum_{F_i \in \mathcal{F}_\text{sub}} \mathcal{N}(F_i) \qquad \text{where} \quad \prod_{F_i \in \mathcal{F}_\text{sub}} F_i = F.\]
\end{definition}
Intuitively, the induced weight of a fault describes the \say{cost of the cheapest construction} of that fault.
So, we calculate the cost of a fault $F$ by considering all sets of atomic faults $\mathcal{F}_{sub} \subset \mathcal{F}$ that multiply together to create $F$.
Out of all those sets, the cost of $F$ is the cost of the cheapest set.
This means that the induced weight of an atomic fault may be smaller than its atomic weight.
For example, if the atomic faults $ZI$ and $IZ$ have weight $1$ and the atomic fault $ZZ$ has weight $3$, then the induced weight of $ZZ$ is $2$, as the cheapest way to construct $ZZ$ is by combining $ZI$ and $IZ$.
When the noise model $\mathcal{N}$ inducing the weight function $\weightfunc_{\mathcal{N}}$ is clear from context, we simply write $\weightfunc$.

With these two definitions, we generalise the noise model from~\autocite{rodatz2025faulttoleranceconstruction}, which is a special case where all atomic faults have weight $1$.

Throughout this work, we will primarily focus on one class of noise models:
\begin{definition}
 A noise model $\mathcal{N}: \mathcal{F} \to \mathbb{N}^+$ is an \textit{edge-flip noise model} if all atomic faults $\mathcal{F}$ act non-trivially on at most one edge.
\end{definition}

Edge-flip noise is a helpful restriction as it substantially simplifies reasoning about diagrams under noise.
Simultaneously, in~\cref{prop:weighted-edge-flip-universality}, we will argue that edge-flip noise is a sufficiently general noise model, as for any diagram under any noise model, we can find an equivalent diagram under edge-flip noise.

Furthermore, edge-flip noise allows us to depict the ZX diagram and the noise model in a single picture by annotating each edge with a triplet $(w_X, w_Y, w_Z)$ indicating the respective weight of the $X$, $Y$, and $Z$ edge-flip affecting that edge.
If a fault is not part of the atomic noise model, we indicate its weight as $-$.
As we will later primarily reason about $Z$ edge flips, if on a specific edge neither $X$ nor $Y$ edge flips are part of the noise model, we will annotate the edge with a single number $w$ as a shorthand for $(-, -, w)$.
If an edge allows no faults, we will omit the annotation entirely and draw the edge in purple and bold to show that it is idealised as fault-free.
\cref{fig:fault-examples-ZX} (c) shows an example noise model.
In practice, noise models will often be more uniform.

\subsection{Fault Equivalence}\label{subsec:fault-equivalence}
In a noise-free setting, we say that two diagrams are equivalent if they represent the same linear map.
However, in the presence of faults, we need a stronger notion of equivalence that also takes the behaviour of the circuits under noise into account.
For this purpose, the notion of \textit{fault equivalence} was first introduced by~\autocite{rodatzFloquetifyingStabiliserCodes2024} and later refined in~\autocite{rodatz2025faulttoleranceconstruction}.
For an in-depth motivation of fault equivalence and an overview of its applications, we refer to~\autocite{rodatz2025faulttoleranceconstruction,rodatzFloquetifyingStabiliserCodes2024}.

In this paper, we relax the notion of fault equivalence to an asymmetric relation called fault boundedness.
If two diagrams are mutually fault-bounded by each other, they are fault-equivalent.
\begin{definition}[$w$-fault boundedness, fault boundedness]
 Let $D_1,D_2$ be ZX diagrams with respective noise models $\mathcal{N}_1: \mathcal{F}_1 \to \mathbb{N}^+, \mathcal{N}_2: \mathcal{F}_2 \to \mathbb{N}^+$.
 We say that $D_1$ under $\mathcal{N}_1$ is \textit{$w$-fault-bounded} by $D_2$ under $\mathcal{N}_2$ for some $w \in \mathbb{N}^+$ if for any fault $F_1$ on $D_1$ with $\weightfunc_{\mathcal{N}_1}(F_1) < w$, we have either:
 \begin{itemize}
        \item $F_1$ is detectable, \textbf{or}
        \item there exists a fault $F_2$ on $D_2$ with $\weightfunc_{\mathcal{N}_2}(F_2) \leq \weightfunc_{\mathcal{N}_1}(F_1)$ such that $D_1^{F_1} = D_2^{F_2}$.
 \end{itemize}
 We write $D_1 \wFaultBnd{w} D_2$, leaving the noise models implicit when they are clear from context.
    $D_1$ under $\mathcal{N}_1$ is \textit{fault-bounded} by $D_2$ under $\mathcal{N}_2$, written $D_1 \FaultBnd D_2$, if $D_1 \wFaultBnd{w} D_2$ for all $w \in \mathbb{N}^+$.
\end{definition}

Using this, we can now define:
\begin{definition}[$w$-fault equivalence, fault equivalence]
 Let $D_1,D_2$ be ZX diagrams with respective noise models $\mathcal{N}_1, \mathcal{N}_2$.
 We say that $D_1$ under $\mathcal{N}_1$ is \textit{$w$-fault-equivalent} to $D_2$ under $\mathcal{N}_2$ for some $w \in \mathbb{N}^+$, written $D_1 \wFaultEq{w} D_2$, if $D_1 \wFaultBnd{w} D_2$ and $D_2 \wFaultBnd{w} D_1$.
    $D_1$ under $\mathcal{N}_1$ is \textit{fault-equivalent} to $D_2$ under $\mathcal{N}_2$, written $D_1 \FaultEq D_2$, if $D_1 \wFaultEq{w} D_2$ for all $w \in \mathbb{N}^+$.
\end{definition}

\begin{example}{Four-legged spider}{}
   One example of two fault-equivalent diagrams is the four-legged spider unfusion under edge-flip noise~\autocite{rodatzFloquetifyingStabiliserCodes2024}:
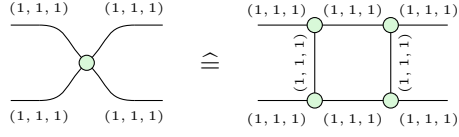
   \[\input{running-example/fault-equivalence.tikz}\]
   The unfused diagram introduces one detecting region (see~\cref{subsec:pauli-webs}).
   Therefore, a fault with an odd number of $X$ and $Y$ edge-flips on the internal edges is detectable.
   To check fault equivalence, we can iterate over all undetectable faults and check that they have corresponding faults on the other side.
   For more details, see~\autocite{rodatzFloquetifyingStabiliserCodes2024}.
   In this work, we will use this fault equivalence as a running example to show how these two diagrams can be deformed into one another using our proposed rule set.
\end{example}

This paper is mainly devoted to showing the completeness of a set of rewrite rules for fault equivalence.
For this, an essential property of $w$-fault boundedness, and therefore also of fault boundedness and fault equivalence, is that it is transitive:
\begin{proposition}
    \label{prop:w-fault-boundedness-transitive}
    $w$-fault boundedness is transitive, \ie for ZX diagrams $D_1,D_2,D_3$ under respective noise models $\mathcal{N}_1: \mathcal{F}_1 \rightarrow \mathbb{N}^+,\mathcal{N}_2: \mathcal{F}_2 \rightarrow \mathbb{N}^+,\mathcal{N}_3: \mathcal{F}_3 \rightarrow \mathbb{N}^+$ it holds that
    \[ D_1 \wFaultBnd{w_1} D_2 \quad\text{and}\quad D_2 \wFaultBnd{w_2} D_3 \qquad\Longrightarrow\qquad D_1 \wFaultBnd{\min(w_1,w_2)} D_3\,. \]
\end{proposition}
\begin{proof}
 We start with an undetectable fault $F_1 \in \mathcal{F}_1$ with $\weightfunc(F) < \min(w_1,w_2)$.
 As we assumed $D_1 \wFaultBnd{w_1} D_2$, there is some $F_2 \in \mathcal{F}_2$ with $D_1^{F_1} = D_2^{F_2}$ and $\weightfunc(F_2) \leq \weightfunc(F_1)$.
 As $D_2 \wFaultBnd{w_2} D_3$, we can similarly obtain some $F_3 \in \mathcal{F}_3$ with $D_2^{F_2} = D_3^{F_3}$ and $\weightfunc(F_3) \leq \weightfunc(F_2)$.

 But then via transitivity, $D_1^{F_1} = D_2^{F_2} = D_3^{F_3}$ and $\weightfunc(F_3) \leq \weightfunc(F_2) \leq \weightfunc(F_1)$, completing the claim.
\end{proof}

Additionally, we can observe that $w$-fault boundedness is preserved under sequential and parallel composition:
\begin{proposition}
 Fault boundedness is compositional, \ie for ZX diagrams $D_1,D_1',D_2,D_2'$ with respective noise models $\mathcal{N}_1,\mathcal{N}_1',\mathcal{N}_2,\mathcal{N}_2'$ it holds that
 \begin{align*}
 D_1 \wFaultBnd{w_1} D_2 \quad&\text{and}\quad D_1' \wFaultBnd{w_2} D_2' && \qquad&\Longrightarrow\qquad && D_1' \circ D_1 &\wFaultBnd{\min(w_1,w_2)} D_2' \circ D_2\,, \\
 D_1 \wFaultBnd{w_1} D_2 \quad&\text{and}\quad D_1' \wFaultBnd{w_2} D_2' && \qquad&\Longrightarrow\qquad && D_1 \otimes D_2 &\wFaultBnd{\min(w_1,w_2)} D_1' \otimes D_2'\,,
 \end{align*}
 where the composed diagrams have noise models uniquely composed of $\mathcal{F}_1,\mathcal{F}_1',\mathcal{F}_2,\mathcal{F}_2'$ under atomic fault set union (while padding faults with $I$ for edges in the other diagrams) and addition of atomic weights.
\end{proposition}
\begin{proof}
 We prove the claim for sequential composition.
 The proof for parallel composition is analogous.
 Let $\mathcal{N}_1^\circ,\mathcal{N}_2^\circ$ be the noise models obtained for $D_1' \circ D_1$ and $D_2' \circ D_2$ respectively.
 These noise models are well defined, because the original noise models they are respectively sourced from are operating on disjoint sets of edges, and thus in disjoint spaces of $\paulifaults{|\cdot|}$.

 We start with considering an undetectable fault $F_1^\circ \in \mathcal{F}_1^\circ$ of weight less than $\min(w_1,w_2)$.
 $F_1^\circ$ is composed of $F_1, F_1'$ drawn from $\mathcal{F}_1,\mathcal{F}_1'$ respectively, and since $\mathcal{F}_1^\circ$ was formed under the addition of atomic weights, it must be that
    \[ \weightfunc(F_1^\circ) = \weightfunc(F_1) + \weightfunc(F_1') < \min(w_1,w_2)\,. \]
 Then, via the precondition of $w_1$- and $w_2$-fault boundedness, we can find semantically equivalent $F_2, F_2'$ drawn from $\mathcal{F}_2,\mathcal{F}_2'$ that have weights of at most those of $F_1, F_1'$.
 Their composite $F_2^\circ = F_2 F_2'$ is then a fault on $D_2' \circ D_2$ that can be drawn from $\mathcal{F}_2^\circ$ and it holds that
 \begin{gather*}
 (D_2' \circ D_2)^{F_2^\circ} \quad=\quad {D_2'}^{F_2'} \circ D_2^{F_2} \quad=\quad {D_1'}^{F_1'} \circ D_1^{F_1} \quad=\quad (D_1' \circ D_1)^{F_1^\circ}\,, \\
 \weightfunc(F_2^\circ) \quad=\quad \weightfunc(F_2) + \weightfunc(F_2') \quad\leq\quad \weightfunc(F_1) + \weightfunc(F_1') \quad=\quad \weightfunc(F_1^\circ)\,.
 \end{gather*}
 Thus, there is an equivalent $F_2^\circ$ for $F_1^\circ$ with lower or equal weight.
\end{proof}

We close this discussion on fault equivalence by considering the case that a noise model $\mathcal{N}$ contains no faults for a diagram $D$.
In this case, we call $D$ \textit{fully idealised}.
Intuitively, $D$ being fault-bounded by a second diagram $D'$ is equivalent to stating that the noise model $\mathcal{N}'$ of $D'$ describes \say{worse} behaviour under noise for $D'$ than $\mathcal{N}$ describes for $D$.
Then, the following result is unsurprising:
\begin{proposition}
    \label{prop:fully-idealised-equivalence}
    Let $D$ be a diagram and $\mathcal{N}$ a noise model such that $D$ is fully idealised.
    Then for any diagram $D'$ with $D = D'$ and any accompanying noise model $\mathcal{N}'$, it holds that $D \FaultBnd D'$.

    Furthermore, if $\mathcal{N}'$ is such that $D'$ is fully idealised as well, the symmetric $D \FaultEq D'$ holds.
\end{proposition}
\begin{proof}
   As $\mathcal{N}$ has no atomic faults, trivially, there exists a corresponding fault in $\mathcal{N}'$ for all faults in $\mathcal{N}$.

   The second statement follows from applying the first statement in both directions, giving fault equivalence of the two diagrams.
\end{proof}

\subsection{Pauli Webs}\label{subsec:pauli-webs}

Pauli webs are decorations for ZX diagrams that highlight some of the edges of the diagram in red and/or green.
They are useful tools for tracking stabilisers and detectors, and as we will formalise in \cref{subsec:inconsequential-faults}, they can be used to characterise and classify faults.

For this, we first observe that spiders are stabilised by Paulis, as they represent Clifford maps.
We can define:
\begin{definition}[Spider stabilisers]
    \label{def:stabilisers-of-spider}
 Let a ZX diagram $D$ consist of one spider $s$ with $k$ legs.
 Then a Pauli $P \in \paulifaults{k}$ is a \textit{spider stabiliser} of $s$ if $D^P = D$.
 In a larger diagram, a Pauli $P$ is a spider stabiliser of a spider $s$
 if it is a spider stabiliser of the subdiagram formed by $s$ alone and acts trivially on all other edges.
\end{definition}

\begin{figure}
 \centering
 \begin{minipage}{0.32\textwidth}
 \centering
 \input{02-preliminaries/spider-stabiliser-1.tikz}
 \caption*{(a)}
 \end{minipage}
 \begin{minipage}{0.32\textwidth}
 \centering
 \input{02-preliminaries/spider-stabiliser-2.tikz}
 \caption*{(b)}
 \end{minipage}
 \begin{minipage}{0.32\textwidth}
 \centering
 \input{02-preliminaries/pauli-web-example.tikz}
 \caption*{(c)}
 \end{minipage}
 \caption{(a) A spider stabiliser of a single spider. (b) The same spider stabiliser in a larger diagram. (c) A Pauli web on a diagram.}
    \label{fig:spider-stabiliser}
\end{figure}

For example,~\cref{fig:spider-stabiliser}~(a) shows a spider stabiliser of a single spider, while~\cref{fig:spider-stabiliser}~(b) shows the same spider stabiliser in a larger diagram.

Based on this, we define:
\begin{definition}[Pauli web]
 Let $D$ be a ZX diagram with edges $E$.
 An element $P \in \paulifaults{|E|}$ is a \textit{Pauli web} if for each spider $s$ of $D$, $P$ restricted to the neighbourhood of $s$ is a spider stabiliser of $s$.
\end{definition}
We can view a Pauli web $P = \bigotimes P_i$ as highlighting the edge $i \in 1,\dots,|E|$ of a diagram as red if $P_i = X$, green if $P_i = Z$, and red and green if $P_i = Y$.

It is important to note that not every spider stabiliser is a valid Pauli web in the context of the entire diagram.
Instead, we require that locally, restricted to the neighbourhood of each spider, the Pauli web acts like a spider stabiliser.
For example, \cref{fig:spider-stabiliser}~(c) shows a Pauli web that includes the spider stabiliser of \cref{fig:spider-stabiliser}~(b) and additionally has a non-trivial action on other edges.
In contrast, \cref{fig:spider-stabiliser}~(b) would not be a valid Pauli web as it does not act like a spider stabiliser on some of the spiders, such as the green spider at the top.

Using this, we can recover the usual definition of Pauli webs as in~\autocite{bombinUnifyingFlavorsFault2024,rodatzFloquetifyingStabiliserCodes2024}:
\begin{proposition}[Pauli Web]
 Let $D$ be a ZX diagram with edges $E$.
    $P \in \paulifaults{|E|}$ is a Pauli web if and only if:
 \begin{itemize}
        \item a spider with phase in $\{0,\pi\}$ satisfies:
 \begin{itemize}
                \item $P$ highlighting an \textit{even} number of its legs in its own colour \textit{and all or none} in the opposite colour
 \end{itemize}
        \item a spider with phase in $\{-\frac{\pi}{2},\frac{\pi}{2}\}$ satisfies:
 \begin{itemize}
                \item $P$ highlighting an \textit{even} number of its legs in its own colour \textit{and none} in the opposite colour, \textbf{or}
                \item $P$ highlighting an \textit{odd} number of its legs in its own colour \textit{and all} in the opposite colour.
 \end{itemize}
 \end{itemize}
\end{proposition}
\begin{proof}
 The constraints on the neighbourhood of spiders guarantee exactly that the corresponding restrictions of $P$ to the legs of each spider are spider stabilisers.
 Therefore, $P$ is a Pauli web.
 Similarly, all Pauli webs must satisfy these constraints.
\end{proof}

As Pauli webs, locally restricted to a spider, correspond to spider stabilisers, we can `fire' spiders according to a Pauli web.
As we will care about the global phase change when firing a Pauli web, we have to be a little bit more careful about Pauli $Y$.
Thus, we introduce a Pauli $Y$ box, which is defined as $Y = iZX$.
We have:
\begin{definition}[Firing a diagram according to a Pauli web]
 Let $D$ be a ZX diagram and $P$ be a Pauli web on $D$.
 Then we say we \textit{fire all spiders in $D$ according to $P$} when we:
   \begin{enumerate}
      \item Surround all spiders with Pauli $Z$'s on edges highlighted in green.
      \item Surround all spiders with Pauli $X$'s on edges highlighted in red.
      \item On all the internal edges, merge the newly fired Paulis.
      \item Recombining Pauli $X$ and Pauli $Z$ spiders on the boundary into Pauli $Y$.
   \end{enumerate}
\end{definition}
As all the internal edges are connected to two spiders, any highlighted internal edge will have two Paulis of the same type introduced on it, which are removed in the second step.
Thus, after merging, the only Paulis left are $Z$ and $X$ Paulis on the boundary edges.
If we have a $Z$ and an $X$ left on the same edge, we recombine them into a $Y$.
For example, we have:
\[\input{02-preliminaries/firing-Pauli-web.tikz}\]

Firing a spider according to the Pauli web changes the underlying linear map at most by a global phase.
We can see this by considering that a Pauli web locally corresponds to stabilisers of the spiders, so the first two steps may at most change the global phase.
Similarly, merging Paulis does not change the semantics of the diagram, as Pauli $X$ and Pauli $Z$ are self-adjoint and thus square to the identity.
Finally, merging the Paulis on the boundary absorbs global phases of $i$ into the $Y$ box, leaving at most a phase of $1$ or $-1$.

After firing a diagram according to a Pauli web, having cancelled out all the Paulis on the internal edges, all we are left with are Paulis on the boundary edges.
We just reasoned that firing the web does not change the diagram's semantics, so we can freely introduce or remove these Paulis on the boundaries --- they stabilise the diagram:
\begin{definition}[Stabilising Pauli webs]
 A Pauli web on some diagram is called a stabilising Pauli web if it acts non-trivially on at least one boundary edge.
\end{definition}

Stabilising Pauli webs of a diagram $D$ are in direct correspondence with the stabilisers of the underlying linear map of $D$~\autocite{bombinUnifyingFlavorsFault2024, rodatzFloquetifyingStabiliserCodes2024}; the Paulis remaining on the boundary edges after firing a diagram according to a Pauli web exactly correspond to its stabilisers.

Using the observation that firing the spiders of a diagram according to a web at most changes the global phase up to a factor of $-1$, we define:
\begin{definition}[Sign of a Pauli web]
 Let $D$ be a ZX diagram and $P$ be a Pauli web on $D$.
 Then the sign of $P$ is the change in global phase we get when firing all the spiders of $D$ according to $P$.
\end{definition}

For stabilising Pauli webs, the sign of the Pauli web indicates whether the whole diagram is a $+1$ or $-1$ eigenstate of the corresponding stabiliser.
An essential nuance here is that the sign of a Pauli web is independent of the actual global phase of the diagram.
Two ZX diagrams that are equivalent up to a non-zero global phase must have the same sign for all stabilising Pauli webs.
On the flip side, if two ZX diagrams have the same stabilising Pauli webs but different signs for those webs, they are not equivalent, even up to global phase.
For example $\ket{0}$ and $\ket{1}$ both stabilise $Z$, but $\ket{0}$ is a $+1$ eigenstate of $Z$ while $\ket{1}$ is a $-1$ eigenstate of $Z$, so they are not equivalent up to global phase.
While a Pauli web is independent of the actual global phase, to establish its sign, we have to consider the \emph{change} in global phase as we fire the web.
Therefore, reasoning about the sign of a Pauli web is, along with \cref{prop:zero-diagram-disconnect}, one of the few instances where we must consider the global phase of a diagram, while everywhere else in the paper, we can ignore it.

There is a second type of Pauli webs for which this sign is even more essential:
\begin{definition}[Detecting region]
 A non-trivial Pauli web on some diagram is called a detecting region if it acts trivially on all boundary edges.
\end{definition}

For detecting regions, we can show:
\begin{proposition}
   \label{prop:detecting-regions}
   Let $D$ be a ZX diagram with some detecting region $R$ that has a negative sign.
   Then $D$ is equivalent to zero.
\end{proposition}
\begin{proof}
   To show that $D = 0$, we will instead show that $D$ is exactly equal to $-D$, not just up to global phase. 
   This can only be the case when $D = 0$.
   For this, we can fire all the spiders of $D$ according to $R$.
   This changes the underlying linear map by a global phase of $-1$.
   $R$ is a detecting region, all highlighted edges are internal and therefore now have two Paulis of the same type.
   These cancel out, meaning we are back to the original diagram.
   But then $D = -D$, so $D = 0$.
\end{proof}
But then, we can make a second observation:
\begin{proposition}
   \label{prop:faults-flip-pauli-webs}
   Let $D$ be a ZX diagram with edges $E$, $F \in \paulifaults{|E|}$ be a fault on $D$.
   Then for all Pauli webs on $P \in \paulifaults{|E|}$ there exists a corresponding Pauli web $P_F$ on $D^F$ that has the same action on the boundary.
   The sign of $P$ is the same as the sign of $P_F$ if and only if $F$ commutes with $P$.
\end{proposition}
\begin{proof}
 Let $P$ be a Pauli web on $D$.
 Then we can construct a corresponding Pauli web on $D^F$.
 We can observe that $D^F$ has the same edges as $D$, just that some of the edges of $D$ are split into two by the faults.
 Therefore, we can construct a Pauli web $P_F$ that acts the same as $P$ on all edges, even the split ones.
 On all the original spiders of $D$, $P_F$ clearly satisfies the Pauli web constraints of locally corresponding to a spider stabiliser.
 For the new spiders introduced by the fault, it also satisfies these constraints as the faults are two-legged $\pi$-phased spiders and $P_F$ has the same action on both sides of that spider.
 Therefore, $P_F$ is a corresponding Pauli web to $P$ and clearly has the same action on the boundary as $P$.

 The sign of $P_F$ is obtained by firing all the spiders of $D^F$ according to $P_F$.
 For all the spiders that are already present in $D$, this firing has the same effect as firing the spiders of $D$ according to $P$.
 However, additionally, we have to fire the spiders newly introduced by $F$.
 For each edge, if $F$ commutes with $P$ on that edge, firing that spider does not change the global phase.
 However, if $F$ anticommutes with $P$, firing the fault on that edge according to $P_F$ flips the global phase.
 Therefore, if $F$ commutes with $P$, the global phase will be changed evenly many times, and therefore the signs of $P$ and $P_F$ are the same.
 However, if they anticommute, the sign of $P_F$ will be the opposite of the sign of $P$.
\end{proof}

For stabilising Pauli webs, anticommuting faults flip the sign of the corresponding stabilisers, thereby changing whether $D^F$ lives in the same eigenspace of the stabiliser as $D$.
For detecting regions, the commutative relationship between the fault and the regions tells us whether the fault is detectable.
If we assume $D$ to be non-zero, by \cref{prop:detecting-regions}, all detecting regions in $D$ must have a positive sign.
But then, any fault that anti-commutes with a detecting region will make the sign of that region in $D^F$ negative.
Therefore, by \cref{prop:detecting-regions}, $D^F$ must be zero, meaning $F$ is detected.
Thus, detecting regions are a crucial tool for reasoning about detectable faults.

    \section{The Rules and Strategy}\label{sec:the-rules-and-strategy}
The main focus of this paper is to provide a diagrammatic rewrite system for the ZX calculus that is sound and complete for fault equivalence, of which the latter is often significantly harder to prove.
Historically, showing completeness of a ZX rewrite system usually follows one of two general approaches: Either one reduces one calculus to another which is already known to be complete, \eg~\autocite{backens2017AsimplifiedStabiliserZXCalculus,wangKang2017ZXCalculusUniversalComplation,vilmart2019NearMinimalAxiomatisation}, or one provides a unique normal form for diagrams from the calculus reachable through its rewrite system, \eg~\autocite{backens2014ZXCalculusCompleteCliffordT,KissingerWetering2024Book}.
As fault equivalence is a novel idea introduced only recently, there are no rewrite systems known to be complete.
Therefore, we will show the completeness of the proposed rule set via a unique normal form.

\subsection{The Normal Form and Proof Sketch}
Completeness proofs via a unique normal form consist of showing that any diagram can be rewritten into a normal form that is unique for each equivalence class of diagrams, following some notion of equivalence.
Then, any two equivalent diagrams must have the same normal form, reachable through two potentially different sequences of rewrites.
All ZX rewrites are undirected, so we can invert one of those sequences to obtain a sequence that deforms one diagram into the other diagram, effectively showing the equivalence diagrammatically.

We now begin constructing our unique normal form.
For two diagrams to be fault-equivalent, they have to satisfy two conditions: They have to implement the same linear map, and the sets of potentially generated fault effects along with their minimum weight must be equivalent.
The first condition is exactly the \textit{semantic equivalence} between ZX diagrams, which is a core feature of the calculus, and many (normal) forms have surfaced for it~\autocite{KissingerWetering2024Book}.
If we are able to find a diagrammatic representation to address both fault equivalence conditions individually, we are able to reuse these forms for the semantic part of the diagram.

For this, we make use of \textit{fault gadgets}, a tool that can find a diagrammatic representation for arbitrary faults on ZX diagrams~\autocite{rodatz2025faulttoleranceconstruction}.
Thus, for a diagram $D$ and an associated noise model $\mathcal{N}$, we can find another diagram $D_\mathcal{N}$ describing both.
However, going beyond the work of~\autocite{rodatz2025faulttoleranceconstruction} in \cref{sec:separating-semantics-and-noise}, we show how fault gadgets can be fault-equivalently manipulated and moved.
This allows us to separate the fault gadgets from the remainder of the diagram, such that we obtain a fully idealised diagram with non-trivial semantics followed by a noisy diagram consisting of fault gadgets with trivial semantics:
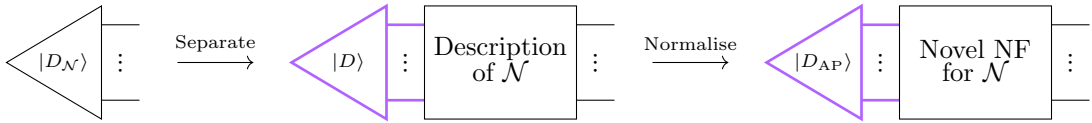
\[ \begin{tikzpicture}
	\begin{pgfonlayer}{nodelayer}
		\node [style=none] (28) at (-1, 1.5) {};
		\node [style=none] (29) at (-1, -1.5) {};
		\node [style=none] (30) at (-3.5, 0) {};
		\node [style=sLabel] (31) at (-2, 0) {$\ket{D}$};
		\node [style=none] (32) at (0, 1) {};
		\node [style=none] (33) at (0, -1) {};
		\node [style=none] (34) at (-1, 1) {};
		\node [style=none] (35) at (-1, -1) {};
		\node [style=none, rotate=90] (36) at (-0.5, 0) {...};
		\node [style=none] (37) at (4, 1) {};
		\node [style=none] (38) at (4, -1) {};
		\node [style=none] (39) at (0, 1.5) {};
		\node [style=none] (40) at (4, 1.5) {};
		\node [style=none] (41) at (4, -1.5) {};
		\node [style=none] (42) at (0, -1.5) {};
		\node [style=none] (43) at (5, 1) {};
		\node [style=none] (44) at (5, -1) {};
		\node [style=none] (45) at (2, 0.375) {Description};
		\node [style=none, rotate=90] (46) at (4.5, 0) {...};
		\node [style=none] (48) at (2, -0.25) {of $\mathcal{N}$};
		\node [style=none] (49) at (11.5, 1.5) {};
		\node [style=none] (50) at (11.5, -1.5) {};
		\node [style=none] (51) at (9, 0) {};
		\node [style=sLabel] (52) at (10.5, 0) {$\ket{D_\text{rAP}}$};
		\node [style=none] (53) at (12.5, 1) {};
		\node [style=none] (54) at (12.5, -1) {};
		\node [style=none] (55) at (11.5, 1) {};
		\node [style=none] (56) at (11.5, -1) {};
		\node [style=none, rotate=90] (57) at (12, 0) {...};
		\node [style=none] (58) at (16.5, 1) {};
		\node [style=none] (59) at (16.5, -1) {};
		\node [style=none] (60) at (12.5, 1.5) {};
		\node [style=none] (61) at (16.5, 1.5) {};
		\node [style=none] (62) at (16.5, -1.5) {};
		\node [style=none] (63) at (12.5, -1.5) {};
		\node [style=none] (64) at (17.5, 1) {};
		\node [style=none] (65) at (17.5, -1) {};
		\node [style=none] (66) at (14.5, 0.375) {Novel NF};
		\node [style=none, rotate=90] (67) at (17, 0) {...};
		\node [style=none] (68) at (14.5, -0.25) {for $\mathcal{N}$};
		\node [style=none] (69) at (6, 0) {};
		\node [style=none] (70) at (8, 0) {};
		\node [style=sLabel] (71) at (7, 0.5) {Normalise};
		\node [style=none] (72) at (-6.5, 0) {};
		\node [style=none] (73) at (-4.5, 0) {};
		\node [style=sLabel] (74) at (-5.5, 0.5) {Separate};
		\node [style=none] (95) at (-7.5, 1) {};
		\node [style=none] (96) at (-7.5, -1) {};
		\node [style=none] (99) at (-8.5, 1.5) {};
		\node [style=none] (100) at (-8.5, -1.5) {};
		\node [style=none] (101) at (-11, 0) {};
		\node [style=sLabel] (102) at (-9.5, 0) {$\ket{D_\mathcal{N}}$};
		\node [style=none] (103) at (-8.5, 1) {};
		\node [style=none] (104) at (-8.5, -1) {};
		\node [style=none, rotate=90] (105) at (-8, 0) {...};
	\end{pgfonlayer}
	\begin{pgfonlayer}{edgelayer}
		\draw [style=fault-free] (28.center)
			 to (29.center)
			 to (30.center)
			 to cycle;
		\draw [style=fault-free] (34.center) to (32.center);
		\draw [style=fault-free] (35.center) to (33.center);
		\draw (41.center)
			 to (40.center)
			 to (39.center)
			 to (42.center)
			 to cycle;
		\draw (38.center) to (44.center);
		\draw (37.center) to (43.center);
		\draw [style=fault-free] (50.center)
			 to (51.center)
			 to (49.center)
			 to cycle;
		\draw [style=fault-free] (55.center) to (53.center);
		\draw [style=fault-free] (56.center) to (54.center);
		\draw (63.center)
			 to (62.center)
			 to (61.center)
			 to (60.center)
			 to cycle;
		\draw (59.center) to (65.center);
		\draw (58.center) to (64.center);
		\draw [style=arrow] (69.center) to (70.center);
		\draw [style=arrow] (72.center) to (73.center);
		\draw (99.center)
			 to (100.center)
			 to (101.center)
			 to cycle;
		\draw (103.center) to (95.center);
		\draw (104.center) to (96.center);
	\end{pgfonlayer}
\end{tikzpicture}
 \]
Our composite normal form consists of the unique \textit{reduced AP form} for the idealised part of the diagram from~\autocite{KissingerWetering2024Book}, followed by a novel normal form on the noisy part of the diagram as specified in \cref{sec:normal-form}.

It then remains to show that the normal form for the diagrammatic description of noise is reachable through our proposed set of rewrites, and indeed unique for each fault equivalence class of diagrams.
We continue in this section with introducing our rewrites, and show reachability of the normal form in \cref{sec:normal-form}.
Then, if both parts of the composite normal form are unique, they yield an overall unique normal form and thus completeness of the rewrite system.

\subsection{The Rules}
\begin{figure}
    \input{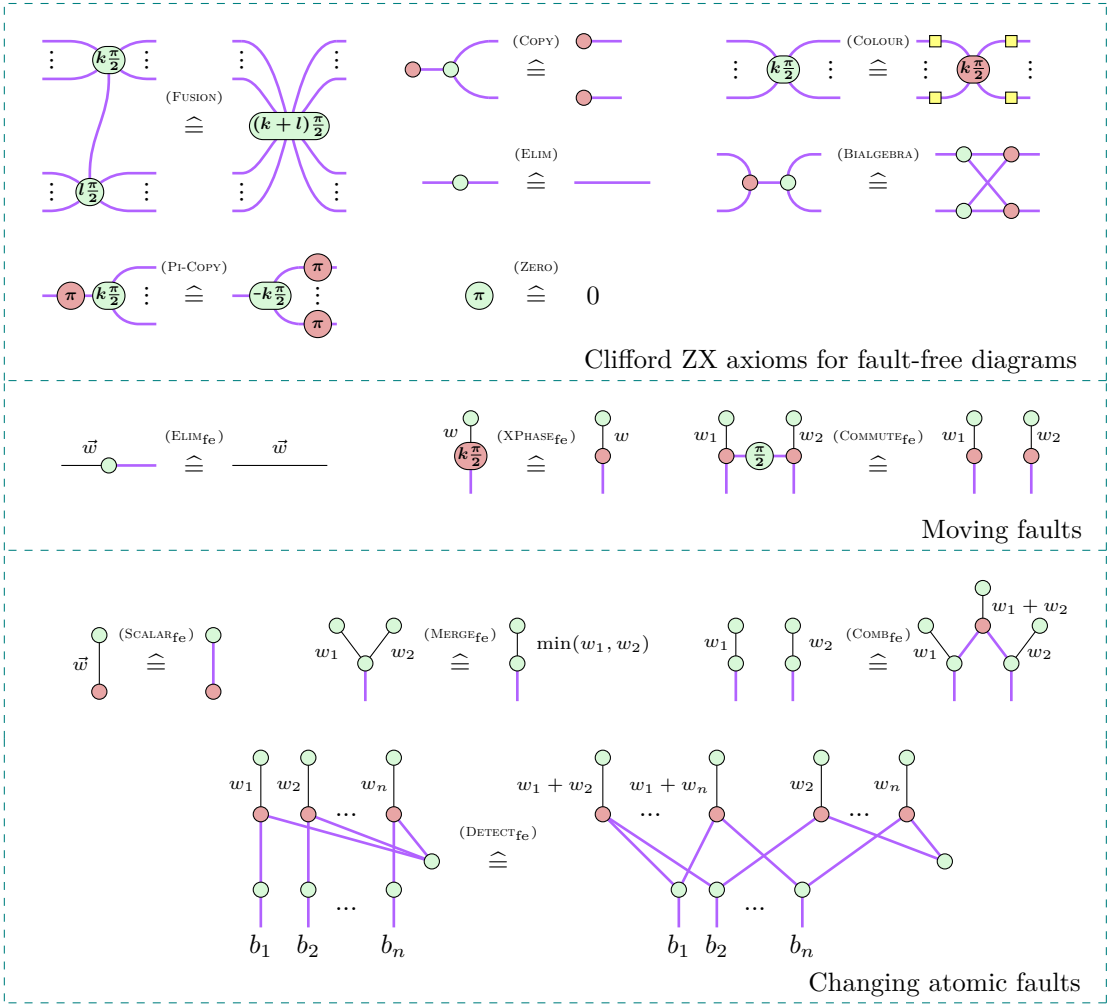}
    \caption{The list of rules which we propose to be complete for fault equivalence.}
    \label{fig:complete-axioms}
\end{figure}
We present our set of rules, proposed to form a complete calculus, in \cref{fig:complete-axioms}.
These rules can be divided into the following three classes:
\begin{enumerate}
    \item \textbf{Fault-Free Clifford}: All rules in this class are fully idealised variants of the basic ZX axiomatisation that we use.
        Among other things, they will be used to obtain the normal form for the semantic part of the separated diagram.
    \item \textbf{Moving Fault Gadgets}: We require additional rules to move fault gadgets throughout a diagram.
        Along with $\TextFEElim$, which enables applications of fully idealised rules next to noisy settings, we introduce $\TextFEXPhase$ and $\TextFECommute$, which allow us to ignore global phases introduced when moving and commuting fault gadgets, respectively.
    \item \textbf{Changing Atomic Faults}: To obtain our novel normal form, we will have to fault-equivalently change the atomic fault set and the weights associated with its members.
        Using only descriptions of $Z$ flips, rules from this class describe intuitive changes to the atomic fault set.

        The first rule $\TextFEScalar$ allows ignoring faults that have no observable effect.
        The second rule $\TextFEMerge$ removes equivalent faults, keeping the one with lower weight.
        The third rule $\TextFEComb$ introduces combinations of faults, again in a way that does not affect the induced weight function.
        Finally, the last rule $\TextFEDetect$ is essential to removing detectable atomic faults that are irrelevant for fault equivalence.
\end{enumerate}

It is essential that these rules are sound, \ie that they preserve fault equivalence, as otherwise we could accidentally equate diagrams that are not fault-equivalent.
Thus, we show:
\begin{restatable}{proposition}{soundnessprop}
    The rules displayed in \cref{fig:complete-axioms} are sound with respect to fault equivalence.
\end{restatable}
\begin{proof}
    See \cref{appendix:missing-proofs}.
\end{proof}
Since fault equivalence is compositional and transitive, we can freely apply these rules in a larger context, and any sequence of sound rewrites is still sound.

Finally, we are now able to state our main result:
\begin{theorem*}
    Let $D, D'$ be Clifford ZX diagrams with respective noise models $\mathcal N, \mathcal N'$ such that $D$ under $\mathcal N$ is fault-equivalent to $D'$ under $\mathcal N'$.
    Further, let $D_{\mathcal{N}}$ be the diagram describing both $D$ and $\mathcal{N}$ using fault gadgets, and similar for $D'_{\mathcal{N}'}$.

    Then, only using the rules in \cref{fig:complete-axioms}, $D_{\mathcal{N}}$ can be deformed into $D'_{\mathcal{N}'}$ and vice versa.
\end{theorem*}
We revisit this theorem in \cref{sec:completeness} with \cref{thm:completeness}, where we will prove it as well.
Afterward, building on the completeness result, we will propose additional rules for obtaining a sound and complete calculus for $w$-fault equivalence, fault boundedness, and $w$-fault boundedness.
    \section{Separating Semantics and Noise}\label{sec:separating-semantics-and-noise}
To bring any diagram into its unique normal form, we follow two steps: First, push all faults to the boundary, then normalise and reduce them as much as possible.
This section focuses on the first step, pushing the faults outside the diagram.
We use fault gadgets~\autocite{rodatz2025faulttoleranceconstruction} to represent faults diagrammatically.
This lets us express any diagram under any noise model as fault-equivalent to one under edge-flip noise.
Extending the work of~\autocite{rodatz2025faulttoleranceconstruction}, we show that fault gadgets can furthermore be used to diagrammatically transform faults into equivalent ones.

\subsection{Fault Gadgets}\label{subsec:fault-gadgets}
To reason about faults in spacetime, we use \textit{fault gadgets}~\autocite{rodatz2025faulttoleranceconstruction}.
First, we define:
\begin{definition}[Pauli boxes]
    \label{def:pauli-boxes}
 The following four diagrams are the \textit{Pauli boxes}, associated with the annotated Pauli rotation:
    \[ \input{04-pushing-out/pauli-boxes.tikz} \]
\end{definition}

Using these, we define:
\begin{definition}
 Let $F$ be a fault of weight $w$ for a ZX diagram $D$ with edges $E$.
 A \textit{fault gadget} for $F$ is constructed by including fault-free Pauli boxes, called \textit{targets}, of the types and on the edges indicated by $F$.
 These boxes are connected via idealised edges to a red spider.
 Further, the red spider is connected via a regular edge to a green spider, called the \textit{spawning edge} of the fault gadget.
 The spawning edge is annotated with $w$.
\end{definition}
For example, the following diagram consists of four standalone wires, \ie it implements a four-qubit identity.
Suppose we now consider the fault $I \otimes Z \otimes X \otimes Y$ that has weight $5$.
The fault gadget corresponding to this fault is:
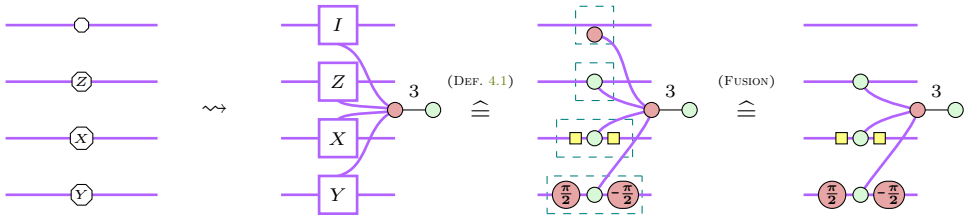
\[ \input{04-pushing-out/gadget-example.tikz} \]
The identity Pauli box is disconnected, so its one-legged spider can be merged into the red spider.
Therefore, for the remainder of this work, we choose not to draw the identity Pauli box at all.

We can now justify our focus on edge-flip noise:
\begin{proposition}
    \label{prop:weighted-edge-flip-universality}
 Let $D$ be a ZX diagram and $\mathcal{N}$ be any weighted noise model for $D$.
 Then there exists a ZX diagram $D_{\mathcal{N}}$ along with a weighted \textit{edge-flip} noise model $\mathcal{N}'$ such that $D = D_{\mathcal{N}}$ and $D$ under $\mathcal{N}$ is fault-equivalent to $D_{\mathcal{N}}$ under $\mathcal{N}'$.
\end{proposition}
\begin{proof}
 Let $D$ be a ZX diagram and $\mathcal{N}: \mathcal{F} \to \mathbb{N}^+$ be a weighted noise model for $D$.
 To obtain $D_{\mathcal{N}}$ and $\mathcal{N}'$, first idealise all edges within $D$.
 Then add a fault gadget for each atomic fault $F \in \mathcal{F}$ with a weight annotation equal to $\mathcal{N}(F)$ on the spawning edge.

 Fault gadgets do not influence the linear map of the underlying diagram~\autocite{rodatz2025faulttoleranceconstruction}, so we directly have $D = D_{\mathcal{N}}$.
 Additionally, for every fault $F$ on $D$, there exists a corresponding fault in $D_{\mathcal{N}}$ on its spawning edge of the corresponding fault gadget.
 Thus, for all $F$ in $\mathcal{F}$, there exists a fault $F'$ on $D_{\mathcal{N}}$ with $wt_{\mathcal{N}'}(F') = wt_{\mathcal{N}}(F)$ such that $D^F = D_{\mathcal{N}}^{F'}$.
 Similarly, by construction, the only allowed faults in $D_{\mathcal{N}}$ are on spawning edges of fault gadgets and naturally correspond to a fault in $D$.
 Therefore, the two diagrams are fault-equivalent.
\end{proof}

For our completeness result, we only reason about diagrams in this form, meaning that prior to the diagrammatic reasoning, some straightforward pre-processing is necessary.

\begin{example}{Four-legged spider}{}
    Returning to our example of the four-legged spider, we can show what this construction looks like.
    To keep an interesting example and reduce clutter at the same time, we focus on the expanded diagram under just $X$ edge-flips.
    \[\input{running-example/create-dn.tikz}\]
\end{example}

\begin{remark}
 Going one step further, we can observe that, in fact, unweighted edge-flip noise, \ie assigning each atomic edge-flip a weight of one, is already sufficiently expressive.
 More formally, for any fault gadget with weight $w$, we can find an equivalent fault gadget under unweighted edge-flip noise:
    \[\input{04-pushing-out/head-addition-all-one.tikz}\]
 For more details, in particular how this is proven using just the rules presented in \cref{sec:the-rules-and-strategy}, see~\autocite[Lemma 4.5]{rueschCompletenessFaultTolerant2025}.

 Weighted edge-flip noise can be seen as syntactic sugar for unweighted edge-flip noise.
 However, as weighted edge-flip noise is easier to work with, we choose to focus on that.
\end{remark}

\subsection{Inconsequential Faults}\label{subsec:inconsequential-faults}

We can identify a special group of faults called inconsequential faults~\autocite{MagdalenadelaFuente2025XZYruby}, also referred to as trivial faults~\autocite{blackwell2025codedistancefloquetcodes}.
These are the faults in spacetime that do not change a diagram's underlying linear map, \ie the faults that have no consequence.

We define:
\begin{definition}[Inconsequential faults]
 Let $D$ be a non-zero ZX diagram with edges $E$.
 Then a Pauli $S \in \paulifaults{|E|}$ is an \textit{inconsequential fault} of $D$ if $D^S = D$.
\end{definition}

Faults that correspond to stabilisers of $D$ are inconsequential.
However, there are more inconsequential faults that are spread through spacetime.
For example, we can observe that the following faults all leave the diagram unchanged:
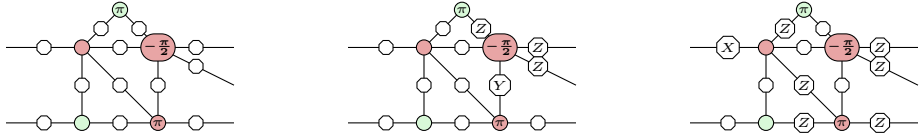
\[\input{04-pushing-out/example-inconsequential-faults.tikz}\]
The first fault is the trivial fault, which is clearly inconsequential.
The second fault corresponds to a spider stabiliser of the red $-\frac{\pi}{2}$-phased spider.
The third fault is a product of multiple spider stabilisers.
This leads us to a more general proposition:

\begin{proposition}
    \label{prop:generating-inconsequential-faults}
 The group of inconsequential faults is equivalent to the group generated by all spider stabilisers.
\end{proposition}
\begin{proof}
 Clearly, any spider stabiliser leaves the diagram unchanged and therefore is an inconsequential fault.
 Furthermore, any product of inconsequential faults is again inconsequential.
 Therefore, it remains to be shown that any inconsequential fault can be generated by stabilisers of spiders.

 Let $D$ be a non-zero ZX diagram and let $F$ be an inconsequential fault of $D$.
 Then we can construct a diagram $D'$ which helps us find a generating set of spider stabilisers that generate $F$.
 Let $D'$ be $D$ where we additionally add a fault gadget for $F$, however, leaving the spawning edge open, making it a boundary edge.
 For example, when $F$ has two non-trivial targets, we get:
    \[\input{04-pushing-out/firing-assignment-building-d-prime.tikz}\]
 Then, we can observe that placing a single $Z$ flip on the newly created boundary edge and pushing it inwards has the same effect as applying $F$ to $D$ and must therefore be inconsequential.
 In other words, the Pauli that acts as $Z$ on the new boundary edge and trivially everywhere else is a stabiliser of $D'$.

 But then, by~\autocite{borghansZXcalculusQuantumStabilizer2019}, there must exist a generating set of spider stabilisers of $D'$ that generate this Pauli.
 We can study this set of spider stabilisers more closely to find a generating set of spider stabilisers of $D$ that generate $F$.

 First, we observe that to create the $Z$ flip on the spawning edge, we must use the stabiliser of the red spider of the fault gadget, introducing $Z$ flips on all the edges connected to the targets.
 To remove these flips, within the Pauli boxes, we must use the stabiliser of the green spider, introducing another flip on the target edge, either to the left or to the right of this spider.
 To remove these flips, the stabiliser of the Clifford conjugation must be used, introducing a flip on the target edge of the colour of the Pauli box.
    \[\input{04-pushing-out/firing-assigning-in-D-prime.tikz}\]
 But then, having multiplied these spider stabilisers, we now have a fault that is equivalent to $F$ on $D$.
 As we know that the overall multiplication is inconsequential on $D'$, the remaining stabilisers must cancel out $F$.
 But then, as all the remaining spiders have a direct correspondence in $D$, the corresponding stabilisers in $D$ must similarly multiply to $F$.
 Therefore, we have shown that there must exist a generating set of spider stabilisers in $D$ that generate $F$.

 We remark that the algorithm from~\autocite{borghansZXcalculusQuantumStabilizer2019} only works for diagrams in red-green form, \ie diagrams where red spiders are only connected to green spiders and vice versa.
 However, we can always bring a diagram into red-green form by introducing an identity spider of the opposite colour on all edges that connect two spiders of the same colour.
 The obtained spider stabilisers of this diagram can be directly translated back to spider stabilisers of the original diagram by ignoring all introduced identity spiders.
\end{proof}

Using these inconsequential faults, we can define an equivalence relation on faults:
\begin{definition}[Spacetime equivalence]
 Two faults $F_1, F_2$ on a non-zero ZX diagram $D$ are \textit{spacetime-equivalent} if there exists an inconsequential fault $S$ such that $F_1 S = F_2$.
\end{definition}

Spacetime equivalence of faults guarantees that two faults are the same up to some inconsequential fault, which does not change the diagram, \ie that the two faults change the diagram in the same way.
Therefore, when considering how to identify and, in particular, correct faults, it is sufficient to identify faults up to spacetime equivalence.
We can formalise this:
\begin{proposition}
    \label{prop:spacetime-equivalence}
    Let $D$ be a non-zero ZX diagram and $F_1, F_2$ be faults on $D$.
    $F_1$ and $F_2$ are spacetime-equivalent if and only if they anticommute with the same Pauli webs.
\end{proposition}
\begin{proof}
    If $F_1$ and $F_2$ are spacetime-equivalent, there exists an inconsequential fault $S$ such that $F_1 S = F_2$.
    As $F_1 S$ and $F_2$ are the same, we know that they anticommute with the same Pauli webs.
    We now argue that $S$ commutes with all Pauli webs and therefore, $F_1$ and $F_2$ must anticommute with the same Pauli webs.
    By \cref{prop:faults-flip-pauli-webs}, we know that faults flip the sign of Pauli webs when the fault anticommutes with the Pauli web.
    When they flip the sign of a detecting region, they make the diagram zero and are therefore detected.
    As inconsequential faults do not change the diagram, they do not make it go to zero.
    Therefore, they have to commute with all detection regions.
    For undetectable faults, when they flip the sign of a stabilising Pauli web, they change the eigenspace that $D^F$ lives in with respect to the corresponding stabiliser.
    But as inconsequential faults do not change the diagram, they can also not change the eigenspace that $D^F$ lives in.
    Therefore, they have to commute with all stabilising Pauli webs as well.
    But then, as $F_1 S$ and $F_2$ anticommute with the same Pauli webs and $S$ commutes with all Pauli webs, $F_1$ and $F_2$ must anticommute with the same Pauli webs.

    Conversely, let $F_1$ and $F_2$ be faults that anticommute with exactly the same Pauli webs.
    But then $F_1 F_2$ must commute with all Pauli webs, as $F_1$ and $F_2$ either both commute or both anticommute with any given web.
    Thus, $D^{F_1 F_2} = D$ as $F_1 F_2$ does not flip the sign of any detecting region or stabiliser.
    Thus, $S = F_1 F_2$ must be an inconsequential fault.
    But then $F_1 S = F_1 F_1 F_2 = F_2$, so $F_1$ and $F_2$ are spacetime-equivalent.
\end{proof}

But then we have:
\begin{proposition}
    \label{prop:spacetime-equivalence-undetectable}
    Two undetectable faults $F_1, F_2$ on a non-zero diagram $D$ are spacetime-equivalent if and only if $D^{F_1} = D^{F_2}$.
\end{proposition}
\begin{proof}
    By \cref{prop:spacetime-equivalence}, $F_1$ and $F_2$ are spacetime-equivalent if and only if they anticommute with the same Pauli webs.
    Therefore, they must anticommute with the same stabilising Pauli webs.
    But then $D^{F_1}$ and $D^{F_2}$ must live in the same eigenspaces of all the stabilisers of $D$ and therefore, by stabiliser theory, they must be the same.
\end{proof}

For detectable faults, we recall that a fault is detectable if it causes the diagram to go to zero.
Therefore, we cannot strengthen the statement from \cref{prop:spacetime-equivalence-undetectable} to include detectable faults, since two detectable faults may not be spacetime-equivalent, yet both make the diagram go to zero, \ie satisfy $D^{F_1} = D^{F_2}$.
However, as spacetime-equivalent faults must anticommute with the same Pauli webs, including all the detecting regions, each spacetime equivalence class of faults is either detectable or undetectable.
This leaves us with the following adaptation of \cref{prop:spacetime-equivalence-undetectable}:
\begin{corollary}
    \label{corr:spacetime-equivalence}
    If $F_1, F_2$ are two spacetime-equivalent faults on a diagram non-zero $D$, it must be that $D^{F_1} = D^{F_2}$.
\end{corollary}
\begin{proof}
    If $F_1$ and $F_2$ are undetectable, this follows from \cref{prop:spacetime-equivalence-undetectable}.
    If they are detectable, they both make the diagram go to zero.
    Therefore, the statement holds in all cases.
\end{proof}
Fault equivalence, as defined, is only concerned with the faulty diagrams $D^{F_1}, D^{F_2}$, so through \cref{corr:spacetime-equivalence} it is safe to handle faults up to spacetime equivalence.

\subsection{Manipulating Fault Gadgets}\label{subsec:manipulating-fault-gadgets}
Beyond being capable of representing faults, fault gadgets can also be used to reason about faults.
Since we just argued that for fault equivalence we only care about faults up to spacetime equivalence, we need to show that we can transform spacetime-equivalent fault gadgets into one another.
First, we show:
\begin{restatable}[Adding spider stabilisers to fault gadgets]{proposition}{addingspiderstabilisers}
    \label{prop:adding-spider-stabilisers}
    Let $D$ be a fault-free ZX diagram with a single fault gadget for a fault $F$.
    Then, for any spider stabiliser $S$ of a spider $s$ in $D$, we can use our fault-equivalent rules to add the elements of $S$ as targets to the fault gadget.
\end{restatable}
\begin{proof}
    Let $s$ be a green spider in $D$ with stabiliser $S$ introducing $Z$ flips on two of its legs and acting trivially on all other edges.
    Then, we can first observe:
    \begin{equation}
        \label{eq:stab-fe-own-color}
        \input{04-pushing-out/stab-fe-own-color.tikz}
    \end{equation}

    Using this, we have:
    \[\input{04-pushing-out/stab-fe-add-target.tikz}\]

    All other stabilisers can be added similarly.
    See \cref{appendix:missing-proofs}.
\end{proof}

But then, using this, we can go one step further:
\begin{proposition}
    \label{prop:changing-fault-gadgets-1}
    Let $D$ be a fault-free ZX diagram with spacetime-equivalent faults $F_1, F_2$.
    Then, using our fault-equivalent rules, we can rewrite the fault gadget for $F_1$ into a fault gadget for $F_2$.
\end{proposition}
\begin{proof}
    As $F_1$ and $F_2$ are spacetime-equivalent, by definition, there exists some inconsequential fault $S$ such that $F_1 S = F_2$.
    But then, by~\cref{prop:generating-inconsequential-faults}, we know that there exists a set of stabilisers $S_1, \dots, S_m$ of spiders in $D$ such that $S_1\dots S_m = S$.
    Then, we can iteratively apply the previous proposition to add the targets of $S_1, \dots, S_m$ to the fault gadget for $F_1$.
    We know that $F_1 S_1 \dots S_m$ multiplies to $F_2$.
    Thus, all we have to show is that if multiple targets are on the same edge, we can merge them into at most a single target, obeying the multiplication rules of Pauli matrices.

    First, we show that if we have two targets of the same type on the same edge, they cancel each other out.
    To start, we see that every non-trivial Pauli box consists of a Clifford diagram $C$, a green spider, and the adjoint $C^\dagger$ of $C$:
    \begin{equation}
        \label{eq:pauli-box-clifford-decomposition}
        \input{04-pushing-out/pauli-boxes-cliffords.tikz}
    \end{equation}

    Using this, we can observe:
    \[\input{04-pushing-out/targets-self-inverse.tikz}\]

    Next, we show that a $Y$ target can be rewritten into an $X$ and a $Z$ target:
    \[\vcenter{\hbox{\scalebox{0.95}{\input{04-pushing-out/y-decomposition.tikz}}}}\]

    Finally, we show that we can transpose $Y$ targets:
    \[\input{04-pushing-out/target-transpose-y.tikz}\]
    Using this, we can commute $X$ and $Z$ by merging them into a $Y$ target, transposing the $Y$ target, and then decomposing it again in the other direction.

    Using these three rules, we can merge any number of targets on the same edge into at most a single target, obeying the multiplication rules of Pauli matrices.
    If any of the targets are $Y$ targets, we decompose them into their $X$ and $Z$ components.
    Then, we commute all $X$ targets to one side of the edge and all $Z$ targets to the other side.
    Finally, we cancel out pairs of $X$ and $Z$ targets, leaving at most one $X$ and one $Z$ target on the edge.
    If we have both, we can replace them with a single $Y$ target.
\end{proof}

\begin{example}{Four-legged spider}{}
    For example, we can use a spider stabiliser of the top, right spider to move the following fault gadget around the diagram:
    \[\input{running-example/move-single-gadget.tikz}\]
\end{example}

Thus, for diagrams with a single fault gadget, we can manipulate the fault gadget to represent any spacetime-equivalent fault.
However, we can go one step further.
For this, we first observe:
\begin{proposition}
    \label{prop:commuting-targets}
    Let $D$ be a fault-free ZX diagram where the only allowed faults are on the spawning edges of fault gadgets.
    Then, we can freely commute the targets of different fault gadgets on the same edge past each other.
\end{proposition}
\begin{proof}
    Intuitively, we can commute fault gadgets past each other, as we only care about diagrams up to a global phase.
    For completeness, we have to show that we can do this with the proposed rule set as well.
    For targets of the same type, we have:
    \[\input{04-pushing-out/pp-commute.tikz}\]
    For an $X$ and a $Z$ target, we have:
    \[\input{04-pushing-out/zx-commute.tikz}\]
    We can now commute arbitrary targets on the same edge by first decomposing all targets into their $X$ and $Z$ components and then commuting the components using the two rules above.
\end{proof}

But then, we have:
\begin{proposition}
    \label{prop:changing-fault-gadgets-2}
    Let $D$ be a fault-free ZX diagram where the only allowed faults are on the spawning edges of fault gadgets, and let $F_1, F_2$ be spacetime-equivalent faults on $D$.
    Then, using our fault-equivalent rules, we can rewrite the fault gadget for $F_1$ into a fault gadget for $F_2$.
\end{proposition}
\begin{proof}
    Similar to~\cref{prop:changing-fault-gadgets-1}, we can add the targets of a generating set of stabilisers to the fault gadget for $F_1$.
    To merge all the targets, we can apply \cref{prop:commuting-targets} to make sure that all targets on the same edge are next to each other.
    Finally, we can merge all targets on the same edge as in the previous proposition.
\end{proof}

\subsection{Pushing Faults Out}\label{subsec:pushing-out-faults}
Our goal is to push all the faults to the boundary.
For this, we provide a subset of edge-flips that together generate all spacetime equivalence classes of faults.
This allows us to express any fault as a spacetime equivalent fault that only acts on those edges.
We then choose these faults to mostly act on the boundary, with a few remaining in the interior of the diagram.
Afterwards, in \cref{subsec:removing-flip-operators}, we remove these remaining internal faults as well.

For this, we first define:
\begin{definition}[Flip operator collection]
    Given a ZX diagram $D$ with edges $E$ and a basis of the diagram's Pauli webs $P_1,\dots,P_m \in \paulifaults{|E|}$, a flip operator collection is a set
    \[\mathcal{F}_{flipop} \subseteq \paulifaults{|E|} \]
    such that for each $P_i$, there exists a unique fault $F_{P_i} \in \mathcal{F}_{flipop}$ that anticommutes with $P_i$ and commutes with all other $P_j$ of the basis.
\end{definition}

\begin{example}{Four-legged spider}{}
    \label{ex:flip-operators}
    Our running example has a Pauli web basis consisting of one detecting region and three stabilising webs.
    There are a range of valid flip operator collections, and we show this in the choice of the region flip operator.
    Recall that flip operators are faults, so we can visualize them using fault gadgets.
    Together with the three stabiliser flip operators $F_1,F_2,F_3$, one of the three region flip operators $F_4, F_4'$ or $F_4''$ would form a valid collection:
    \vspace{-4pt} \[ \input{running-example/flip-operators.tikz} \vspace{-4pt} \]
    For simplicity, we continue our example using the operator $F_4$.
\end{example}

A flip operator collection can be understood as providing a way to flip any Pauli web of the diagram, meaning that we can create any anticommutation relationship with the diagram's Pauli webs by taking the product of multiple flip operators.
As the spacetime equivalence class of a fault is uniquely defined through the Pauli webs with which it commutes by \cref{prop:spacetime-equivalence}, the group generated by a flip operator collection generates all spacetime equivalence classes of faults.
We have:

\begin{proposition}
    \label{prop:flip-op-rewrite}
    Let $D$ be a non-zero ZX diagram with some noise model $\mathcal{N}$ and flip operator collection $\mathcal{F}_{flipop}$ for $D$.
    Then, we can fault-equivalently rewrite $D_{\mathcal{N}}$ into a diagram $D_{\mathcal{N'}}$ such that all atomic faults in $\mathcal{N'}$ are generated by $\langle \mathcal{F}_{flipop} \rangle$.
\end{proposition}
\begin{proof}
 Recall that $D_{\mathcal{N}}$ is the diagram obtained by \cref{prop:weighted-edge-flip-universality} to create a diagram under edge-flip noise that is fault-equivalent to $D$ under $\mathcal{N}$.
 Further recall that by \cref{prop:spacetime-equivalence} two faults are spacetime-equivalent if and only if they have the same commutation relationship with all Pauli webs of $D$.
 Thus, if some fault $F$ has some anticommutation relationship with the Pauli web basis that the flip operators are defined with respect to, we can construct a fault $F_{flipop}$ that has the same anticommutation relationship with that Pauli web basis.
 By the linearity of Pauli web commutation and since the Pauli webs, with respect to which the flip operators are defined, form a basis, $F$ and $F_{flipop}$ must have the same anticommutation relationship with all Pauli webs.

 As we can replicate flips of those basis elements using $\mathcal{F}_{flipop}$, we can, for any atomic fault $F \in \mathcal{N}$, construct a spacetime-equivalent fault $F' \in \langle \mathcal{F}_{flipop} \rangle$.
 But by \cref{prop:changing-fault-gadgets-2}, we can rewrite fault gadgets for spacetime-equivalent faults into one another, so the gadget for $F$ can be rewritten into the gadget for $F'$.
\end{proof}

For our normal form, we identify specific flip operator collections that allow us to push all faults to the boundary.
For this, we first observe:
\begin{proposition}
    \label{prop:pauli-web-basis}
    Let $D$ be a ZX diagram with $n$ boundary edges.
    Then there exists a basis of all Pauli webs $S_1, \dots, S_n, R_1, \dots, R_d$, such that $S_1, \dots, S_n$ are stabilising Pauli webs and $R_1,\dots, R_d$ are detecting regions.
\end{proposition}
\begin{proof}
    Let $S_1,\dots,S_{n'}, R_1,\dots,R_{d'}$ be a basis of Pauli webs for $D$.

    We first observe that $n' \geq n$.
    As $D$ has $n$ boundary edges, it has $n$ independent stabilisers.
    By~\autocite{borghansZXcalculusQuantumStabilizer2019}, we know that for each stabiliser there must exist a corresponding Pauli web with the same action on the boundary, and the independence of the stabilisers translates onto independence of the Pauli webs.
    This means that from the assumed basis $S_1,\dots,S_{n'},R_1,\dots,R_{d'}$ at least $n$ of the basis elements must be stabilising Pauli webs.

    If $n' > n$, we can find a new basis where we only have $n$ stabilising Pauli webs.
    Since $D$ only has $n$ independent stabilisers, there must be some linear combination of Pauli webs $S_{i_1}',\dots,S_{i_j}'$ that has no action on the boundary, \ie it is a detecting region.
    But then we can replace one of the elements, \eg $S_{i_1}'$, by the linear combination.
    Clearly, this is still a basis for all Pauli webs, as the original element may be recovered by combining with the remaining elements $S_{i_2}',\dots,S_{i_j}'$.
    However, the new basis contains an additional detecting region and one fewer stabilising web.
    We can do this until we have no more stabilising Pauli webs whose action on the boundary is linearly dependent, which will be the case when $n' = n$.
    Therefore, the updated basis of detecting regions will have exactly $n$ stabilising Pauli webs.
\end{proof}

Next, we can state:
\begin{proposition}
    \label{prop:flip-operators-construction}
    Let $D$ be a ZX diagram with $n$ boundary edges.
    Then there exists a basis of Pauli webs $S_1, \dots, S_n, R_1, \dots, R_d$ along with a flip operator collection such that
    \begin{itemize}
        \item each flip operator consists of exactly one edge flip
        \item the flip operators for the stabilising Pauli webs act on only boundary edges
    \end{itemize}
\end{proposition}
\begin{proof}
    By~\cref{prop:pauli-web-basis}, we know there exists a basis of Pauli webs $S_1, \dots, S_n, R_1, \dots, R_d$ such that exactly $n$ of the basis elements are stabilising Pauli webs.
    We now construct a new basis $R_1',\dots,R_d',S_1',\dots,S_n'$ along with a flip operator collection that satisfies our requirements.
    We do this incrementally for each basis element, starting with the detecting regions.

    Starting with $i = 1, \dots, d$, we pick one of the edges $R'_i$ highlights and consider a flip operator $F_{R'_i}$ that anticommutes with the action of $R'_i$ on that edge.
    This means that $F_{R'_i}$ anticommutes with at least $R'_i$, but there may be other basis elements it anticommutes with.
    To ensure that $F_{R'_i}$ only anticommutes with $R'_i$, for all $j \neq i$ we update $R'_j \leftarrow R'_j R'_i$ if $R'_j$ anticommutes with $F_{R'_i}$ and similarly for the $S'_j$'s.
    The elements still form a basis for all Pauli webs, as the original element may be recovered by combination with $R'_i$.
    So $F_{R'_i}$ is a flip operator for $R'_i$, and the flip operators for the basis elements with $j < i$ remain valid since, through previous iterations, we ensured that they do not interfere with $R'_i$.

    Once we have chosen the flip operators for the detecting regions, we can continue with the flip operators for the stabilising Pauli webs, proceeding in the same fashion.
    However, we can now pick flip operators on the boundary, as each stabilising web non-trivially highlights at least one boundary edge by definition.
    Furthermore, changing the basis as above will never prevent us from picking a flip operator on the boundary by forming a detecting region, since stabilising basis elements correspond to the $n$ independent stabilisers of $D$.

    Therefore, we have constructed a basis of all Pauli webs $R'_1, \dots, R'_d,S'_1, \dots, S'_n$ along with a flip operator collection consisting of single edge-flips, where the flip operators for the $S'_i$'s act on the boundary of $D$.
\end{proof}

Leveraging \cref{prop:flip-op-rewrite} and \cref{prop:flip-operators-construction}, we can rewrite any diagram $D_{\mathcal{N}}$ into a spacetime-equivalent one where all faults only act on flip operators.
Furthermore, the only flip operators that remain inside the diagram are flip operators for the detecting regions.

\begin{example}{Four-legged spider}{}
    We can now push all the fault gadgets of our diagram to only act on the flip operators from \cref{ex:flip-operators}:
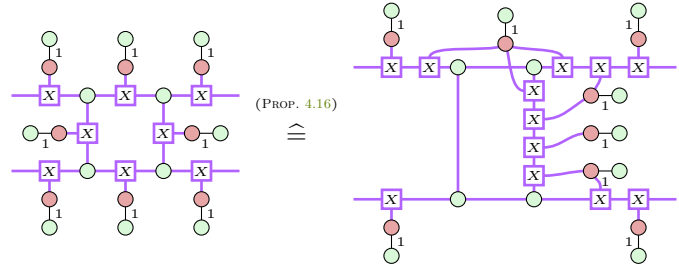
    \[\input{running-example/move-all-gadgets.tikz}\]
\end{example}

\subsection{Removing Internal Flip Operators}\label{subsec:removing-flip-operators}
We have now taken major steps towards fully separating the fault gadgets from the bulk of the diagram.
However, fault gadgets may still have targets in the bulk of the diagram: the flip operators for the detecting regions.
Therefore, as a last step to fully isolate the fault gadgets, we now have to remove these flip operators from the diagram.

First, we observe:
\begin{proposition}
    \label{prop:zero-diagram-disconnect}
 Let $D$ be a Clifford ZX diagram with boundary edges $B$, and let $b \in B$ be a boundary edge such that capping $b$ with a green $\pi$-phase makes the composite diagram zero, \ie:
    \[ \input{04-pushing-out/zero-diagram-precondition.tikz} \]
 Then we can rewrite $D$ fault-equivalently as:
 \begin{equation}
        \label{eq:zero-diagram-disconnect-proposition}
 \input{04-pushing-out/zero-diagram-proposition.tikz}
 \end{equation}
\end{proposition}
\begin{proof}
 As the two diagrams are both idealised as fault-free, by \cref{prop:fully-idealised-equivalence} we only need to show that the two diagrams are semantically equivalent.
 Then, as our rewrite rules include the complete set of Clifford ZX rewrites for fault-free diagrams, we know that we can rewrite one into the other.

 To show semantic equivalence, we employ the fact that two diagrams are semantically equivalent if and only if they have the same action on all elements of a basis~\autocite{coecke2018picturing}.
 For our purposes, we choose the $X$-basis, represented through green $\{0,\pi\}$-phase spiders.
 We have:
    \[ \input{04-pushing-out/disconnect-+.tikz} \]
 and
    \[ \input{04-pushing-out/disconnect--.tikz} \]
 So the two diagrams have the same action on the $X$-basis elements.

 We remark that for these equalities to compose to the equality we desire, it is crucial that any reasoning done here is scalar accurate.
 Therefore, we have rescaled the LHS diagram accordingly.
 As the only equalities we used were OCM and the spider fusion rule, which are both scalar accurate, the equalities are also scalar accurate.
 But now, as the rescaled diagram is exactly equivalent to the RHS, we know that the original two diagrams are semantically equivalent up to a global scalar.
\end{proof}

But then we can finally state:
\begin{proposition}
    \label{prop:disconnect-spiders-summary}
 Let $D$ be a non-zero ZX diagram with some noise model $\mathcal{N}$.
 Then, we can fault-equivalently rewrite $D_{\mathcal{N}}$ into a diagram $D_{\mathcal{N}'}$ such that all faults drawn from $\mathcal{N}'$ are represented by fault gadgets whose targets only act on the boundary or free-floating spiders.
\end{proposition}
\begin{proof}
    First, we use~\cref{prop:flip-operators-construction} to find a flip operator collection for $D$.
    Then, we use~\cref{prop:flip-op-rewrite} to bring the diagram into a form where all faults only act on these flip operators and the boundary.
    Finally, we use~\cref{prop:zero-diagram-disconnect} to remove all flip operators that act on internal edges.

    We recall that each flip operator obtained by~\cref{prop:flip-operators-construction} only acts on a single edge.
    Then, in preparation for the last step, we see that we can merge all targets belonging to the flip operator on such an edge:
    \[\input{04-pushing-out/poping-out-1.tikz}\]
    So we effectively recover the single edge flip operator as a Pauli box that is controlled by the fault gadgets.

    But then, as the internal flip operators were constructed specifically for anticommuting with detecting regions, we know that this edge with the new Pauli box on it is in the context of a detecting region, and thus a $\pi$ phase as a control on the box would make the overall diagram zero.
    Thus, by~\cref{prop:zero-diagram-disconnect}, where we symbolize the surrounding diagram as $D_{env}$, we can disconnect a free-floating spider from the remainder of the diagram:
    \[\input{04-pushing-out/poping-out-2.tikz}\]
    After we iteratively apply this to remove all internal flip operators, we have fault-equivalently rewritten the diagram into a form where all faults only act on the boundary or free-floating spiders.
\end{proof}

We can observe that if an odd number of faults occur that are connected to a free-floating spider, the overall diagram becomes zero.
As such, the free-floating spiders are taking the role of the detecting regions, which are being detached from the original diagram.
Faults, now represented as fault gadgets, are connected exactly to those free-floating spiders that represent the detecting regions they violate.
Therefore, we can read off the detector-error model~\autocite{Gidney2021stimfaststabilizer} directly from the diagram; each fault gadget represents an atomic fault, each free-floating spider a detecting region, and the connections between them represent which faults violate which detecting regions.

\begin{example}{Four-legged spider}{}
    In our example, we only have one detecting region, corresponding to one flip operator.
    Disconnecting this operator, we get:
    \[\input{running-example/pop-out-detecting-region.tikz}\]
    The only part of the diagram with a non-trivial semantics is the fault-free subdiagram highlighted by the dotted box.
    Everything else consists of fault gadgets and detecting regions which have a trivial semantics, however, non-trivial behaviour under noise.
\end{example}

We have now managed to completely separate the semantic part of the diagram, which is idealised as fault-free, from the faults, which are only represented through fault gadgets connected to the free-floating detecting regions and the boundary edges using only fault-equivalent rewrites.
As we have previously shown that these rewrites are sound, \ie that they preserve fault equivalence, we are guaranteed that the resulting diagram must be fault-equivalent to the diagram we started with.
    \section{Normal form}\label{sec:normal-form}
So far, we have rewritten our diagram such that all fault gadgets only act on the boundary or on detecting region spiders.
We thus achieved a separation between the part of the diagram describing the underlying semantics and the part of the diagram describing the noise model through fault gadgets.
By rewriting both parts into respective unique normal forms, we can achieve an overall unique normal form for the entire diagram.

The first part, describing the underlying semantics, is a fully idealised Clifford ZX diagram.
Every Clifford ZX diagram has a unique normal form, namely the reduced AP form~\autocite{KissingerWetering2024Book}.
We allow all rules from~\autocite{backens2017AsimplifiedStabiliserZXCalculus} in their fully idealised form, and as they provide a complete calculus, the reduced AP form is reachable through our axiomatisation.
Thus, we can reuse this form in the fully idealised setting for the semantic part of the diagram.

It remains to show that the fault gadget part of the diagram has a unique normal form.
For this, we consider how the pushed-out fault gadgets of two fault-equivalent diagrams may differ, which reduces to four challenges:
\begin{enumerate}
    \item There are many possible choices of flip operators on the boundary.
    To show that two diagrams are fault-equivalent, we have to ensure that we can choose the same flip operators on both diagrams.
    \item The fault-equivalent diagrams may have different detectable faults.
    Fault equivalence only requires their undetectable faults to be the same.
    \item Two fault-equivalent diagrams may have different atomic faults, \ie generating sets for all possible faults.
    Fault equivalence only cares about the group of all possible undetectable faults.
    \item In a single diagram, multiple faults may have the same effect.
    For fault equivalence, we only care about the lowest-weight fault in each spacetime equivalence class.
    \item Fault gadgets naturally have an order on the boundary edges, \ie every diagram features a sequence of fault gadgets.
        For fault equivalence, we only care about the set of all undetectable faults, \ie an unordered collection.
\end{enumerate}
We will tackle these challenges individually and fault-equivalently until reaching a unique normal form for fault gadget descriptions of noise models.

For the first challenge, we observe that across the two fault-equivalent diagrams under comparison, we must be careful to choose the same flip operator collection for the boundary.
Depending on the Pauli web basis we started with, the construction from \cref{prop:flip-operators-construction} may produce different flip operator collections.
We now show:
\begin{proposition}
    \label{prop:flip-op-stabiliser-order}
    Let $D$ be a ZX diagram with $n$ boundary edges.
    There exists a total order for stabiliser flip operator collections valid for $D$ such that:
    \begin{itemize}
        \item there exists a uniquely smallest edge-flip collection \wrt that order, and
        \item for any two diagrams $D_1 = D_2$, the uniquely smallest edge-flip collections \wrt that order are the same.
    \end{itemize}
\end{proposition}
\begin{proof}
    We start by defining a total order of stabiliser flip operator collections.
    Since $\{I,X,Y,Z\}$ can be totally ordered lexicographically, and a ZX diagram has a canonical well-order of its boundary edges, we receive an induced total order on Pauli products.
    Further, we can lift this total order to sets of Pauli products of a fixed size: For two collections of the same fixed size, determine the smallest operator in each and compare it across the collections, progressing to compare the second-smallest operator iff the smallest operators are equal, etc.

    The total number of possible flip operator collections for a diagram $D$ is finite, so there exists a uniquely smallest collection for $D$, and in particular, there exists one consisting just of edge-flips.
    Furthermore, two diagrams $D_1,D_2$ satisfy $D_1 = D_2$ iff their stabilisers are the same, so \cref{prop:flip-operators-construction} yields a single flip operator collection that works for both diagrams.
    But then, from that shared start point, the uniquely smallest edge-flip collection must be the same.
\end{proof}
We remark that, even though it is convenient for defining a normal form, in practice, it is not required to use the smallest flip operator collection.
Once a stabiliser flip operator collection is constructed for one diagram, it can be reused for any other semantically equivalent diagram.

The second challenge is that two fault-equivalent diagrams might have different detectable faults.
To illustrate, in our running example, the unfused spider has a detecting region and therefore has detectable faults, whereas the fused spider has no detectable faults.
We show that we can remove detectable faults entirely by transforming the noise model:
\begin{proposition}
    \label{prop:detecting-region-popping}
 Let $D$ be a non-zero ZX diagram with a noise model $\mathcal{N}$.
 Then, using the fault-equivalent axioms, we can fault-equivalently rewrite $D_{\mathcal{N}}$ into $D'$ where all fault gadgets represent undetectable faults, \ie no fault gadget is connected to a detecting region spider.
\end{proposition}
\begin{proof}
 As shown in \cref{sec:separating-semantics-and-noise}, we can fault-equivalently rewrite $D_{\mathcal{N}}$ such that all fault gadgets only act on the boundary of the diagram and detecting region spiders.
 For a single detecting region spider $s$, $\FEDetect$ can be applied to reduce the number of fault gadgets connected to $s$ by one, at the expense of introducing pairwise combinations with all gadgets still connected to $s$.
 This rewrite may itself be iterated until no fault gadget connects to $s$, leaving it as a free-floating scalar diagram equal to the scalar $2$.
 But as we handle diagrams up to a non-zero scalar, we can simply remove $s$.
 We apply this rewrite until no such spiders remain, leaving fault gadgets that can only act on the boundary.
 These fault gadgets represent faults that must be undetectable.
\end{proof}

\begin{example}{Four-legged spider}{}
 Applying this procedure to our running example, we get:
    \[\input{running-example/removing-detecting-region.tikz}\]
 We make use of the fact that Pauli boxes can be merged, as shown in the proof of~\cref{prop:disconnect-spiders-summary}, to write Pauli boxes with multiple output wires, instead of multiple Pauli boxes.
\end{example}

Note that when there are two detecting region spiders $s,s'$ that are connected through a common fault gadget, the number of combinations introduced through the above process is multiplicative in the number of combinations involving that gadget for $s$ and $s'$ individually.
In the worst case, the total number of combinations may grow exponentially with the number of detecting regions.
However, this is not unexpected: Checking fault equivalence is NP-hard~\autocite{rodatz2025faulttoleranceconstruction}.
Therefore, any unique normal form is expected to require exponentially many steps to be obtained.

The third challenge is that we are not yet guaranteed to have described \textit{all} possible undetectable faults.
In particular, fault-equivalent diagrams may have different, equivalent generating sets of atomic faults.
To overcome this, in our normal form, we simply include all possible faults.
Therefore, if two sets of atomic faults generate the same faults, we must end up with the same fault gadgets.

We start locally by producing combinations of two fault gadgets:
\begin{proposition}
    \label{prop:combining-fault-gadgets}
 Given two fault gadgets, we can fault-equivalently introduce their combination into the diagram using the fault-equivalent axioms.
\end{proposition}
\begin{proof}
 We have:
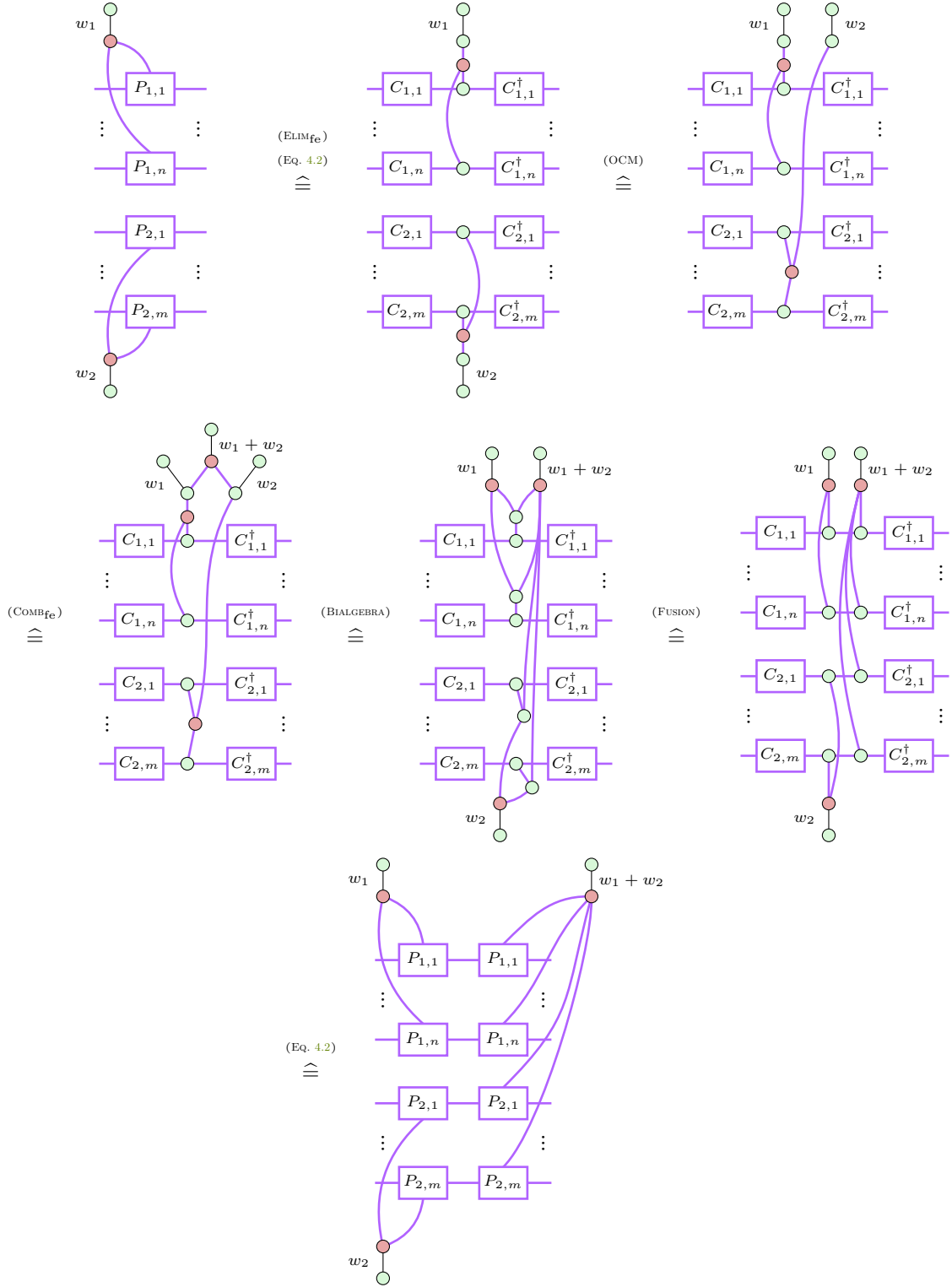
\[\begin{tikzpicture}
	\begin{pgfonlayer}{nodelayer}
		\node [style=none] (1) at (7, 3) {};
		\node [style=none] (2) at (13.25, 3) {};
		\node [style=none] (3) at (7, 0.5) {};
		\node [style=none] (4) at (13.25, 0.5) {};
		\node [style=none] (5) at (12.75, 3.5) {};
		\node [style=none] (6) at (12.75, 2.5) {};
		\node [style=none] (7) at (11.25, 2.5) {};
		\node [style=none] (8) at (11.25, 3.5) {};
		\node [style=none] (9) at (9, 3.5) {};
		\node [style=none] (10) at (9, 2.5) {};
		\node [style=none] (11) at (7.5, 2.5) {};
		\node [style=none] (12) at (7.5, 3.5) {};
		\node [style=none] (13) at (7.5, 3) {};
		\node [style=none] (14) at (9, 3) {};
		\node [style=none] (15) at (11.25, 3) {};
		\node [style=none] (16) at (12.75, 3) {};
		\node [style=none] (17) at (9, 1) {};
		\node [style=none] (18) at (9, 0) {};
		\node [style=none] (19) at (7.5, 0) {};
		\node [style=none] (20) at (7.5, 1) {};
		\node [style=none] (21) at (7.5, 0.5) {};
		\node [style=none] (22) at (9, 0.5) {};
		\node [style=none] (23) at (12.75, 1) {};
		\node [style=none] (24) at (12.75, 0) {};
		\node [style=none] (25) at (11.25, 0) {};
		\node [style=none] (26) at (11.25, 1) {};
		\node [style=none] (27) at (11.25, 0.5) {};
		\node [style=none] (28) at (12.75, 0.5) {};
		\node [style=Z dot] (29) at (10, 5.5) {};
		\node [style=sLabel] (30) at (9.25, 5) {$w_1$};
		\node [style=none, rotate=90] (31) at (7.25, 1.75) {...};
		\node [style=none] (32) at (7, -1.5) {};
		\node [style=none] (33) at (13.25, -1.5) {};
		\node [style=none] (34) at (9, -1) {};
		\node [style=none] (35) at (9, -2) {};
		\node [style=none] (36) at (7.5, -2) {};
		\node [style=none] (37) at (7.5, -1) {};
		\node [style=none] (38) at (7.5, -1.5) {};
		\node [style=none] (39) at (9, -1.5) {};
		\node [style=none] (40) at (12.75, -1) {};
		\node [style=none] (41) at (12.75, -2) {};
		\node [style=none] (42) at (11.25, -2) {};
		\node [style=none] (43) at (11.25, -1) {};
		\node [style=none] (44) at (11.25, -1.5) {};
		\node [style=none] (45) at (12.75, -1.5) {};
		\node [style=none] (46) at (7, -4) {};
		\node [style=none] (47) at (13.25, -4) {};
		\node [style=none] (48) at (9, -3.5) {};
		\node [style=none] (49) at (9, -4.5) {};
		\node [style=none] (50) at (7.5, -4.5) {};
		\node [style=none] (51) at (7.5, -3.5) {};
		\node [style=none] (52) at (7.5, -4) {};
		\node [style=none] (53) at (9, -4) {};
		\node [style=none] (54) at (12.75, -3.5) {};
		\node [style=none] (55) at (12.75, -4.5) {};
		\node [style=none] (56) at (11.25, -4.5) {};
		\node [style=none] (57) at (11.25, -3.5) {};
		\node [style=none] (58) at (11.25, -4) {};
		\node [style=none] (59) at (12.75, -4) {};
		\node [style=none, rotate=90] (60) at (7.25, -2.75) {...};
		\node [style=sLabel] (61) at (12.25, 5) {$w_2$};
		\node [style=Z dot] (62) at (10, -1.5) {};
		\node [style=Z dot] (63) at (10, -4) {};
		\node [style=Z dot] (64) at (10, 0.5) {};
		\node [style=Z dot] (65) at (10, 3) {};
		\node [style=X dot] (66) at (10, 3.75) {};
		\node [style=none, rotate=90] (67) at (13, -2.75) {...};
		\node [style=none, rotate=90] (68) at (13, 1.75) {...};
		\node [style=Z dot] (69) at (10, 4.5) {};
		\node [style=Z dot] (70) at (11.5, 4.5) {};
		\node [style=X dot] (71) at (10.25, -2.75) {};
		\node [style=Z dot] (72) at (11.5, 5.5) {};
		\node [style=none] (73) at (-3, 3) {};
		\node [style=none] (74) at (3, 3) {};
		\node [style=none] (75) at (-3, 0.5) {};
		\node [style=none] (76) at (3, 0.5) {};
		\node [style=none] (77) at (2.5, 3.5) {};
		\node [style=none] (78) at (2.5, 2.5) {};
		\node [style=none] (79) at (1, 2.5) {};
		\node [style=none] (80) at (1, 3.5) {};
		\node [style=none] (81) at (-1, 3.5) {};
		\node [style=none] (82) at (-1, 2.5) {};
		\node [style=none] (83) at (-2.5, 2.5) {};
		\node [style=none] (84) at (-2.5, 3.5) {};
		\node [style=none] (85) at (-2.5, 3) {};
		\node [style=none] (86) at (-1, 3) {};
		\node [style=none] (87) at (1, 3) {};
		\node [style=none] (88) at (2.5, 3) {};
		\node [style=none] (89) at (-1, 1) {};
		\node [style=none] (90) at (-1, 0) {};
		\node [style=none] (91) at (-2.5, 0) {};
		\node [style=none] (92) at (-2.5, 1) {};
		\node [style=none] (93) at (-2.5, 0.5) {};
		\node [style=none] (94) at (-1, 0.5) {};
		\node [style=none] (95) at (2.5, 1) {};
		\node [style=none] (96) at (2.5, 0) {};
		\node [style=none] (97) at (1, 0) {};
		\node [style=none] (98) at (1, 1) {};
		\node [style=none] (99) at (1, 0.5) {};
		\node [style=none] (100) at (2.5, 0.5) {};
		\node [style=Z dot] (101) at (0, 5.5) {};
		\node [style=sLabel] (102) at (-0.75, 5) {$w_1$};
		\node [style=none, rotate=90] (103) at (-2.75, 1.75) {...};
		\node [style=none] (104) at (-3, -1.5) {};
		\node [style=none] (105) at (3, -1.5) {};
		\node [style=none] (106) at (-1, -1) {};
		\node [style=none] (107) at (-1, -2) {};
		\node [style=none] (108) at (-2.5, -2) {};
		\node [style=none] (109) at (-2.5, -1) {};
		\node [style=none] (110) at (-2.5, -1.5) {};
		\node [style=none] (111) at (-1, -1.5) {};
		\node [style=none] (112) at (2.5, -1) {};
		\node [style=none] (113) at (2.5, -2) {};
		\node [style=none] (114) at (1, -2) {};
		\node [style=none] (115) at (1, -1) {};
		\node [style=none] (116) at (1, -1.5) {};
		\node [style=none] (117) at (2.5, -1.5) {};
		\node [style=none] (118) at (-3, -4) {};
		\node [style=none] (119) at (3, -4) {};
		\node [style=none] (120) at (-1, -3.5) {};
		\node [style=none] (121) at (-1, -4.5) {};
		\node [style=none] (122) at (-2.5, -4.5) {};
		\node [style=none] (123) at (-2.5, -3.5) {};
		\node [style=none] (124) at (-2.5, -4) {};
		\node [style=none] (125) at (-1, -4) {};
		\node [style=none] (126) at (2.5, -3.5) {};
		\node [style=none] (127) at (2.5, -4.5) {};
		\node [style=none] (128) at (1, -4.5) {};
		\node [style=none] (129) at (1, -3.5) {};
		\node [style=none] (130) at (1, -4) {};
		\node [style=none] (131) at (2.5, -4) {};
		\node [style=none, rotate=90] (132) at (-2.75, -2.75) {...};
		\node [style=sLabel] (133) at (0.75, -6) {$w_2$};
		\node [style=Z dot] (134) at (0, -1.5) {};
		\node [style=Z dot] (135) at (0, -4) {};
		\node [style=Z dot] (136) at (0, 0.5) {};
		\node [style=Z dot] (137) at (0, 3) {};
		\node [style=X dot] (138) at (0, 3.75) {};
		\node [style=none, rotate=90] (139) at (2.75, -2.75) {...};
		\node [style=none, rotate=90] (140) at (2.75, 1.75) {...};
		\node [style=Z dot] (141) at (0, 4.5) {};
		\node [style=Z dot] (142) at (0, -5.5) {};
		\node [style=X dot] (143) at (0, -4.75) {};
		\node [style=Z dot] (144) at (0, -6.5) {};
		\node [style=none] (145) at (5, 0) {$\FaultEq$};
		\node [style=none] (146) at (-5, 0) {$\FaultEq$};
		\node [style=none] (147) at (-5, 1.5) {$\FEElim$};
		\node [style=none] (148) at (-11.5, 3) {};
		\node [style=none] (149) at (-11.5, 0.5) {};
		\node [style=none] (150) at (-9, 3.5) {};
		\node [style=none] (151) at (-9, 2.5) {};
		\node [style=none] (152) at (-10.5, 2.5) {};
		\node [style=none] (153) at (-10.5, 3.5) {};
		\node [style=none] (154) at (-10.5, 3) {};
		\node [style=none] (155) at (-9, 3) {};
		\node [style=none] (156) at (-9, 1) {};
		\node [style=none] (157) at (-9, 0) {};
		\node [style=none] (158) at (-10.5, 0) {};
		\node [style=none] (159) at (-10.5, 1) {};
		\node [style=none] (160) at (-10.5, 0.5) {};
		\node [style=none] (161) at (-9, 0.5) {};
		\node [style=X dot] (162) at (-11, 4.5) {};
		\node [style=Z dot] (163) at (-11, 5.5) {};
		\node [style=sLabel] (164) at (-11.75, 5) {$w_1$};
		\node [style=none] (165) at (-9.75, 3.5) {};
		\node [style=none] (166) at (-9.75, 1) {};
		\node [style=none] (167) at (-11.5, -1.5) {};
		\node [style=none] (168) at (-9, -1) {};
		\node [style=none] (169) at (-9, -2) {};
		\node [style=none] (170) at (-10.5, -2) {};
		\node [style=none] (171) at (-10.5, -1) {};
		\node [style=none] (172) at (-10.5, -1.5) {};
		\node [style=none] (173) at (-9, -1.5) {};
		\node [style=none] (174) at (-9.75, -2) {};
		\node [style=none] (175) at (-11.5, -4) {};
		\node [style=none] (176) at (-9, -3.5) {};
		\node [style=none] (177) at (-9, -4.5) {};
		\node [style=none] (178) at (-10.5, -4.5) {};
		\node [style=none] (179) at (-10.5, -3.5) {};
		\node [style=none] (180) at (-10.5, -4) {};
		\node [style=none] (181) at (-9, -4) {};
		\node [style=none] (182) at (-9.75, -4.5) {};
		\node [style=X dot] (183) at (-11, -5.5) {};
		\node [style=Z dot] (184) at (-11, -6.5) {};
		\node [style=sLabel] (185) at (-11.75, -6) {$w_2$};
		\node [style=none] (186) at (-8, 3) {};
		\node [style=none] (187) at (-8, 0.5) {};
		\node [style=none] (188) at (-8, -1.5) {};
		\node [style=none] (189) at (-8, -4) {};
		\node [style=none, rotate=90] (190) at (-8.25, 1.75) {...};
		\node [style=none, rotate=90] (191) at (-8.25, -2.75) {...};
		\node [style=none, rotate=90] (192) at (-11.25, 1.75) {...};
		\node [style=none, rotate=90] (193) at (-11.25, -2.75) {...};
		\node [style=sLabel] (194) at (8.25, 3) {$C_ {1,1}$};
		\node [style=sLabel] (195) at (8.25, 0.5) {$C_{1,n}$};
		\node [style=sLabel] (196) at (8.25, -1.5) {$C_{2,1}$};
		\node [style=sLabel] (197) at (8.25, -4) {$C_{2,m}$};
		\node [style=sLabel] (198) at (12, 3) {$C_{1,1}^\dagger$};
		\node [style=sLabel] (199) at (12, 0.5) {$C_{1,n}^\dagger$};
		\node [style=sLabel] (200) at (12, -1.5) {$C_{2,1}^\dagger$};
		\node [style=sLabel] (201) at (12, -4) {$C_{2,m}^\dagger$};
		\node [style=sLabel] (202) at (-1.75, 3) {$C_ {1,1}$};
		\node [style=sLabel] (203) at (-1.75, 0.5) {$C_{1,n}$};
		\node [style=sLabel] (204) at (-1.75, -1.5) {$C_{2,1}$};
		\node [style=sLabel] (205) at (-1.75, -4) {$C_{2,m}$};
		\node [style=sLabel] (206) at (1.75, 3) {$C_{1,1}^\dagger$};
		\node [style=sLabel] (207) at (1.75, 0.5) {$C_{1,n}^\dagger$};
		\node [style=sLabel] (208) at (1.75, -1.5) {$C_{2,1}^\dagger$};
		\node [style=sLabel] (209) at (1.75, -4) {$C_{2,m}^\dagger$};
		\node [style=sLabel] (210) at (-9.75, 3) {$P_{1,1}$};
		\node [style=sLabel] (211) at (-9.75, 0.5) {$P_{1,n}$};
		\node [style=sLabel] (212) at (-9.75, -1.5) {$P_{2,1}$};
		\node [style=sLabel] (213) at (-9.75, -4) {$P_{2,m}$};
		\node [style=none] (214) at (-5, 0.75) {$\Eq{eq:pauli-box-clifford-decomposition}$};
		\node [style=none] (215) at (5, 0.75) {$\OCM$};
	\end{pgfonlayer}
	\begin{pgfonlayer}{edgelayer}
		\draw [style=fault-free] (6.center)
			 to (7.center)
			 to (8.center)
			 to (5.center)
			 to cycle;
		\draw [style=fault-free] (11.center)
			 to (12.center)
			 to (9.center)
			 to (10.center)
			 to cycle;
		\draw [style=fault-free] (20.center)
			 to (17.center)
			 to (18.center)
			 to (19.center)
			 to cycle;
		\draw [style=fault-free] (25.center)
			 to (26.center)
			 to (23.center)
			 to (24.center)
			 to cycle;
		\draw [style=fault-free] (3.center) to (21.center);
		\draw [style=fault-free] (22.center) to (27.center);
		\draw [style=fault-free] (28.center) to (4.center);
		\draw [style=fault-free] (2.center) to (16.center);
		\draw [style=fault-free] (14.center) to (15.center);
		\draw [style=fault-free] (1.center) to (13.center);
		\draw [style=fault-free] (35.center)
			 to (36.center)
			 to (37.center)
			 to (34.center)
			 to cycle;
		\draw [style=fault-free] (41.center)
			 to (42.center)
			 to (43.center)
			 to (40.center)
			 to cycle;
		\draw [style=fault-free] (32.center) to (38.center);
		\draw [style=fault-free] (39.center) to (44.center);
		\draw [style=fault-free] (45.center) to (33.center);
		\draw [style=fault-free] (51.center)
			 to (48.center)
			 to (49.center)
			 to (50.center)
			 to cycle;
		\draw [style=fault-free] (57.center)
			 to (54.center)
			 to (55.center)
			 to (56.center)
			 to cycle;
		\draw [style=fault-free] (46.center) to (52.center);
		\draw [style=fault-free] (53.center) to (58.center);
		\draw [style=fault-free] (59.center) to (47.center);
		\draw [style=fault-free] (66) to (65);
		\draw [style=fault-free, bend right] (66) to (64);
		\draw [style=fault-free] (69) to (66);
		\draw (29) to (69);
		\draw [style=fault-free] (71) to (62);
		\draw [style=fault-free] (71) to (63);
		\draw [style=fault-free, in=75, out=-120] (70) to (71);
		\draw (72) to (70);
		\draw [style=fault-free] (79.center)
			 to (80.center)
			 to (77.center)
			 to (78.center)
			 to cycle;
		\draw [style=fault-free] (81.center)
			 to (82.center)
			 to (83.center)
			 to (84.center)
			 to cycle;
		\draw [style=fault-free] (91.center)
			 to (92.center)
			 to (89.center)
			 to (90.center)
			 to cycle;
		\draw [style=fault-free] (98.center)
			 to (95.center)
			 to (96.center)
			 to (97.center)
			 to cycle;
		\draw [style=fault-free] (75.center) to (93.center);
		\draw [style=fault-free] (94.center) to (99.center);
		\draw [style=fault-free] (100.center) to (76.center);
		\draw [style=fault-free] (74.center) to (88.center);
		\draw [style=fault-free] (86.center) to (87.center);
		\draw [style=fault-free] (73.center) to (85.center);
		\draw [style=fault-free] (106.center)
			 to (107.center)
			 to (108.center)
			 to (109.center)
			 to cycle;
		\draw [style=fault-free] (112.center)
			 to (113.center)
			 to (114.center)
			 to (115.center)
			 to cycle;
		\draw [style=fault-free] (104.center) to (110.center);
		\draw [style=fault-free] (111.center) to (116.center);
		\draw [style=fault-free] (117.center) to (105.center);
		\draw [style=fault-free] (120.center)
			 to (121.center)
			 to (122.center)
			 to (123.center)
			 to cycle;
		\draw [style=fault-free] (129.center)
			 to (126.center)
			 to (127.center)
			 to (128.center)
			 to cycle;
		\draw [style=fault-free] (118.center) to (124.center);
		\draw [style=fault-free] (125.center) to (130.center);
		\draw [style=fault-free] (131.center) to (119.center);
		\draw [style=fault-free] (138) to (137);
		\draw [style=fault-free, bend right] (138) to (136);
		\draw [style=fault-free] (141) to (138);
		\draw (101) to (141);
		\draw [style=fault-free, bend right] (143) to (134);
		\draw [style=fault-free] (143) to (135);
		\draw [style=fault-free] (142) to (143);
		\draw (144) to (142);
		\draw [style=fault-free] (152.center)
			 to (153.center)
			 to (150.center)
			 to (151.center)
			 to cycle;
		\draw [style=fault-free] (157.center)
			 to (158.center)
			 to (159.center)
			 to (156.center)
			 to cycle;
		\draw [style=fault-free] (149.center) to (160.center);
		\draw [style=fault-free] (148.center) to (154.center);
		\draw [style=fault-free, bend left] (162) to (165.center);
		\draw [style=fault-free, bend right] (162) to (166.center);
		\draw (163) to (162);
		\draw [style=fault-free] (169.center)
			 to (170.center)
			 to (171.center)
			 to (168.center)
			 to cycle;
		\draw [style=fault-free] (167.center) to (172.center);
		\draw [style=fault-free] (177.center)
			 to (178.center)
			 to (179.center)
			 to (176.center)
			 to cycle;
		\draw [style=fault-free] (175.center) to (180.center);
		\draw (184) to (183);
		\draw [style=fault-free, bend right] (183) to (182.center);
		\draw [style=fault-free, bend right] (174.center) to (183);
		\draw [style=fault-free] (189.center) to (181.center);
		\draw [style=fault-free] (173.center) to (188.center);
		\draw [style=fault-free] (187.center) to (161.center);
		\draw [style=fault-free] (186.center) to (155.center);
	\end{pgfonlayer}
\end{tikzpicture}
\]
\[\begin{tikzpicture}
	\begin{pgfonlayer}{nodelayer}
		\node [style=none] (0) at (7, 3) {};
		\node [style=none] (1) at (13.5, 3) {};
		\node [style=none] (2) at (7, 0.5) {};
		\node [style=none] (3) at (13.5, 0.5) {};
		\node [style=none] (4) at (13, 3.5) {};
		\node [style=none] (5) at (13, 2.5) {};
		\node [style=none] (6) at (11.5, 2.5) {};
		\node [style=none] (7) at (11.5, 3.5) {};
		\node [style=none] (8) at (9, 3.5) {};
		\node [style=none] (9) at (9, 2.5) {};
		\node [style=none] (10) at (7.5, 2.5) {};
		\node [style=none] (11) at (7.5, 3.5) {};
		\node [style=none] (12) at (7.5, 3) {};
		\node [style=none] (13) at (9, 3) {};
		\node [style=none] (14) at (11.5, 3) {};
		\node [style=none] (15) at (13, 3) {};
		\node [style=none] (16) at (9, 1) {};
		\node [style=none] (17) at (9, 0) {};
		\node [style=none] (18) at (7.5, 0) {};
		\node [style=none] (19) at (7.5, 1) {};
		\node [style=none] (20) at (7.5, 0.5) {};
		\node [style=none] (21) at (9, 0.5) {};
		\node [style=none] (22) at (13, 1) {};
		\node [style=none] (23) at (13, 0) {};
		\node [style=none] (24) at (11.5, 0) {};
		\node [style=none] (25) at (11.5, 1) {};
		\node [style=none] (26) at (11.5, 0.5) {};
		\node [style=none] (27) at (13, 0.5) {};
		\node [style=X dot] (28) at (9.75, 4.5) {};
		\node [style=X dot] (29) at (10.75, 4.5) {};
		\node [style=Z dot] (30) at (9.75, 5.5) {};
		\node [style=Z dot] (31) at (10.75, 5.5) {};
		\node [style=sLabel] (32) at (12, 5) {$w_1+w_2$};
		\node [style=sLabel] (33) at (9, 5) {$w_1$};
		\node [style=sLabel] (34) at (8.25, 3) {$C_ {1,1}$};
		\node [style=sLabel] (35) at (12.25, 3) {$C_{1,1}^\dagger$};
		\node [style=sLabel] (36) at (8.25, 0.5) {$C_{1,n}$};
		\node [style=sLabel] (37) at (12.25, 0.5) {$C_{1,n}^\dagger$};
		\node [style=none, rotate=90] (38) at (7.25, 1.75) {...};
		\node [style=none] (39) at (7, -1.5) {};
		\node [style=none] (40) at (13.5, -1.5) {};
		\node [style=none] (41) at (9, -1) {};
		\node [style=none] (42) at (9, -2) {};
		\node [style=none] (43) at (7.5, -2) {};
		\node [style=none] (44) at (7.5, -1) {};
		\node [style=none] (45) at (7.5, -1.5) {};
		\node [style=none] (46) at (9, -1.5) {};
		\node [style=none] (47) at (13, -1) {};
		\node [style=none] (48) at (13, -2) {};
		\node [style=none] (49) at (11.5, -2) {};
		\node [style=none] (50) at (11.5, -1) {};
		\node [style=none] (51) at (11.5, -1.5) {};
		\node [style=none] (52) at (13, -1.5) {};
		\node [style=sLabel] (53) at (8.25, -1.5) {$C_{2,1}$};
		\node [style=sLabel] (54) at (12.25, -1.5) {$C_{2,1}^\dagger$};
		\node [style=none] (55) at (7, -4) {};
		\node [style=none] (56) at (13.5, -4) {};
		\node [style=none] (57) at (9, -3.5) {};
		\node [style=none] (58) at (9, -4.5) {};
		\node [style=none] (59) at (7.5, -4.5) {};
		\node [style=none] (60) at (7.5, -3.5) {};
		\node [style=none] (61) at (7.5, -4) {};
		\node [style=none] (62) at (9, -4) {};
		\node [style=none] (63) at (13, -3.5) {};
		\node [style=none] (64) at (13, -4.5) {};
		\node [style=none] (65) at (11.5, -4.5) {};
		\node [style=none] (66) at (11.5, -3.5) {};
		\node [style=none] (67) at (11.5, -4) {};
		\node [style=none] (68) at (13, -4) {};
		\node [style=sLabel] (69) at (8.25, -4) {$C_{2,m}$};
		\node [style=sLabel] (70) at (12.25, -4) {$C_{2,m}^\dagger$};
		\node [style=none, rotate=90] (71) at (13.25, 1.75) {...};
		\node [style=X dot] (72) at (9.75, -5.5) {};
		\node [style=Z dot] (73) at (9.75, -6.5) {};
		\node [style=sLabel] (74) at (9, -6) {$w_2$};
		\node [style=Z dot] (75) at (9.75, -4) {};
		\node [style=Z dot] (76) at (9.75, -1.5) {};
		\node [style=Z dot] (77) at (9.75, 0.5) {};
		\node [style=Z dot] (78) at (9.75, 3) {};
		\node [style=Z dot] (79) at (10.75, 3) {};
		\node [style=Z dot] (80) at (10.75, 0.5) {};
		\node [style=Z dot] (81) at (10.75, -1.5) {};
		\node [style=Z dot] (82) at (10.75, -4) {};
		\node [style=none] (83) at (5, -0.25) {$\FaultEq$};
		\node [style=none] (84) at (-3, 2.75) {};
		\node [style=none] (85) at (3, 2.75) {};
		\node [style=none] (86) at (-3, 0.25) {};
		\node [style=none] (87) at (3, 0.25) {};
		\node [style=none] (88) at (2.5, 3.25) {};
		\node [style=none] (89) at (2.5, 2.25) {};
		\node [style=none] (90) at (1, 2.25) {};
		\node [style=none] (91) at (1, 3.25) {};
		\node [style=none] (92) at (-1, 3.25) {};
		\node [style=none] (93) at (-1, 2.25) {};
		\node [style=none] (94) at (-2.5, 2.25) {};
		\node [style=none] (95) at (-2.5, 3.25) {};
		\node [style=none] (96) at (-2.5, 2.75) {};
		\node [style=none] (97) at (-1, 2.75) {};
		\node [style=none] (98) at (1, 2.75) {};
		\node [style=none] (99) at (2.5, 2.75) {};
		\node [style=none] (100) at (-1, 0.75) {};
		\node [style=none] (101) at (-1, -0.25) {};
		\node [style=none] (102) at (-2.5, -0.25) {};
		\node [style=none] (103) at (-2.5, 0.75) {};
		\node [style=none] (104) at (-2.5, 0.25) {};
		\node [style=none] (105) at (-1, 0.25) {};
		\node [style=none] (106) at (2.5, 0.75) {};
		\node [style=none] (107) at (2.5, -0.25) {};
		\node [style=none] (108) at (1, -0.25) {};
		\node [style=none] (109) at (1, 0.75) {};
		\node [style=none] (110) at (1, 0.25) {};
		\node [style=none] (111) at (2.5, 0.25) {};
		\node [style=X dot] (112) at (-0.75, 4.5) {};
		\node [style=X dot] (113) at (0.75, 4.5) {};
		\node [style=Z dot] (114) at (-0.75, 5.5) {};
		\node [style=Z dot] (115) at (0.75, 5.5) {};
		\node [style=sLabel] (116) at (2, 5) {$w_1+w_2$};
		\node [style=sLabel] (117) at (-1.5, 5) {$w_1$};
		\node [style=none] (118) at (-3, -1.75) {};
		\node [style=none] (119) at (3, -1.75) {};
		\node [style=none] (120) at (-1, -1.25) {};
		\node [style=none] (121) at (-1, -2.25) {};
		\node [style=none] (122) at (-2.5, -2.25) {};
		\node [style=none] (123) at (-2.5, -1.25) {};
		\node [style=none] (124) at (-2.5, -1.75) {};
		\node [style=none] (125) at (-1, -1.75) {};
		\node [style=none] (126) at (2.5, -1.25) {};
		\node [style=none] (127) at (2.5, -2.25) {};
		\node [style=none] (128) at (1, -2.25) {};
		\node [style=none] (129) at (1, -1.25) {};
		\node [style=none] (130) at (1, -1.75) {};
		\node [style=none] (131) at (2.5, -1.75) {};
		\node [style=none] (132) at (-3, -4.25) {};
		\node [style=none] (133) at (3, -4.25) {};
		\node [style=none] (134) at (-1, -3.75) {};
		\node [style=none] (135) at (-1, -4.75) {};
		\node [style=none] (136) at (-2.5, -4.75) {};
		\node [style=none] (137) at (-2.5, -3.75) {};
		\node [style=none] (138) at (-2.5, -4.25) {};
		\node [style=none] (139) at (-1, -4.25) {};
		\node [style=none] (140) at (2.5, -3.75) {};
		\node [style=none] (141) at (2.5, -4.75) {};
		\node [style=none] (142) at (1, -4.75) {};
		\node [style=none] (143) at (1, -3.75) {};
		\node [style=none] (144) at (1, -4.25) {};
		\node [style=none] (145) at (2.5, -4.25) {};
		\node [style=X dot] (146) at (-0.5, -5.5) {};
		\node [style=Z dot] (147) at (-0.5, -6.5) {};
		\node [style=sLabel] (148) at (-1.25, -6) {$w_2$};
		\node [style=Z dot] (149) at (0, -1.75) {};
		\node [style=Z dot] (150) at (0, -4.25) {};
		\node [style=Z dot] (151) at (0.5, -5) {};
		\node [style=Z dot] (152) at (0.25, -2.75) {};
		\node [style=Z dot] (153) at (0, 0.25) {};
		\node [style=Z dot] (154) at (0, 1) {};
		\node [style=Z dot] (155) at (0, 2.75) {};
		\node [style=Z dot] (156) at (0, 3.5) {};
		\node [style=none] (157) at (-5, -0.25) {$\FaultEq$};
		\node [style=none] (158) at (-5, 0.5) {$\Bialgebra$};
		\node [style=none] (159) at (-13, 2.75) {};
		\node [style=none] (160) at (-7, 2.75) {};
		\node [style=none] (161) at (-13, 0.25) {};
		\node [style=none] (162) at (-7, 0.25) {};
		\node [style=none] (163) at (-7.5, 3.25) {};
		\node [style=none] (164) at (-7.5, 2.25) {};
		\node [style=none] (165) at (-9, 2.25) {};
		\node [style=none] (166) at (-9, 3.25) {};
		\node [style=none] (167) at (-11, 3.25) {};
		\node [style=none] (168) at (-11, 2.25) {};
		\node [style=none] (169) at (-12.5, 2.25) {};
		\node [style=none] (170) at (-12.5, 3.25) {};
		\node [style=none] (171) at (-12.5, 2.75) {};
		\node [style=none] (172) at (-11, 2.75) {};
		\node [style=none] (173) at (-9, 2.75) {};
		\node [style=none] (174) at (-7.5, 2.75) {};
		\node [style=none] (175) at (-11, 0.75) {};
		\node [style=none] (176) at (-11, -0.25) {};
		\node [style=none] (177) at (-12.5, -0.25) {};
		\node [style=none] (178) at (-12.5, 0.75) {};
		\node [style=none] (179) at (-12.5, 0.25) {};
		\node [style=none] (180) at (-11, 0.25) {};
		\node [style=none] (181) at (-7.5, 0.75) {};
		\node [style=none] (182) at (-7.5, -0.25) {};
		\node [style=none] (183) at (-9, -0.25) {};
		\node [style=none] (184) at (-9, 0.75) {};
		\node [style=none] (185) at (-9, 0.25) {};
		\node [style=none] (186) at (-7.5, 0.25) {};
		\node [style=X dot] (187) at (-9.5, 5.25) {};
		\node [style=Z dot] (188) at (-11, 5.25) {};
		\node [style=Z dot] (189) at (-9.5, 6.25) {};
		\node [style=sLabel] (190) at (-8.25, 5.75) {$w_1+w_2$};
		\node [style=sLabel] (191) at (-11.25, 4.5) {$w_1$};
		\node [style=none, rotate=90] (192) at (-12.75, 1.5) {...};
		\node [style=none] (193) at (-13, -1.75) {};
		\node [style=none] (194) at (-7, -1.75) {};
		\node [style=none] (195) at (-11, -1.25) {};
		\node [style=none] (196) at (-11, -2.25) {};
		\node [style=none] (197) at (-12.5, -2.25) {};
		\node [style=none] (198) at (-12.5, -1.25) {};
		\node [style=none] (199) at (-12.5, -1.75) {};
		\node [style=none] (200) at (-11, -1.75) {};
		\node [style=none] (201) at (-7.5, -1.25) {};
		\node [style=none] (202) at (-7.5, -2.25) {};
		\node [style=none] (203) at (-9, -2.25) {};
		\node [style=none] (204) at (-9, -1.25) {};
		\node [style=none] (205) at (-9, -1.75) {};
		\node [style=none] (206) at (-7.5, -1.75) {};
		\node [style=none] (207) at (-13, -4.25) {};
		\node [style=none] (208) at (-7, -4.25) {};
		\node [style=none] (209) at (-11, -3.75) {};
		\node [style=none] (210) at (-11, -4.75) {};
		\node [style=none] (211) at (-12.5, -4.75) {};
		\node [style=none] (212) at (-12.5, -3.75) {};
		\node [style=none] (213) at (-12.5, -4.25) {};
		\node [style=none] (214) at (-11, -4.25) {};
		\node [style=none] (215) at (-7.5, -3.75) {};
		\node [style=none] (216) at (-7.5, -4.75) {};
		\node [style=none] (217) at (-9, -4.75) {};
		\node [style=none] (218) at (-9, -3.75) {};
		\node [style=none] (219) at (-9, -4.25) {};
		\node [style=none] (220) at (-7.5, -4.25) {};
		\node [style=none, rotate=90] (221) at (-12.75, -3) {...};
		\node [style=sLabel] (222) at (-7.75, 4.5) {$w_2$};
		\node [style=Z dot] (223) at (-10.25, -1.75) {};
		\node [style=Z dot] (224) at (-10.25, -4.25) {};
		\node [style=Z dot] (225) at (-10.25, 0.25) {};
		\node [style=Z dot] (226) at (-10.25, 2.75) {};
		\node [style=X dot] (227) at (-10.25, 3.5) {};
		\node [style=none, rotate=90] (228) at (-7.25, -3) {...};
		\node [style=none, rotate=90] (229) at (-7.25, 1.5) {...};
		\node [style=Z dot] (230) at (-10.25, 4.25) {};
		\node [style=Z dot] (231) at (-8.75, 4.25) {};
		\node [style=X dot] (232) at (-10, -3) {};
		\node [style=Z dot] (233) at (-8, 5.25) {};
		\node [style=none] (234) at (-15, -0.25) {$\FaultEq$};
		\node [style=none] (235) at (-15, 0.5) {$\FEComb$};
		\node [style=none, rotate=90] (236) at (7.25, -2.75) {...};
		\node [style=none, rotate=90] (237) at (13.25, -2.75) {...};
		\node [style=none, rotate=90] (238) at (-2.75, -3) {...};
		\node [style=none, rotate=90] (239) at (-2.75, 1.5) {...};
		\node [style=none, rotate=90] (240) at (2.75, 1.5) {...};
		\node [style=none, rotate=90] (241) at (2.75, -3) {...};
		\node [style=sLabel] (242) at (-1.75, 2.75) {$C_ {1,1}$};
		\node [style=sLabel] (243) at (-1.75, 0.25) {$C_{1,n}$};
		\node [style=sLabel] (244) at (-1.75, -1.75) {$C_{2,1}$};
		\node [style=sLabel] (245) at (-1.75, -4.25) {$C_{2,m}$};
		\node [style=sLabel] (246) at (1.75, 2.75) {$C_{1,1}^\dagger$};
		\node [style=sLabel] (247) at (1.75, 0.25) {$C_{1,n}^\dagger$};
		\node [style=sLabel] (248) at (1.75, -1.75) {$C_{2,1}^\dagger$};
		\node [style=sLabel] (249) at (1.75, -4.25) {$C_{2,m}^\dagger$};
		\node [style=sLabel] (250) at (-11.75, 2.75) {$C_ {1,1}$};
		\node [style=sLabel] (251) at (-11.75, 0.25) {$C_{1,n}$};
		\node [style=sLabel] (252) at (-11.75, -1.75) {$C_{2,1}$};
		\node [style=sLabel] (253) at (-11.75, -4.25) {$C_{2,m}$};
		\node [style=sLabel] (254) at (-8.25, 2.75) {$C_{1,1}^\dagger$};
		\node [style=sLabel] (255) at (-8.25, 0.25) {$C_{1,n}^\dagger$};
		\node [style=sLabel] (256) at (-8.25, -1.75) {$C_{2,1}^\dagger$};
		\node [style=sLabel] (257) at (-8.25, -4.25) {$C_{2,m}^\dagger$};
		\node [style=none] (258) at (5, 0.5) {$\Fusion$};
	\end{pgfonlayer}
	\begin{pgfonlayer}{edgelayer}
		\draw [style=fault-free] (6.center)
			 to (7.center)
			 to (4.center)
			 to (5.center)
			 to cycle;
		\draw [style=fault-free] (10.center)
			 to (11.center)
			 to (8.center)
			 to (9.center)
			 to cycle;
		\draw [style=fault-free] (17.center)
			 to (18.center)
			 to (19.center)
			 to (16.center)
			 to cycle;
		\draw [style=fault-free] (24.center)
			 to (25.center)
			 to (22.center)
			 to (23.center)
			 to cycle;
		\draw [style=fault-free] (2.center) to (20.center);
		\draw [style=fault-free] (21.center) to (26.center);
		\draw [style=fault-free] (27.center) to (3.center);
		\draw [style=fault-free] (1.center) to (15.center);
		\draw [style=fault-free] (13.center) to (14.center);
		\draw [style=fault-free] (0.center) to (12.center);
		\draw (30) to (28);
		\draw (31) to (29);
		\draw [style=fault-free] (44.center)
			 to (41.center)
			 to (42.center)
			 to (43.center)
			 to cycle;
		\draw [style=fault-free] (47.center)
			 to (48.center)
			 to (49.center)
			 to (50.center)
			 to cycle;
		\draw [style=fault-free] (39.center) to (45.center);
		\draw [style=fault-free] (46.center) to (51.center);
		\draw [style=fault-free] (52.center) to (40.center);
		\draw [style=fault-free] (59.center)
			 to (60.center)
			 to (57.center)
			 to (58.center)
			 to cycle;
		\draw [style=fault-free] (66.center)
			 to (63.center)
			 to (64.center)
			 to (65.center)
			 to cycle;
		\draw [style=fault-free] (55.center) to (61.center);
		\draw [style=fault-free] (62.center) to (67.center);
		\draw [style=fault-free] (68.center) to (56.center);
		\draw (73) to (72);
		\draw [style=fault-free] (72) to (75);
		\draw [style=fault-free, bend right=15] (72) to (76);
		\draw [style=fault-free] (28) to (78);
		\draw [style=fault-free, bend right=15] (28) to (77);
		\draw [style=fault-free] (29) to (79);
		\draw [style=fault-free, bend right=15] (29) to (80);
		\draw [style=fault-free, bend right=15] (29) to (81);
		\draw [style=fault-free, bend left=15] (82) to (29);
		\draw [style=fault-free] (89.center)
			 to (90.center)
			 to (91.center)
			 to (88.center)
			 to cycle;
		\draw [style=fault-free] (94.center)
			 to (95.center)
			 to (92.center)
			 to (93.center)
			 to cycle;
		\draw [style=fault-free] (101.center)
			 to (102.center)
			 to (103.center)
			 to (100.center)
			 to cycle;
		\draw [style=fault-free] (106.center)
			 to (107.center)
			 to (108.center)
			 to (109.center)
			 to cycle;
		\draw [style=fault-free] (86.center) to (104.center);
		\draw [style=fault-free] (105.center) to (110.center);
		\draw [style=fault-free] (111.center) to (87.center);
		\draw [style=fault-free] (85.center) to (99.center);
		\draw [style=fault-free] (97.center) to (98.center);
		\draw [style=fault-free] (84.center) to (96.center);
		\draw (114) to (112);
		\draw (115) to (113);
		\draw [style=fault-free] (120.center)
			 to (121.center)
			 to (122.center)
			 to (123.center)
			 to cycle;
		\draw [style=fault-free] (126.center)
			 to (127.center)
			 to (128.center)
			 to (129.center)
			 to cycle;
		\draw [style=fault-free] (118.center) to (124.center);
		\draw [style=fault-free] (125.center) to (130.center);
		\draw [style=fault-free] (131.center) to (119.center);
		\draw [style=fault-free] (134.center)
			 to (135.center)
			 to (136.center)
			 to (137.center)
			 to cycle;
		\draw [style=fault-free] (141.center)
			 to (142.center)
			 to (143.center)
			 to (140.center)
			 to cycle;
		\draw [style=fault-free] (132.center) to (138.center);
		\draw [style=fault-free] (139.center) to (144.center);
		\draw [style=fault-free] (145.center) to (133.center);
		\draw (147) to (146);
		\draw [style=fault-free] (156) to (155);
		\draw [style=fault-free] (154) to (153);
		\draw [style=fault-free] (149) to (152);
		\draw [style=fault-free] (150) to (151);
		\draw [style=fault-free, bend right=15] (146) to (151);
		\draw [style=fault-free, bend left=15] (146) to (152);
		\draw [style=fault-free, bend left=15] (112) to (156);
		\draw [style=fault-free, bend right=15] (112) to (154);
		\draw [style=fault-free, bend right=15] (113) to (156);
		\draw [style=fault-free, in=60, out=-90] (113) to (154);
		\draw [style=fault-free, in=270, out=90] (151) to (113);
		\draw [style=fault-free, in=270, out=90] (152) to (113);
		\draw [style=fault-free] (164.center)
			 to (165.center)
			 to (166.center)
			 to (163.center)
			 to cycle;
		\draw [style=fault-free] (170.center)
			 to (167.center)
			 to (168.center)
			 to (169.center)
			 to cycle;
		\draw [style=fault-free] (176.center)
			 to (177.center)
			 to (178.center)
			 to (175.center)
			 to cycle;
		\draw [style=fault-free] (183.center)
			 to (184.center)
			 to (181.center)
			 to (182.center)
			 to cycle;
		\draw [style=fault-free] (161.center) to (179.center);
		\draw [style=fault-free] (180.center) to (185.center);
		\draw [style=fault-free] (186.center) to (162.center);
		\draw [style=fault-free] (160.center) to (174.center);
		\draw [style=fault-free] (172.center) to (173.center);
		\draw [style=fault-free] (159.center) to (171.center);
		\draw (189) to (187);
		\draw [style=fault-free] (196.center)
			 to (197.center)
			 to (198.center)
			 to (195.center)
			 to cycle;
		\draw [style=fault-free] (201.center)
			 to (202.center)
			 to (203.center)
			 to (204.center)
			 to cycle;
		\draw [style=fault-free] (193.center) to (199.center);
		\draw [style=fault-free] (200.center) to (205.center);
		\draw [style=fault-free] (206.center) to (194.center);
		\draw [style=fault-free] (210.center)
			 to (211.center)
			 to (212.center)
			 to (209.center)
			 to cycle;
		\draw [style=fault-free] (217.center)
			 to (218.center)
			 to (215.center)
			 to (216.center)
			 to cycle;
		\draw [style=fault-free] (207.center) to (213.center);
		\draw [style=fault-free] (214.center) to (219.center);
		\draw [style=fault-free] (220.center) to (208.center);
		\draw [style=fault-free] (227) to (226);
		\draw [style=fault-free, bend right] (227) to (225);
		\draw [style=fault-free] (230) to (227);
		\draw [style=fault-free] (230) to (187);
		\draw (188) to (230);
		\draw [style=fault-free] (232) to (223);
		\draw [style=fault-free] (232) to (224);
		\draw [style=fault-free, in=75, out=-120] (231) to (232);
		\draw [style=fault-free] (231) to (187);
		\draw (233) to (231);
	\end{pgfonlayer}
\end{tikzpicture}
\]
\[\begin{tikzpicture}
	\begin{pgfonlayer}{nodelayer}
		\node [style=none] (0) at (-5, 3) {};
		\node [style=none] (1) at (0.5, 3) {};
		\node [style=none] (2) at (-5, 0.5) {};
		\node [style=none] (3) at (0.5, 0.5) {};
		\node [style=none] (4) at (-0.25, 3.5) {};
		\node [style=none] (5) at (-0.25, 2.5) {};
		\node [style=none] (6) at (-1.75, 2.5) {};
		\node [style=none] (7) at (-1.75, 3.5) {};
		\node [style=none] (8) at (-2.75, 3.5) {};
		\node [style=none] (9) at (-2.75, 2.5) {};
		\node [style=none] (10) at (-4.25, 2.5) {};
		\node [style=none] (11) at (-4.25, 3.5) {};
		\node [style=none] (12) at (-4.25, 3) {};
		\node [style=none] (13) at (-2.75, 3) {};
		\node [style=none] (14) at (-1.75, 3) {};
		\node [style=none] (15) at (-0.25, 3) {};
		\node [style=none] (16) at (-2.75, 1) {};
		\node [style=none] (17) at (-2.75, 0) {};
		\node [style=none] (18) at (-4.25, 0) {};
		\node [style=none] (19) at (-4.25, 1) {};
		\node [style=none] (20) at (-4.25, 0.5) {};
		\node [style=none] (21) at (-2.75, 0.5) {};
		\node [style=none] (22) at (-0.25, 1) {};
		\node [style=none] (23) at (-0.25, 0) {};
		\node [style=none] (24) at (-1.75, 0) {};
		\node [style=none] (25) at (-1.75, 1) {};
		\node [style=none] (26) at (-1.75, 0.5) {};
		\node [style=none] (27) at (-0.25, 0.5) {};
		\node [style=X dot] (28) at (-4.75, 5) {};
		\node [style=X dot] (29) at (1.75, 5) {};
		\node [style=Z dot] (30) at (-4.75, 6) {};
		\node [style=Z dot] (31) at (1.75, 6) {};
		\node [style=sLabel] (32) at (3, 5.5) {$w_1+w_2$};
		\node [style=sLabel] (33) at (-5.5, 5.5) {$w_1$};
		\node [style=none] (34) at (-3.5, 3.5) {};
		\node [style=none] (35) at (-3.5, 1) {};
		\node [style=none] (36) at (-1, 3.5) {};
		\node [style=none] (37) at (-1, 1) {};
		\node [style=sLabel] (38) at (-1, 3) {$P_{1,1}$};
		\node [style=sLabel] (39) at (-1, 0.5) {$P_{1,n}$};
		\node [style=none, rotate=90] (40) at (-4.75, 1.75) {...};
		\node [style=none] (41) at (-7, -0.5) {$\FaultEq$};
		\node [style=none] (42) at (-5, -1.5) {};
		\node [style=none] (43) at (0.5, -1.5) {};
		\node [style=none] (44) at (-2.75, -1) {};
		\node [style=none] (45) at (-2.75, -2) {};
		\node [style=none] (46) at (-4.25, -2) {};
		\node [style=none] (47) at (-4.25, -1) {};
		\node [style=none] (48) at (-4.25, -1.5) {};
		\node [style=none] (49) at (-2.75, -1.5) {};
		\node [style=none] (50) at (-0.25, -1) {};
		\node [style=none] (51) at (-0.25, -2) {};
		\node [style=none] (52) at (-1.75, -2) {};
		\node [style=none] (53) at (-1.75, -1) {};
		\node [style=none] (54) at (-1.75, -1.5) {};
		\node [style=none] (55) at (-0.25, -1.5) {};
		\node [style=none] (56) at (-3.5, -2) {};
		\node [style=none] (57) at (-1, -1) {};
		\node [style=sLabel] (58) at (-1, -1.5) {$P_{2,1}$};
		\node [style=none] (59) at (-5, -4) {};
		\node [style=none] (60) at (0.5, -4) {};
		\node [style=none] (61) at (-2.75, -3.5) {};
		\node [style=none] (62) at (-2.75, -4.5) {};
		\node [style=none] (63) at (-4.25, -4.5) {};
		\node [style=none] (64) at (-4.25, -3.5) {};
		\node [style=none] (65) at (-4.25, -4) {};
		\node [style=none] (66) at (-2.75, -4) {};
		\node [style=none] (67) at (-0.25, -3.5) {};
		\node [style=none] (68) at (-0.25, -4.5) {};
		\node [style=none] (69) at (-1.75, -4.5) {};
		\node [style=none] (70) at (-1.75, -3.5) {};
		\node [style=none] (71) at (-1.75, -4) {};
		\node [style=none] (72) at (-0.25, -4) {};
		\node [style=none] (73) at (-3.5, -4.5) {};
		\node [style=none] (74) at (-1, -3.5) {};
		\node [style=sLabel] (75) at (-1, -4) {$P_{2,m}$};
		\node [style=X dot] (76) at (-4.75, -6) {};
		\node [style=Z dot] (77) at (-4.75, -7) {};
		\node [style=sLabel] (78) at (-5.5, -6.5) {$w_2$};
		\node [style=sLabel] (79) at (-3.5, 3) {$P_{1,1}$};
		\node [style=sLabel] (80) at (-3.5, 0.5) {$P_{1,n}$};
		\node [style=sLabel] (81) at (-3.5, -1.5) {$P_{2,1}$};
		\node [style=sLabel] (82) at (-3.5, -4) {$P_{2,m}$};
		\node [style=none, rotate=90] (83) at (0.25, 1.75) {...};
		\node [style=none, rotate=90] (84) at (-4.75, -2.75) {...};
		\node [style=none, rotate=90] (85) at (0.25, -2.75) {...};
		\node [style=none] (86) at (-7, 0.25) {$\Eq{eq:pauli-box-clifford-decomposition}$};
	\end{pgfonlayer}
	\begin{pgfonlayer}{edgelayer}
		\draw [style=fault-free] (6.center)
			 to (7.center)
			 to (4.center)
			 to (5.center)
			 to cycle;
		\draw [style=fault-free] (10.center)
			 to (11.center)
			 to (8.center)
			 to (9.center)
			 to cycle;
		\draw [style=fault-free] (16.center)
			 to (17.center)
			 to (18.center)
			 to (19.center)
			 to cycle;
		\draw [style=fault-free] (23.center)
			 to (24.center)
			 to (25.center)
			 to (22.center)
			 to cycle;
		\draw [style=fault-free] (2.center) to (20.center);
		\draw [style=fault-free] (21.center) to (26.center);
		\draw [style=fault-free] (27.center) to (3.center);
		\draw [style=fault-free] (1.center) to (15.center);
		\draw [style=fault-free] (13.center) to (14.center);
		\draw [style=fault-free] (0.center) to (12.center);
		\draw [style=fault-free, bend left] (28) to (34.center);
		\draw [style=fault-free, bend right] (28) to (35.center);
		\draw [style=fault-free, in=45, out=-180] (29) to (36.center);
		\draw (30) to (28);
		\draw (31) to (29);
		\draw [style=fault-free] (46.center)
			 to (47.center)
			 to (44.center)
			 to (45.center)
			 to cycle;
		\draw [style=fault-free] (50.center)
			 to (51.center)
			 to (52.center)
			 to (53.center)
			 to cycle;
		\draw [style=fault-free] (42.center) to (48.center);
		\draw [style=fault-free] (49.center) to (54.center);
		\draw [style=fault-free] (55.center) to (43.center);
		\draw [style=fault-free] (62.center)
			 to (63.center)
			 to (64.center)
			 to (61.center)
			 to cycle;
		\draw [style=fault-free] (69.center)
			 to (70.center)
			 to (67.center)
			 to (68.center)
			 to cycle;
		\draw [style=fault-free] (59.center) to (65.center);
		\draw [style=fault-free] (66.center) to (71.center);
		\draw [style=fault-free] (72.center) to (60.center);
		\draw (77) to (76);
		\draw [style=fault-free, bend right] (76) to (73.center);
		\draw [style=fault-free, bend right] (56.center) to (76);
		\draw [style=fault-free, in=45, out=-135] (29) to (37.center);
		\draw [style=fault-free, in=45, out=-105, looseness=1.25] (29) to (57.center);
		\draw [style=fault-free, in=-90, out=45, looseness=0.50] (74.center) to (29);
	\end{pgfonlayer}
\end{tikzpicture}
\]
\end{proof}
To ensure that there is at least one fault gadget describing each undetectable fault, we apply \cref{prop:combining-fault-gadgets} iteratively.
We have:
\begin{proposition}
    \label{prop:unfold-all-gadget-combinations}
 A diagram $D_\mathcal{N}$ describing its noise model using fault gadgets can be fault-equivalently rewritten into a diagram $D_{\mathcal{N}'}$ describing a noise model $\mathcal{N}': \mathcal{F}' \rightarrow \mathbb{N}^+$, such that every undetectable fault $F \in \langle \mathcal{F}' \rangle$ is described by at least one fault gadget annotated with $\weightfunc(F)$.
\end{proposition}
\begin{proof}
 All undetectable faults in $\langle \mathcal{F}' \rangle$ can, by definition, be algebraically formed by multiplying elements in $\mathcal{F}'$.
 We can replicate any particular combination diagrammatically using fault gadgets by combining gadgets two at a time using \cref{prop:combining-fault-gadgets}.

 In particular, for a given undetectable fault $F$, we can include the combination that determines the weight $\weightfunc(F)$, \ie the combination that has the lowest sum of atomic weights.
 Replicating such a combination with fault gadgets leads to this sum as the annotation, yielding the claim.
\end{proof}
The resulting diagram enumerates all undetectable faults of $\mathcal{N}'$, and since the rewrite is fault-equivalent, it also enumerates all undetectable faults of the original noise model $\mathcal{N}$.
We have thus, in addition to the first two, resolved the third challenge.

Once again, this step substantially increases the number of fault gadgets, to the point that there is little value in showing the result for our running example.
We will return to the example after the next step, which reduces the number of fault gadgets.

For the fourth challenge, we realise that two fault gadgets that have the same targets describe faults that trivially have the same effect.
When checking fault equivalence, one of these faults is redundant; fault equivalence only requires that for any undetectable fault on one side, there exists an equivalent fault of at most the same weight on the other side.
Therefore, if one diagram has two equivalent fault gadgets, we can fault-equivalently remove the one with higher weight.
We have:
\begin{proposition}
    \label{prop:merge-gadgets}
    If two fault gadgets have the same targets, we can fault-equivalently remove the one with higher weight.
\end{proposition}
\begin{proof}
    \[ \begin{tikzpicture}
	\begin{pgfonlayer}{nodelayer}
		\node [style=none] (0) at (-9.75, 5.25) {};
		\node [style=none] (1) at (-4.75, 5.25) {};
		\node [style=none] (2) at (-9.75, 2.75) {};
		\node [style=none] (3) at (-4.75, 2.75) {};
		\node [style=none] (8) at (-5.25, 5.75) {};
		\node [style=none] (9) at (-5.25, 4.75) {};
		\node [style=none] (10) at (-6.25, 4.75) {};
		\node [style=none] (11) at (-6.25, 5.75) {};
		\node [style=none] (12) at (-8.25, 5.75) {};
		\node [style=none] (13) at (-8.25, 4.75) {};
		\node [style=none] (14) at (-9.25, 4.75) {};
		\node [style=none] (15) at (-9.25, 5.75) {};
		\node [style=none] (16) at (-9.25, 5.25) {};
		\node [style=none] (17) at (-8.25, 5.25) {};
		\node [style=none] (18) at (-6.25, 5.25) {};
		\node [style=none] (19) at (-5.25, 5.25) {};
		\node [style=none] (20) at (-8.25, 3.25) {};
		\node [style=none] (21) at (-8.25, 2.25) {};
		\node [style=none] (22) at (-9.25, 2.25) {};
		\node [style=none] (23) at (-9.25, 3.25) {};
		\node [style=none] (24) at (-9.25, 2.75) {};
		\node [style=none] (25) at (-8.25, 2.75) {};
		\node [style=none] (26) at (-5.25, 3.25) {};
		\node [style=none] (27) at (-5.25, 2.25) {};
		\node [style=none] (28) at (-6.25, 2.25) {};
		\node [style=none] (29) at (-6.25, 3.25) {};
		\node [style=none] (30) at (-6.25, 2.75) {};
		\node [style=none] (31) at (-5.25, 2.75) {};
		\node [style=X dot] (32) at (-7.75, 7.25) {};
		\node [style=X dot] (33) at (-6.75, 7.25) {};
		\node [style=Z dot] (34) at (-7.75, 8.25) {};
		\node [style=Z dot] (35) at (-6.75, 8.25) {};
		\node [style=sLabel] (36) at (-6, 7.75) {$w_2$};
		\node [style=sLabel] (37) at (-8.5, 7.75) {$w_1$};
		\node [style=none] (38) at (-8.75, 5.75) {};
		\node [style=none] (39) at (-8.75, 3.25) {};
		\node [style=none] (40) at (-5.75, 5.75) {};
		\node [style=none] (41) at (-5.75, 3.25) {};
		\node [style=sLabel] (42) at (-8.75, 5.25) {$P_1$};
		\node [style=sLabel] (43) at (-5.75, 5.25) {$P_1$};
		\node [style=sLabel] (44) at (-8.75, 2.75) {$P_n$};
		\node [style=sLabel] (45) at (-5.75, 2.75) {$P_n$};
		\node [style=none, rotate=90] (46) at (-7.25, 4) {...};
		\node [style=none] (47) at (-2.75, 5.25) {};
		\node [style=none] (48) at (2.25, 5.25) {};
		\node [style=none] (49) at (-2.75, 2.75) {};
		\node [style=none] (50) at (2.25, 2.75) {};
		\node [style=none] (51) at (1.75, 5.75) {};
		\node [style=none] (52) at (1.75, 4.75) {};
		\node [style=none] (53) at (0.75, 4.75) {};
		\node [style=none] (54) at (0.75, 5.75) {};
		\node [style=none] (55) at (-1.25, 5.75) {};
		\node [style=none] (56) at (-1.25, 4.75) {};
		\node [style=none] (57) at (-2.25, 4.75) {};
		\node [style=none] (58) at (-2.25, 5.75) {};
		\node [style=none] (59) at (-2.25, 5.25) {};
		\node [style=none] (60) at (-1.25, 5.25) {};
		\node [style=none] (61) at (0.75, 5.25) {};
		\node [style=none] (62) at (1.75, 5.25) {};
		\node [style=none] (63) at (-1.25, 3.25) {};
		\node [style=none] (64) at (-1.25, 2.25) {};
		\node [style=none] (65) at (-2.25, 2.25) {};
		\node [style=none] (66) at (-2.25, 3.25) {};
		\node [style=none] (67) at (-2.25, 2.75) {};
		\node [style=none] (68) at (-1.25, 2.75) {};
		\node [style=none] (69) at (1.75, 3.25) {};
		\node [style=none] (70) at (1.75, 2.25) {};
		\node [style=none] (71) at (0.75, 2.25) {};
		\node [style=none] (72) at (0.75, 3.25) {};
		\node [style=none] (73) at (0.75, 2.75) {};
		\node [style=none] (74) at (1.75, 2.75) {};
		\node [style=X dot] (75) at (-0.75, 7.25) {};
		\node [style=X dot] (76) at (0.25, 7.25) {};
		\node [style=Z dot] (77) at (-0.75, 8.25) {};
		\node [style=Z dot] (78) at (0.25, 8.25) {};
		\node [style=sLabel] (79) at (1, 7.75) {$w_2$};
		\node [style=sLabel] (80) at (-1.5, 7.75) {$w_1$};
		\node [style=none] (81) at (-1.75, 5.75) {};
		\node [style=none] (82) at (-1.75, 3.25) {};
		\node [style=none] (83) at (1.25, 5.75) {};
		\node [style=none] (84) at (1.25, 3.25) {};
		\node [style=sLabel] (85) at (-1.75, 5.25) {$C_1$};
		\node [style=sLabel] (86) at (1.25, 5.25) {$C_1^\dagger$};
		\node [style=sLabel] (87) at (-1.75, 2.75) {$C_n$};
		\node [style=sLabel] (88) at (1.25, 2.75) {$C_n^\dagger$};
		\node [style=none] (90) at (-3.75, 4) {$\FaultEq$};
		\node [style=Z dot] (91) at (0.25, 5.25) {};
		\node [style=Z dot] (92) at (-0.75, 5.25) {};
		\node [style=Z dot] (93) at (-0.75, 2.75) {};
		\node [style=Z dot] (94) at (0.25, 2.75) {};
		\node [style=none] (95) at (4.25, 5.25) {};
		\node [style=none] (96) at (9.25, 5.25) {};
		\node [style=none] (97) at (4.25, 2.75) {};
		\node [style=none] (98) at (9.25, 2.75) {};
		\node [style=none] (99) at (8.75, 5.75) {};
		\node [style=none] (100) at (8.75, 4.75) {};
		\node [style=none] (101) at (7.75, 4.75) {};
		\node [style=none] (102) at (7.75, 5.75) {};
		\node [style=none] (103) at (5.75, 5.75) {};
		\node [style=none] (104) at (5.75, 4.75) {};
		\node [style=none] (105) at (4.75, 4.75) {};
		\node [style=none] (106) at (4.75, 5.75) {};
		\node [style=none] (107) at (4.75, 5.25) {};
		\node [style=none] (108) at (5.75, 5.25) {};
		\node [style=none] (109) at (7.75, 5.25) {};
		\node [style=none] (110) at (8.75, 5.25) {};
		\node [style=none] (111) at (5.75, 3.25) {};
		\node [style=none] (112) at (5.75, 2.25) {};
		\node [style=none] (113) at (4.75, 2.25) {};
		\node [style=none] (114) at (4.75, 3.25) {};
		\node [style=none] (115) at (4.75, 2.75) {};
		\node [style=none] (116) at (5.75, 2.75) {};
		\node [style=none] (117) at (8.75, 3.25) {};
		\node [style=none] (118) at (8.75, 2.25) {};
		\node [style=none] (119) at (7.75, 2.25) {};
		\node [style=none] (120) at (7.75, 3.25) {};
		\node [style=none] (121) at (7.75, 2.75) {};
		\node [style=none] (122) at (8.75, 2.75) {};
		\node [style=X dot] (123) at (6.25, 7.25) {};
		\node [style=X dot] (124) at (7.25, 7.25) {};
		\node [style=Z dot] (125) at (6.25, 8.25) {};
		\node [style=Z dot] (126) at (7.25, 8.25) {};
		\node [style=sLabel] (127) at (8, 7.75) {$w_2$};
		\node [style=sLabel] (128) at (5.5, 7.75) {$w_1$};
		\node [style=none] (129) at (5.25, 5.75) {};
		\node [style=none] (130) at (5.25, 3.25) {};
		\node [style=none] (131) at (8.25, 5.75) {};
		\node [style=none] (132) at (8.25, 3.25) {};
		\node [style=sLabel] (133) at (5.25, 5.25) {$C_1$};
		\node [style=sLabel] (134) at (8.25, 5.25) {$C_1^\dagger$};
		\node [style=sLabel] (135) at (5.25, 2.75) {$C_n$};
		\node [style=sLabel] (136) at (8.25, 2.75) {$C_n^\dagger$};
		\node [style=none, rotate=90] (137) at (-1.75, 4) {...};
		\node [style=none] (138) at (3.25, 4) {$\FaultEq$};
		\node [style=Z dot] (143) at (6.75, 5.25) {};
		\node [style=Z dot] (144) at (6.75, 6) {};
		\node [style=Z dot] (145) at (6.75, 2.75) {};
		\node [style=Z dot] (146) at (6.75, 3.5) {};
		\node [style=none, rotate=90] (147) at (1.25, 4) {...};
		\node [style=none] (148) at (-9.75, -1.5) {};
		\node [style=none] (149) at (-5.25, -1.5) {};
		\node [style=none] (150) at (-9.75, -4) {};
		\node [style=none] (151) at (-5.25, -4) {};
		\node [style=none] (152) at (-5.75, -1) {};
		\node [style=none] (153) at (-5.75, -2) {};
		\node [style=none] (154) at (-6.75, -2) {};
		\node [style=none] (155) at (-6.75, -1) {};
		\node [style=none] (156) at (-8.25, -1) {};
		\node [style=none] (157) at (-8.25, -2) {};
		\node [style=none] (158) at (-9.25, -2) {};
		\node [style=none] (159) at (-9.25, -1) {};
		\node [style=none] (160) at (-9.25, -1.5) {};
		\node [style=none] (161) at (-8.25, -1.5) {};
		\node [style=none] (162) at (-6.75, -1.5) {};
		\node [style=none] (163) at (-5.75, -1.5) {};
		\node [style=none] (164) at (-8.25, -3.5) {};
		\node [style=none] (165) at (-8.25, -4.5) {};
		\node [style=none] (166) at (-9.25, -4.5) {};
		\node [style=none] (167) at (-9.25, -3.5) {};
		\node [style=none] (168) at (-9.25, -4) {};
		\node [style=none] (169) at (-8.25, -4) {};
		\node [style=none] (170) at (-5.75, -3.5) {};
		\node [style=none] (171) at (-5.75, -4.5) {};
		\node [style=none] (172) at (-6.75, -4.5) {};
		\node [style=none] (173) at (-6.75, -3.5) {};
		\node [style=none] (174) at (-6.75, -4) {};
		\node [style=none] (175) at (-5.75, -4) {};
		\node [style=Z dot] (178) at (-7.75, 1.5) {};
		\node [style=Z dot] (179) at (-6.75, 1.5) {};
		\node [style=sLabel] (180) at (-6.5, 1) {$w_2$};
		\node [style=sLabel] (181) at (-8, 1) {$w_1$};
		\node [style=none] (182) at (-8.75, -1) {};
		\node [style=none] (183) at (-8.75, -3.5) {};
		\node [style=none] (184) at (-6.25, -1) {};
		\node [style=none] (185) at (-6.25, -3.5) {};
		\node [style=sLabel] (186) at (-8.75, -1.5) {$C_1$};
		\node [style=sLabel] (187) at (-6.25, -1.5) {$C_1^\dagger$};
		\node [style=sLabel] (188) at (-8.75, -4) {$C_n$};
		\node [style=sLabel] (189) at (-6.25, -4) {$C_n^\dagger$};
		\node [style=none, rotate=90] (190) at (-6.25, -2.75) {...};
		\node [style=none] (191) at (-10.75, -2.75) {$\FaultEq$};
		\node [style=Z dot] (192) at (-7.75, -1.5) {};
		\node [style=Z dot] (194) at (-7.75, -4) {};
		\node [style=none] (196) at (-10.75, -2) {$\Bialgebra$};
		\node [style=X dot] (197) at (-7.25, -0.5) {};
		\node [style=Z dot] (198) at (-7.25, 0.5) {};
		\node [style=none] (240) at (-4, -2.75) {$\FaultEq$};
		\node [style=none] (243) at (-4, -2) {$\FEMerge$};
		\node [style=none] (246) at (-2.75, -1.5) {};
		\node [style=none] (247) at (1.75, -1.5) {};
		\node [style=none] (248) at (-2.75, -4) {};
		\node [style=none] (249) at (1.75, -4) {};
		\node [style=none] (250) at (1.25, -1) {};
		\node [style=none] (251) at (1.25, -2) {};
		\node [style=none] (252) at (0.25, -2) {};
		\node [style=none] (253) at (0.25, -1) {};
		\node [style=none] (254) at (-1.25, -1) {};
		\node [style=none] (255) at (-1.25, -2) {};
		\node [style=none] (256) at (-2.25, -2) {};
		\node [style=none] (257) at (-2.25, -1) {};
		\node [style=none] (258) at (-2.25, -1.5) {};
		\node [style=none] (259) at (-1.25, -1.5) {};
		\node [style=none] (260) at (0.25, -1.5) {};
		\node [style=none] (261) at (1.25, -1.5) {};
		\node [style=none] (262) at (-1.25, -3.5) {};
		\node [style=none] (263) at (-1.25, -4.5) {};
		\node [style=none] (264) at (-2.25, -4.5) {};
		\node [style=none] (265) at (-2.25, -3.5) {};
		\node [style=none] (266) at (-2.25, -4) {};
		\node [style=none] (267) at (-1.25, -4) {};
		\node [style=none] (268) at (1.25, -3.5) {};
		\node [style=none] (269) at (1.25, -4.5) {};
		\node [style=none] (270) at (0.25, -4.5) {};
		\node [style=none] (271) at (0.25, -3.5) {};
		\node [style=none] (272) at (0.25, -4) {};
		\node [style=none] (273) at (1.25, -4) {};
		\node [style=Z dot] (275) at (-0.25, 1.5) {};
		\node [style=sLabel] (276) at (1.5, 1) {$\min(w_1,w_2)$};
		\node [style=none] (277) at (-1.75, -1) {};
		\node [style=none] (278) at (-1.75, -3.5) {};
		\node [style=none] (279) at (0.75, -1) {};
		\node [style=none] (280) at (0.75, -3.5) {};
		\node [style=sLabel] (281) at (-1.75, -1.5) {$C_1$};
		\node [style=sLabel] (282) at (0.75, -1.5) {$C_1^\dagger$};
		\node [style=sLabel] (283) at (-1.75, -4) {$C_n$};
		\node [style=sLabel] (284) at (0.75, -4) {$C_n^\dagger$};
		\node [style=Z dot] (287) at (-0.75, -1.5) {};
		\node [style=Z dot] (288) at (-0.75, -4) {};
		\node [style=X dot] (290) at (-0.25, -0.5) {};
		\node [style=Z dot] (291) at (-0.25, 0.5) {};
		\node [style=none] (331) at (3, -2.75) {$\FaultEq$};
		\node [style=none] (336) at (4.25, -1.5) {};
		\node [style=none] (337) at (6.75, -1.5) {};
		\node [style=none] (338) at (4.25, -4) {};
		\node [style=none] (339) at (6.75, -4) {};
		\node [style=none] (344) at (6, -1) {};
		\node [style=none] (345) at (6, -2) {};
		\node [style=none] (346) at (5, -2) {};
		\node [style=none] (347) at (5, -1) {};
		\node [style=none] (348) at (5, -1.5) {};
		\node [style=none] (349) at (6, -1.5) {};
		\node [style=none] (352) at (6, -3.5) {};
		\node [style=none] (353) at (6, -4.5) {};
		\node [style=none] (354) at (5, -4.5) {};
		\node [style=none] (355) at (5, -3.5) {};
		\node [style=none] (356) at (5, -4) {};
		\node [style=none] (357) at (6, -4) {};
		\node [style=Z dot] (364) at (5.5, 1.5) {};
		\node [style=sLabel] (365) at (7.25, 1) {$\min(w_1,w_2)$};
		\node [style=none] (366) at (5.5, -1) {};
		\node [style=none] (367) at (5.5, -3.5) {};
		\node [style=sLabel] (370) at (5.5, -1.5) {$P_1$};
		\node [style=sLabel] (372) at (5.5, -4) {$P_n$};
		\node [style=none, rotate=90] (374) at (5.5, -2.75) {...};
		\node [style=X dot] (378) at (5.5, 0.5) {};
		\node [style=none] (379) at (-3.75, 4.75) {$\Eq{eq:pauli-box-clifford-decomposition}$};
		\node [style=none] (380) at (3, -2) {$\Eq{eq:pauli-box-clifford-decomposition}$};
		\node [style=none, rotate=90] (381) at (8.25, 4) {...};
		\node [style=none, rotate=90] (382) at (5.25, 4) {...};
		\node [style=none, rotate=90] (383) at (0.75, -2.75) {...};
		\node [style=none, rotate=90] (384) at (-1.75, -2.75) {...};
		\node [style=none, rotate=90] (385) at (-8.75, -2.75) {...};
		\node [style=none] (386) at (3.25, 4.75) {$\Fusion$};
	\end{pgfonlayer}
	\begin{pgfonlayer}{edgelayer}
		\draw [style=fault-free] (10.center)
			 to (11.center)
			 to (8.center)
			 to (9.center)
			 to cycle;
		\draw [style=fault-free] (14.center)
			 to (15.center)
			 to (12.center)
			 to (13.center)
			 to cycle;
		\draw [style=fault-free] (20.center)
			 to (21.center)
			 to (22.center)
			 to (23.center)
			 to cycle;
		\draw [style=fault-free] (27.center)
			 to (28.center)
			 to (29.center)
			 to (26.center)
			 to cycle;
		\draw [style=fault-free] (2.center) to (24.center);
		\draw [style=fault-free] (25.center) to (30.center);
		\draw [style=fault-free] (31.center) to (3.center);
		\draw [style=fault-free] (1.center) to (19.center);
		\draw [style=fault-free] (17.center) to (18.center);
		\draw [style=fault-free] (0.center) to (16.center);
		\draw [style=fault-free, bend right] (32) to (38.center);
		\draw [style=fault-free, in=75, out=-75, looseness=0.75] (32) to (39.center);
		\draw [style=fault-free, bend left] (33) to (40.center);
		\draw [style=fault-free, in=105, out=-105, looseness=0.75] (33) to (41.center);
		\draw (34) to (32);
		\draw (35) to (33);
		\draw [style=fault-free] (53.center)
			 to (54.center)
			 to (51.center)
			 to (52.center)
			 to cycle;
		\draw [style=fault-free] (57.center)
			 to (58.center)
			 to (55.center)
			 to (56.center)
			 to cycle;
		\draw [style=fault-free] (63.center)
			 to (64.center)
			 to (65.center)
			 to (66.center)
			 to cycle;
		\draw [style=fault-free] (70.center)
			 to (71.center)
			 to (72.center)
			 to (69.center)
			 to cycle;
		\draw [style=fault-free] (49.center) to (67.center);
		\draw [style=fault-free] (68.center) to (73.center);
		\draw [style=fault-free] (74.center) to (50.center);
		\draw [style=fault-free] (48.center) to (62.center);
		\draw [style=fault-free] (60.center) to (61.center);
		\draw [style=fault-free] (47.center) to (59.center);
		\draw (77) to (75);
		\draw (78) to (76);
		\draw [style=fault-free, in=90, out=-90] (75) to (92);
		\draw [style=fault-free] (76) to (91);
		\draw [style=fault-free, in=90, out=-105, looseness=1.50] (76) to (94);
		\draw [style=fault-free, in=90, out=-75, looseness=1.50] (75) to (93);
		\draw [style=fault-free] (100.center)
			 to (101.center)
			 to (102.center)
			 to (99.center)
			 to cycle;
		\draw [style=fault-free] (105.center)
			 to (106.center)
			 to (103.center)
			 to (104.center)
			 to cycle;
		\draw [style=fault-free] (112.center)
			 to (113.center)
			 to (114.center)
			 to (111.center)
			 to cycle;
		\draw [style=fault-free] (119.center)
			 to (120.center)
			 to (117.center)
			 to (118.center)
			 to cycle;
		\draw [style=fault-free] (97.center) to (115.center);
		\draw [style=fault-free] (116.center) to (121.center);
		\draw [style=fault-free] (122.center) to (98.center);
		\draw [style=fault-free] (96.center) to (110.center);
		\draw [style=fault-free] (108.center) to (109.center);
		\draw [style=fault-free] (95.center) to (107.center);
		\draw (125) to (123);
		\draw (126) to (124);
		\draw [style=fault-free] (146) to (145);
		\draw [style=fault-free] (144) to (143);
		\draw [style=fault-free] (123) to (144);
		\draw [style=fault-free] (124) to (144);
		\draw [style=fault-free, in=-90, out=120] (146) to (123);
		\draw [style=fault-free, in=60, out=-90] (124) to (146);
		\draw [style=fault-free] (153.center)
			 to (154.center)
			 to (155.center)
			 to (152.center)
			 to cycle;
		\draw [style=fault-free] (158.center)
			 to (159.center)
			 to (156.center)
			 to (157.center)
			 to cycle;
		\draw [style=fault-free] (165.center)
			 to (166.center)
			 to (167.center)
			 to (164.center)
			 to cycle;
		\draw [style=fault-free] (172.center)
			 to (173.center)
			 to (170.center)
			 to (171.center)
			 to cycle;
		\draw [style=fault-free] (150.center) to (168.center);
		\draw [style=fault-free] (169.center) to (174.center);
		\draw [style=fault-free] (175.center) to (151.center);
		\draw [style=fault-free] (149.center) to (163.center);
		\draw [style=fault-free] (161.center) to (162.center);
		\draw [style=fault-free] (148.center) to (160.center);
		\draw [style=fault-free, bend right=15] (197) to (192);
		\draw [style=fault-free, in=90, out=-90] (197) to (194);
		\draw [style=fault-free] (198) to (197);
		\draw (198) to (178);
		\draw (198) to (179);
		\draw [style=fault-free] (250.center)
			 to (251.center)
			 to (252.center)
			 to (253.center)
			 to cycle;
		\draw [style=fault-free] (254.center)
			 to (255.center)
			 to (256.center)
			 to (257.center)
			 to cycle;
		\draw [style=fault-free] (265.center)
			 to (262.center)
			 to (263.center)
			 to (264.center)
			 to cycle;
		\draw [style=fault-free] (269.center)
			 to (270.center)
			 to (271.center)
			 to (268.center)
			 to cycle;
		\draw [style=fault-free] (248.center) to (266.center);
		\draw [style=fault-free] (267.center) to (272.center);
		\draw [style=fault-free] (273.center) to (249.center);
		\draw [style=fault-free] (247.center) to (261.center);
		\draw [style=fault-free] (259.center) to (260.center);
		\draw [style=fault-free] (246.center) to (258.center);
		\draw [style=fault-free, bend right=15] (290) to (287);
		\draw [style=fault-free, in=90, out=-90] (290) to (288);
		\draw [style=fault-free] (291) to (290);
		\draw (291) to (275);
		\draw [style=fault-free] (344.center)
			 to (345.center)
			 to (346.center)
			 to (347.center)
			 to cycle;
		\draw [style=fault-free] (352.center)
			 to (353.center)
			 to (354.center)
			 to (355.center)
			 to cycle;
		\draw [style=fault-free] (338.center) to (356.center);
		\draw [style=fault-free] (336.center) to (348.center);
		\draw (378) to (364);
		\draw [style=fault-free] (357.center) to (339.center);
		\draw [style=fault-free] (337.center) to (349.center);
		\draw [style=fault-free] (378) to (366.center);
		\draw [style=fault-free, in=60, out=-45] (378) to (367.center);
	\end{pgfonlayer}
\end{tikzpicture}
 \]
\end{proof}
Incidentally, the fault gadgets we receive after \cref{prop:disconnect-spiders-summary} using the flip operators are already normalised in the context of that diagram: Two fault gadgets have the same targets if and only if their faults have the same effect.
Their effect can be fully described by spacetime-equivalence classes through \cref{prop:spacetime-equivalence-undetectable}, which in turn may be fully described by anticommutation with Pauli webs by \cref{prop:spacetime-equivalence}, and thus also by the flip operators.
By \cref{prop:commuting-targets}, we can commute fault gadgets as required to apply \cref{prop:merge-gadgets}, and repeated application resolves the fourth challenge.

Finally, fault gadgets impose an ordering of the faults, whereas the noise model they represent contains the atomic faults in an unordered collection.
However, we can always impose the lexicographical order of fault gadgets given by their representation as Pauli products to provide a unique ordering at any given point, moving fault gadgets using commutation.
Thus, we have addressed all five challenges.

\begin{example}{Four-legged spider}{}
    For our running example, the flip operators we used form the smallest flip operator collection when ordering the boundary edges clockwise, starting with the top left.
    Performing the full rewrite outlined in this section, we discover seven different equivalence classes of $X$-type faults:
    \[\input{running-example/final-form.tikz}\]
\end{example}

We now define:
\begin{definition}[Fault normal form]
    \label{def:fault-normal-form}
    A diagram $D_\mathcal{N}$ describing its noise model $\mathcal{N}: \mathcal{F} \rightarrow \mathbb{N}^+$ using fault gadgets is in \textit{fault normal form} if:
    \begin{enumerate}
        \item $D_\mathcal{N}$ has a separation of the semantics of $D$ and fault gadgets following \cref{sec:separating-semantics-and-noise}, \textbf{and}
        \item the semantic part $D$ inside $D_\mathcal{N}$ is in its unique normal form following~\autocite{KissingerWetering2024Book}, \textbf{and}
        \item the fault gadgets for all faults in $\mathcal{F}$ are ordered lexicographically, \textbf{and}
        \item for every spacetime equivalence class of faults in $\langle \mathcal{F} \rangle$, exactly one member is in $\mathcal{F}$, \textbf{and}
        \item every fault in $\mathcal{F}$ is represented by a fault gadget constructed from the uniquely smallest edge-flip operator collection obtainable for $D$, \textbf{and}
        \item every fault gadget is annotated with the lowest weight attainable in the spacetime equivalence class of the fault it describes.
    \end{enumerate}
\end{definition}
Intuitively, the fault normal form is the diagram that is obtained by constructing the smallest flip operator collection for the base diagram, repeated application of \cref{prop:detecting-region-popping}, \cref{prop:unfold-all-gadget-combinations}, \cref{prop:merge-gadgets}, and applying a unique ordering to the gadgets.
By construction, the fault normal form features none of the issues outlined at the beginning of this section.
Furthermore, we can infer:
\begin{proposition}
    \label{cor:fault-normal-form-exists}
 Every diagram $D_\mathcal{N}$ describing some noise model using fault gadgets has a fault normal form that is obtained using only the rules from \cref{fig:complete-axioms}.
\end{proposition}
\begin{proof}
 This follows directly from the construction outlined in this section.
\end{proof}
    \section{Completeness}\label{sec:completeness}
We have shown that any diagram can be brought into a form that consists of two parts: a completely fault-free Clifford ZX diagram followed by an ordered layer of undetectable faults that have a unique representative for each possible fault.
We call this form the \textit{fault normal form} following \cref{def:fault-normal-form}.
The final step toward completeness is to show that any two fault-equivalent diagrams have the same fault normal form, and can thus be rewritten into each other.

\subsection{Completeness for Fault Equivalence}
We first state:
\begin{proposition}
    \label{prop:normal-form-unique-per-class}
    Let $D, D'$ be Clifford ZX diagrams with respective noise models $\mathcal N,\mathcal N'$ such that $D$ under $\mathcal N$ is fault-equivalent to $D'$ under $\mathcal N'$.
    Then ${D}_{\mathcal N}$ and ${D'}_{\mathcal N'}$ have the same fault normal form.
\end{proposition}
\begin{proof}
    As fault equivalence implies semantic equivalence~\autocite{rodatz2025faulttoleranceconstruction}, the normalised fault-free part of the two diagrams must implement the same linear map and therefore, by~\autocite{backensZXcalculusCompleteStabilizer2014} must be the same.

    For the fault gadgets, we first argue that the two diagrams in fault normal form must have the same gadgets, labelled by the same weights.
    As the two diagrams are fault-equivalent, they allow for the same group of undetectable faults.
    Therefore, they must have exactly the same equivalence classes of faults, and since their smallest edge-flip operator collections are the same, they must have exactly the same fault gadgets.
    The fault gadgets are labelled by the minimum weight of all faults in an equivalence class.
    Fault equivalence requires those minimum representatives between two diagrams to have the same weight, as otherwise, there would exist an undetectable fault on one diagram that does not have a corresponding fault on the other diagram of at most the same weight, violating the assumption that the two diagrams are fault-equivalent.
    But then, the fault gadgets of the two diagrams must not only be the same but also labelled with the same weight.
    Finally, as all fault gadgets must be different --- we have previously eliminated duplicates --- the lexicographical ordering of gadgets on the boundary is unique and must therefore be the same for both diagrams, and thus the diagrams must be the same.
\end{proof}

Finally, it remains to show our main result:
\begin{theorem}
 \label{thm:completeness}
 Let $D, D'$ be Clifford ZX diagrams with respective noise models $\mathcal N,\mathcal N'$ such that $D$ under $\mathcal N$ is fault-equivalent to $D'$ under $\mathcal N'$.
 Then, only using the rules in \cref{fig:complete-axioms}, ${D}_{\mathcal N}$ can be rewritten into ${D'}_{\mathcal N'}$.
\end{theorem}
\begin{proof}
    Via \cref{cor:fault-normal-form-exists}, we can rewrite both ${D}_{\mathcal N}$ and ${D'}_{\mathcal N'}$ into their fault normal form.
    By \cref{prop:normal-form-unique-per-class}, the normal form must be the same.
    We can thus first rewrite $D_{\mathcal{N}}$ into normal form, and then rewrite the normal form into $D'_{\mathcal{N}'}$ by reversing the sequence of rewrites used to bring $D'_{\mathcal{N}'}$ into normal form.
    Therefore, we have diagrammatically shown their fault equivalence.
\end{proof}

As we have previously shown that our rules are sound, we have now shown that the rule set from \cref{fig:complete-axioms} is sound and complete.

\begin{example}{Four-legged spider}{}
 For our running example, we now rewrite the two diagrams into one another:
    \[\input{running-example/full-rewrite.tikz}\]
\end{example}

\subsection{Completeness for $w$-Fault Equivalence and Fault Boundedness}
Adding some additional rules, we can easily extend the completeness result to other relationships of noisy ZX diagrams.
We first show how adding one more rule gives us a complete rule set for $w$-fault boundedness.
Then we show that a different rule gives completeness for fault boundedness.
Finally, we show that together, these two rules extend the completeness to $w$-fault boundedness.

\begin{figure}
    \centering
    \begin{minipage}{0.49\textwidth}
        \centering
        \input{06-completeness/complete-w-equivalence-axiom.tikz}
        \caption*{(a)}
    \end{minipage}
    \begin{minipage}{0.49\textwidth}
        \centering
        \input{06-completeness/complete-boundedness-axiom.tikz}
        \caption*{(b)}
    \end{minipage}
    \caption{Axioms that, in addition to \cref{fig:complete-axioms}, are required for a calculus that is complete for (a) $w$-fault equivalence and (b) fault boundedness. Both are required for a calculus that is complete for $w$-fault boundedness.}
    \label{fig:fault-equivalence-complete-additional-axioms}
\end{figure}

The two rules that are necessary to extend the completeness results are shown in \cref{fig:fault-equivalence-complete-additional-axioms}.
The first rule is necessary to have completeness for $w$-fault equivalence, and the second rule is necessary for fault boundedness.
For the second rule, we note that by $w = -$ we mean that no faults are allowed on the corresponding edge, \ie that it is idealised as fault-free.
An idealised edge is always fault-bounded by a non-idealised one, in that, as there are no allowed faults, any fault on the idealised edge trivially has a correspondence on the non-idealised one.
Therefore, in that case, the rewrite can introduce a fault with any weight.

Similar to the rules complete for fault equivalence, we must show that these additional rewrite rules are sound:
\begin{restatable}{proposition}{wsoundnessprop}
    The rule \cref{fig:fault-equivalence-complete-additional-axioms} (a) is sound for $w$-fault equivalence, and the rule \cref{fig:fault-equivalence-complete-additional-axioms} (b) is sound for fault boundedness.
    Both rules are sound for $w$-fault boundedness.
\end{restatable}
\begin{proof}
    See \cref{appendix:missing-proofs}.
\end{proof}
We note that every fault-equivalent rule by definition also implies weaker relationships like $w$-fault equivalence.
In particular, the previously introduced rewrite rules from \cref{fig:complete-axioms} are sound \wrt such weaker relationships.

Again, since $w$-fault boundedness is compositional and transitive, and so are each of its derived properties, we receive a usable rewrite system when including either of or both of the rules.
So we can now turn to completeness.

First, we show:
\begin{theorem}
 Let $D, D'$ be Clifford ZX diagrams with respective noise models $\mathcal N: \mathcal F \to \mathbb{N}^+, \, \mathcal N': \mathcal F' \to \mathbb{N}^+$ such that $D$ under $\mathcal N$ is $w$-fault-equivalent to $D'$ under $\mathcal N'$.
 Then, only using the rules in \cref{fig:complete-axioms} and \cref{fig:fault-equivalence-complete-additional-axioms} (a), ${D}_{\mathcal N}$ can be rewritten into ${D'}_{\mathcal N'}$.
\end{theorem}
\begin{proof}
 Using the rules from \cref{fig:complete-axioms}, we can fault-equivalently bring both diagrams into their unique fault normal form.
 As $D$ and $D'$ are only $w$-fault-equivalent, there might be fault equivalence classes that are present on one diagram but not the other, or where the fault gadgets have different values; however, only for classes where the minimum attained weight is at least $w$.
 For all equivalence classes with minimum attained weight less than $w$, the two diagrams must be the same.
 We can now apply the rule from \cref{fig:fault-equivalence-complete-additional-axioms}~(a) together with fault-free Clifford rules to remove all other fault gadgets annotated with weight $w$ or more.
 The only remaining fault gadgets must be the same on both sides, and therefore, the two modified normal forms must be the same.
 But then, only using our diagrammatic rewrite rules, we can rewrite one diagram into the other.
\end{proof}

Similarly, we can show: 
\begin{theorem}
 Let $D, D'$ be Clifford ZX diagrams with respective noise models $\mathcal N: \mathcal F \to \mathbb{N}^+, \, \mathcal N': \mathcal F' \to \mathbb{N}^+$ such that $D$ under $\mathcal N$ is fault-bounded by $D'$ under $\mathcal N'$.
 Then, only using the rules in \cref{fig:complete-axioms} and \cref{fig:fault-equivalence-complete-additional-axioms} (b), ${D}_{\mathcal N}$ can be rewritten into ${D'}_{\mathcal N'}$.
\end{theorem}
\begin{proof}
 Similar to above, we can bring both diagrams into their fault normal form using only our fault-equivalent rewrites.
 As $D$ is fault-bounded by $D'$, any fault gadget on $D$ must have at most the same weight as the corresponding fault gadget on $D'$.
 But then, we can use the rewrite from \cref{fig:fault-equivalence-complete-additional-axioms}~(b) to increase the weight of all the fault gadgets in ${D}_{\mathcal N}$ that have a lower weight than their corresponding gadget in ${D'}_{\mathcal N'}$.
 If there are faults that are possible on ${D'}_{\mathcal N'}$ but not on ${D}_{\mathcal N}$, we can use the fault-free rules to introduce a fault-free fault gadget that corresponds to that fault and then use \cref{fig:fault-equivalence-complete-additional-axioms}~(b) to unidealise the spawning edge.
 But then, we have a directed rewrite from ${D}_{\mathcal N}$ to ${D'}_{\mathcal N'}$ using only the allowed rules.
\end{proof}

Finally, using the two rules, we can state: 
\begin{theorem}
 Let $D, D'$ be Clifford ZX diagrams with respective noise models $\mathcal N: \mathcal F \to \mathbb{N}^+, \, \mathcal N': \mathcal F' \to \mathbb{N}^+$ such that $D$ under $\mathcal N$ is $w$-fault-bounded by $D'$ under $\mathcal N'$.
 Then, only using the rules in \cref{fig:complete-axioms} and \cref{fig:fault-equivalence-complete-additional-axioms}, ${D}_{\mathcal N}$ can be rewritten into ${D'}_{\mathcal N'}$.
\end{theorem}
\begin{proof}
 As $D$ is $w$-fault-bounded by $D'$, we are guaranteed that all faults of weight less than $w$ on $D'$ have a corresponding fault of at least the same weight on $D'$.
 Therefore, just as for the previous two proofs, we can first remove all faults of weight more than $w$ from ${D}_{\mathcal N}$ using \cref{fig:fault-equivalence-complete-additional-axioms}~(a) and then increase the weight and introduce the missing ones to rewrite the normal form of ${D}_{\mathcal N}$ into the normal form of ${D'}_{\mathcal N'}$.
\end{proof}

    \section{Conclusion}\label{sec:future-work}
In this work, we provided an axiomatisation of a ZX calculus that is sound and complete for showing fault equivalence between Clifford diagrams.
At the centre of the diagrammatic implementation, we utilised fault gadgets to represent and manipulate faults.
Using the fault-equivalent rewrites, we provided a method to push out all the faults, separating the semantic part of the diagram from the noisy part.
For this, we use the idea of spacetime equivalent faults and show how they can be characterised in terms of their interaction with the Pauli webs.
Having separated the semantic part of the diagram from the noisy part, we can compare two different ZX diagrams by independently looking at each of the two parts.
Leveraging existing completeness results for fault-free Clifford ZX diagrams~\autocite{backensZXcalculusCompleteStabilizer2014,backens2017AsimplifiedStabiliserZXCalculus,KissingerWetering2024Book}, our work primarily focused on comparing the noisy part of the diagram.
We showed that, using fault-equivalent rewrites, the fault gadget on the diagrams can be rewritten to have one representative fault gadget per spacetime equivalence class of faults.

Beyond the theoretical implications of the completeness results, their constructive proofs offer insight into an algorithm for automatic fault equivalence checking.
This work is accompanied by an implementation~\autocite{paritea2026code}, providing \eg an automated variant of checking fault equivalence and an interface for extracting a Tanner graph.
The approach to checking fault equivalence derived from the constructive proof has an exponential complexity advantage over the naive approach~\autocite{rueschCompletenessFaultTolerant2025}.
Instead of scaling exponentially in the number of faults, this algorithm scales in the number of physical qubits and detecting regions.
The implementation in~\autocite{paritea2026code} optimises this further by allowing some transformations outside the ZX calculus.

One way to view the procedure presented in this work is that when we manipulate the fault gadgets, we fault-equivalently change the noise model the diagram is exposed to.
For example, after pushing out the fault gadgets, we can view the result as the original diagram under a different, fault-equivalent noise model, where now all faults only act on the boundary and the flip operators.
As pointed out at the end of~\cref{sec:separating-semantics-and-noise}, this allows us, for example, to extract the detector-error model~\autocite{Gidney2021stimfaststabilizer,derksDesigningFaulttolerantCircuits2024} of the circuit under the given noise model.
However, we can consider that after eliminating all the detecting regions, the remaining faults are the undetectable ones.
This allows us to infer the distribution of undetectable faults induced by the noise model, from which we can derive other properties such as the circuit distance.
In future research, it would be interesting to focus more on the insights that can be gained via such procedures and explore whether they can be leveraged into efficient circuit analysis tools.

Beyond Clifford ZX diagrams, we would like to generalise fault equivalence and the completeness results to other fragments, working towards non-classically simulable ones, such as Clifford+T.
We expect this to be highly non-trivial, as such fragments do not preserve Paulis under conjugation.

Furthermore, the ZX calculus is usually considered in the case of qubits.
It would be interesting to explore fault equivalence and completeness on higher-dimensional systems~\autocite{ranchin2014quditZXcalculus, booth2022qupitZXCalculi, poor2023qupitCliffordCompleteness} or continuous-variable ZX~\autocite{shaikh2024fockedupzxcalculuspicturing, booth2024completeequationaltheoriesclassical, Nagayoshi_2025}.

Additionally, fault equivalence, as defined in this work, considers weighted Pauli noise.
While this captures a common noise model assumed in the literature, an interesting avenue for future work would be to generalise fault equivalence beyond weighted Pauli noise to other noise models, such as stochastic Pauli noise.
For this, the complete set of rewrites presented in this work provides a promising starting point.
Once other noise models are proposed, completeness results for fault equivalence under these noise models would probably be derived from this work.

Finally, the implementation accompanying this work is still in early stages, and usability as well as performance can be considerably improved.
More features could be added, such as integration with other QEC libraries like the highly-parallel sampling library \texttt{stim}~\autocite{Gidney2021stimfaststabilizer}.
Additionally, as research progresses, support for alternative noise models or circuits beyond qubits could be considered.
Considering the exponential cost of the proposed algorithm, adaptations for highly structured, larger circuits could be added.
For example, substantial improvements could be achieved for circuits that are composed of smaller gadgets.

Overall, this work is an important step in understanding diagrammatic calculi for fault equivalence while providing exciting new avenues for analysing and understanding the behaviour of circuits under noise that ranges beyond completeness results.

    \begin{acknowledgements}
        We would like to thank Șerban Cercelescu for the useful discussions that led to the inspiration for the pushing out procedure.
We are grateful to Boldizsár Poór, Julio Magdalena De La Fuente, and Maximilian Schweikart for the discussions, inspiration, and feedback on various drafts of this work.
BR thanks Simon Harrison for his generous support for the Wolfson Harrison UK Research Council Quantum Foundation Scholarship. 
AK is supported by the Engineering and Physical Sciences Research Council grant number EP/Z002230/1, “(De)constructing quantum software (DeQS)”. 
BR is employed part-time by Quantinuum. 
The wording in some sections of this paper has been refined using LLMs.
    \end{acknowledgements}

    \cleardoublepage
    \printbibliography[heading=bibintoc,title={References}]

@article{backensZXcalculusCompleteStabilizer2014,
  title = {The {{ZX-calculus}} Is Complete for Stabilizer Quantum Mechanics},
  author = {Backens, Miriam},
  date = {2014-09-17},
  journaltitle = {New Journal of Physics},
  shortjournal = {New J. Phys.},
  volume = {16},
  number = {9},
  eprint = {1307.7025},
  eprinttype = {arXiv},
  eprintclass = {quant-ph},
  pages = {093021},
  issn = {1367-2630},
  doi = {10.1088/1367-2630/16/9/093021},
}

@article{backens2014ZXCalculusCompleteCliffordT,
   title={The ZX-calculus is complete for the single-qubit Clifford+T group},
   volume={172},
   ISSN={2075-2180},
   DOI={10.4204/eptcs.172.21},
   journal={Electronic Proceedings in Theoretical Computer Science},
   publisher={Open Publishing Association},
   author={Backens, Miriam},
   year={2014},
   month=dec, pages={293-303} 
}

@article{coeckeInteractingQuantumObservables2011,
  title = {Interacting {{Quantum Observables}}: {{Categorical Algebra}} and {{Diagrammatics}}},
  shorttitle = {Interacting {{Quantum Observables}}},
  author = {Coecke, Bob and Duncan, Ross},
  date = {2011-04-14},
  journaltitle = {New Journal of Physics},
  shortjournal = {New J. Phys.},
  volume = {13},
  number = {4},
  eprint = {0906.4725},
  eprinttype = {arXiv},
  eprintclass = {quant-ph},
  pages = {043016},
  issn = {1367-2630},
  doi = {10.1088/1367-2630/13/4/043016},
}

@article{bombinUnifyingFlavorsFault2024,
  title = {Unifying Flavors of Fault Tolerance with the {{ZX}} Calculus},
  author = {Bombin, Hector and Litinski, Daniel and Nickerson, Naomi and Pastawski, Fernando and Roberts, Sam},
  date = {2024-06-18},
  journaltitle = {Quantum},
  shortjournal = {Quantum},
  volume = {8},
  eprint = {2303.08829},
  eprinttype = {arXiv},
  eprintclass = {quant-ph},
  pages = {1379},
  issn = {2521-327X},
  doi = {10.22331/q-2024-06-18-1379},
}

@mastersthesis{borghansZXcalculusQuantumStabilizer2019,
  type = {mathesis},
  title = {{{ZX-calculus}} and Quantum Stabilizer Theory},
  author = {Borghans, Coen},
  date = {2019},
  institution = {Radboud University},
  url = {https://www.cs.ox.ac.uk/people/aleks.kissinger/papers/borghans-thesis.pdf},
  urldate = {2025-08-06}
}

@online{derksDesigningFaulttolerantCircuits2024,
  title = {Designing Fault-Tolerant Circuits Using Detector Error Models},
  author = {Derks, Peter-Jan H. S. and Townsend-Teague, Alex and Burchards, Ansgar G. and Eisert, Jens},
  date = {2024-12-25},
  eprinttype = {arXiv},
  eprintclass = {quant-ph},
  doi = {10.48550/arXiv.2407.13826},
  pubstate = {prepublished},
}

@online{gottesmanIntroductionQuantumError2009,
  title = {An {{Introduction}} to {{Quantum Error Correction}} and {{Fault-Tolerant Quantum Computation}}},
  author = {Gottesman, Daniel},
  date = {2009-04-16},
  eprint = {0904.2557},
  eprinttype = {arXiv},
  eprintclass = {quant-ph},
  pubstate = {prepublished},
}

@inproceedings{kissinger2020Pyzx,
  title = {{{PyZX}}: {{Large}} Scale Automated Diagrammatic Reasoning},
  booktitle = {Proceedings 16th International Conference on Quantum Physics and Logic, Chapman University, Orange, {{CA}}, {{USA}}., 10-14 June 2019},
  author = {Kissinger, Aleks and family=Wetering, given=John, prefix=van de, useprefix=true},
  editor = {Coecke, Bob and Leifer, Matthew},
  date = {2020},
  series = {Electronic Proceedings in Theoretical Computer Science},
  volume = {318},
  pages = {229--241},
  publisher = {Open Publishing Association},
  doi = {10.4204/EPTCS.318.14}
}

@book{KissingerWetering2024Book,
  title = {Picturing Quantum Software: {{An}} Introduction to the {{ZX-calculus}} and Quantum Compilation},
  author = {Kissinger, Aleks and family=Wetering, given=John, prefix=van de, useprefix=true},
  date = {2024},
  publisher = {Preprint}, 
  url = {https://github.com/zxcalc/book}
}

@book{Nielsen_Chuang_2010,
  title = {Quantum Computation and Quantum Information: 10th Anniversary Edition},
  author = {Nielsen, Michael A. and Chuang, Isaac L.},
  date = {2010},
  publisher = {Cambridge University Press},
  location = {Cambridge}
}

@online{rodatzFloquetifyingStabiliserCodes2024,
  title = {Floquetifying Stabiliser Codes with Distance-Preserving Rewrites},
  author = {Rodatz, Benjamin and Poór, Boldizsár and Kissinger, Aleks},
  date = {2024-12-16},
  eprint = {2410.17240},
  eprinttype = {arXiv},
  eprintclass = {quant-ph},
  pubstate = {prepublished},
}

@misc{rodatz2025faulttoleranceconstruction,
  title={Fault Tolerance by Construction},
  author={Benjamin Rodatz and Boldizsár Poór and Aleks Kissinger},
  year={2025},
  eprint={2506.17181},
  archivePrefix={arXiv},
  primaryClass={quant-ph},
}

@misc{delfosse2023spacetimecodescliffordcircuits,
  title={Spacetime codes of Clifford circuits},
  author={Nicolas Delfosse and Adam Paetznick},
  year={2023},
  eprint={2304.05943},
  archivePrefix={arXiv},
  primaryClass={quant-ph},
}

@article{bacon2017SparseQuantumCodes,
  author={Bacon, Dave and Flammia, Steven T. and Harrow, Aram W. and Shi, Jonathan},
  journal={IEEE Transactions on Information Theory},
  title={Sparse Quantum Codes From Quantum Circuits},
  year={2017},
  volume={63},
  number={4},
  pages={2464-2479},
  doi={10.1109/TIT.2017.2663199}
}

@misc{vandewetering2020zxcalculusworkingquantumcomputer,
  title={ZX-calculus for the working quantum computer scientist},
  author={family=Wetering, given=John, prefix=van de, useprefix=true},
  year={2020},
  eprint={2012.13966},
  archivePrefix={arXiv},
  primaryClass={quant-ph},
}

@article{Gidney2021stimfaststabilizer,
  doi = {10.22331/q-2021-07-06-497},
  title = {Stim: a fast stabilizer circuit simulator},
  author = {Gidney, Craig},
  journal = {{Quantum}},
  issn = {2521-327X},
  publisher = {{Verein zur F{\"{o}}rderung des Open Access Publizierens in den Quantenwissenschaften}},
  volume = {5},
  pages = {497},
  month = jul,
  year = {2021}
}

@misc{gottesman1997stabilizercodesquantumerror,
  title={Stabilizer Codes and Quantum Error Correction},
  author={Daniel Gottesman},
  year={1997},
  eprint={quant-ph/9705052},
  archivePrefix={arXiv},
  primaryClass={quant-ph},
}

@misc{gidneyStimReadme2021,
  title = {Stim: A fast stabilizer circuit library. Main README.},
  author = {Gidney, Craig},
  year = {2021},
  month = sep,
  url = {https://github.com/quantumlib/Stim/blob/main/README.md},
  urldate = {2025-08-07}
}

@inproceedings{vilmart2019NearMinimalAxiomatisation,
  author={Vilmart, Renaud},
  booktitle={2019 34th Annual ACM/IEEE Symposium on Logic in Computer Science (LICS)},
  title={A Near-Minimal Axiomatisation of ZX-Calculus for Pure Qubit Quantum Mechanics},
  year={2019},
  volume={},
  number={},
  pages={1-10},
  doi={10.1109/LICS.2019.8785765}
}

@article{wangKang2017ZXCalculusUniversalComplation,
  author = {Ng, Kang and Wang, Quanlong},
  year = {2017},
  month = {06},
  title = {A universal completion of the ZX-calculus},
  doi = {10.48550/arXiv.1706.09877},
}

@misc{gottesman2022opportunitieschallengesfaulttolerantquantum,
  title={Opportunities and Challenges in Fault-Tolerant Quantum Computation},
  author={Daniel Gottesman},
  year={2022},
  eprint={2210.15844},
  archivePrefix={arXiv},
  primaryClass={quant-ph},
}

@software{paritea2026code,
  title={Paritea library, including code used during preparation for this work},
  author={Maximilian Rüsch},
  year = {2025},
  url = {https://github.com/paritea/paritea},
  copyright = {Apache-2.0 License}
}

@article{ranchin2014quditZXcalculus,
   title={Depicting qudit quantum mechanics and mutually unbiased qudit theories},
   volume={172},
   ISSN={2075-2180},
   DOI={10.4204/eptcs.172.6},
   journal={Electronic Proceedings in Theoretical Computer Science},
   publisher={Open Publishing Association},
   author={Ranchin, André},
   year={2014},
   month=dec, pages={68–91}
}

@inproceedings{booth2022qupitZXCalculi,
  author={Booth, Robert I. and Carette, Titouan},
  title={{Complete ZX-Calculi for the Stabiliser Fragment in Odd Prime Dimensions}},
  booktitle={47th International Symposium on Mathematical Foundations of Computer Science (MFCS 2022)},
  pages={24:1--24:15},
  series={Leibniz International Proceedings in Informatics (LIPIcs)},
  ISBN={978-3-95977-256-3},
  ISSN={1868-8969},
  year={2022},
  volume={241},
  editor={Szeider, Stefan and Ganian, Robert and Silva, Alexandra},
  publisher={Schloss Dagstuhl -- Leibniz-Zentrum f{\"u}r Informatik},
  address={Dagstuhl, Germany},
  doi={10.4230/LIPIcs.MFCS.2022.24},
}

@article{poor2023qupitCliffordCompleteness,
  title={The Qupit Stabiliser ZX-travaganza: Simplified Axioms, Normal Forms and Graph-Theoretic Simplification},
  author={Poór, Boldizsár and Booth, Robert I. and Carette, Titouan and van de Wetering, John and Yeh, Lia},
  year={2023},
  eprint={2306.05204},
  archivePrefix={arXiv},
  primaryClass={quant-ph},
}

@article{backens2017AsimplifiedStabiliserZXCalculus,
   title={A Simplified Stabilizer ZX-calculus},
   volume={236},
   ISSN={2075-2180},
   DOI={10.4204/eptcs.236.1},
   journal={Electronic Proceedings in Theoretical Computer Science},
   publisher={Open Publishing Association},
   author={Backens, Miriam and Perdrix, Simon and Wang, Quanlong},
   year={2017},
   month=jan, pages={1–20}
}

@article{MagdalenadelaFuente2025XZYruby,
   title={Ruby Code: Making a Case for a Three-Colored Graphical Calculus for Quantum Error Correction in Spacetime},
   volume={6},
   ISSN={2691-3399},
   DOI={10.1103/prxquantum.6.010360},
   number={1},
   journal={PRX Quantum},
   publisher={American Physical Society (APS)},
   author={Magdalena de la Fuente, Julio C. and Old, Josias and Townsend-Teague, Alex and Rispler, Manuel and Eisert, Jens and Müller, Markus},
   year={2025},
   month=mar }

@mastersthesis{rueschCompletenessFaultTolerant2025,
  author={Rüsch, Maximilian},
  title={Completeness for Fault Equivalence of Clifford ZX Diagrams},
  school={University of Oxford},
  year={2025}, 
  url={https://maximilianruesch.de/scientific/msc-thesis/msc-thesis.pdf}, 
  urldate={2025-10-10},
}

@misc{blackwell2025codedistancefloquetcodes,
      title={The code distance of Floquet codes}, 
      author={Keller Blackwell and Jeongwan Haah},
      year={2025},
      eprint={2510.05549},
      archivePrefix={arXiv},
      primaryClass={quant-ph},
}

@inproceedings{coecke2018picturing,
  title={Picturing quantum processes: A first course on quantum theory and diagrammatic reasoning},
  author={Coecke, Bob and Kissinger, Aleks},
  booktitle={International conference on theory and application of diagrams},
  year={2018},
  organization={Springer}, 
  url={https://www.cs.ox.ac.uk/people/aleks.kissinger/PQP.pdf}
}

@misc{shaikh2024fockedupzxcalculuspicturing,
      title={The Focked-up ZX Calculus: Picturing Continuous-Variable Quantum Computation}, 
      author={Razin A. Shaikh and Lia Yeh and Stefano Gogioso},
      year={2024},
      eprint={2406.02905},
      archivePrefix={arXiv},
      primaryClass={quant-ph},
}

@misc{booth2024completeequationaltheoriesclassical,
      title={Complete equational theories for classical and quantum Gaussian relations}, 
      author={Robert I. Booth and Titouan Carette and Cole Comfort},
      year={2024},
      eprint={2403.10479},
      archivePrefix={arXiv},
      primaryClass={cs.LO},
}

@article{Nagayoshi_2025,
   title={ZX graphical calculus for continuous-variable quantum processes},
   volume={7},
   ISSN={2643-1564},
   DOI={10.1103/physrevresearch.7.033141},
   number={3},
   journal={Physical Review Research},
   publisher={American Physical Society (APS)},
   author={Nagayoshi, Hironari and Asavanant, Warit and Ide, Ryuhoh and Fukui, Kosuke and Sakaguchi, Atsushi and Yoshikawa, Jun-ichi and Menicucci, Nicolas C. and Furusawa, Akira},
   year={2025},
   month=aug }

@article{nam2018automated,
  title={Automated optimization of large quantum circuits with continuous parameters},
  author={Nam, Yunseong and Ross, Neil J and Su, Yuan and Childs, Andrew M and Maslov, Dmitri},
  journal={npj Quantum Information},
  volume={4},
  number={1},
  pages={23},
  year={2018},
  publisher={Nature Publishing Group UK London}, 
  doi={https://doi.org/10.1038/s41534-018-0072-4}
}

@misc{qiskit2024,
      title={Quantum computing with {Q}iskit},
      author={Javadi-Abhari, Ali and Treinish, Matthew and Krsulich, Kevin and Wood, Christopher J. and Lishman, Jake and Gacon, Julien and Martiel, Simon and Nation, Paul D. and Bishop, Lev S. and Cross, Andrew W. and Johnson, Blake R. and Gambetta, Jay M.},
      year={2024},
      eprint={2405.08810},
      archivePrefix={arXiv},
      primaryClass={quant-ph}
}

@article{duncanGraphtheoreticSimplification2020,
  title = {Graph-Theoretic {{Simplification}} of {{Quantum Circuits}} with the {{ZX-calculus}}},
  author = {Duncan, Ross and Kissinger, Aleks and Perdrix, Simon and family=Wetering, given=John, prefix=van de, useprefix=true},
  date = {2020-06},
  journaltitle = {Quantum},
  volume = {4},
  pages = {279},
  publisher = {Verein zur Förderung des Open Access Publizierens in den Quantenwissenschaften},
  issn = {2521-327X},
  doi = {10.22331/q-2020-06-04-279}
}

@inproceedings{cowtanPhaseGadget2020,
  title = {Phase Gadget Synthesis for Shallow Circuits},
  booktitle = {Proceedings 16th International Conference on Quantum Physics and Logic, Chapman University, Orange, {{CA}}, {{USA}}., 10-14 June 2019},
  author = {Cowtan, Alexander and Dilkes, Silas and Duncan, Ross and Simmons, Will and Sivarajah, Seyon},
  editor = {Coecke, Bob and Leifer, Matthew},
  date = {2020},
  series = {Electronic Proceedings in Theoretical Computer Science},
  volume = {318},
  pages = {213--228},
  publisher = {Open Publishing Association},
  doi = {10.4204/EPTCS.318.13}
}

@inproceedings{staudacher2024multicontrolled,
  title = {Multi-Controlled Phase Gate Synthesis with {{ZX-calculus}} Applied to Neutral Atom Hardware},
  booktitle = {Proceedings of the 21st International Conference on Quantum Physics and Logic, Buenos Aires, Argentina, July 15-19, 2024},
  author = {Staudacher, Korbinian and Schmid, Ludwig and Zeiher, Johannes and Wille, Robert and Kranzlmüller, Dieter},
  editor = {Díaz-Caro, Alejandro and Zamdzhiev, Vladimir},
  date = {2024},
  series = {Electronic Proceedings in Theoretical Computer Science},
  volume = {406},
  pages = {96--116},
  publisher = {Open Publishing Association},
  doi = {10.4204/EPTCS.406.5}
}
    \cleardoublepage

    \bookmarksetup{startatroot}
    \appendix
    \renewcommand{\thesection}{\Alph{section}}
    \section{Missing proofs}
\label{appendix:missing-proofs}
\soundnessprop*
\begin{proof}
    In this proof, we refer to the sides of the rules in their displayed form as \LHS and \RHS.
    Furthermore, whenever we find an atomic $Z$ flip with weight $w_i$ on diagram $D$, we will refer to it as $Z_{D,i}$, and to the faulty diagram as $D^{Z_{D,i}}$.
\begin{proofparts}
    \item[\textsf{Clifford Rules}] For each such rule, both the \LHS and \RHS are fully idealised.
        Thus, \cref{prop:fully-idealised-equivalence} is applicable to yield the claim directly.

    \item[$\TextFEElim$] The atomic faults on the \LHS have a one-to-one correspondence to the \RHS, such that for each atomic fault the respective faulty diagram is the same.
        This implies that the group of possible faults generated by both diagrams must be isomorphic, and since the weight annotation does not change, the induced weight function is also the same.
        Thus, there cannot exist any atomic fault on one side that does not have a correspondence or even a different atomic weight on the other, and so it can also not be the case with non-atomic faults.

    \item[$\TextFEXPhase$] Both sides only allow either the trivial fault or a single $Z$ flip to occur.
        Then, we can observe that for both faults, the faulty diagrams are directly the same:
        \[ \input{A-missing-proofs/x-phase-faulty-diagram-proof.tikz} \]
        Furthermore, the annotated weight of $Z_{\LHS}$ and $Z_\RHS$ is the same, so the induced weight function must also be the same, yielding fault equivalence.

    \item[$\TextFECommute$] The proof is analogous to the previous case, in that we can again reason about the faulty diagrams:
        \[ \input{A-missing-proofs/commute-faulty-diagram-proof.tikz} \]
        We are able to remove the scalar in the intermediate step as there is no $k_1,k_2$ where the scalar is zero, and we do not distinguish faulty diagrams up to a scalar.

    \item[$\TextFEScalar$] The \RHS is fully idealised and thus trivially fault-bound by the \LHS.
        Furthermore, we can observe that all faults on the \LHS have faulty diagrams equal to the trivial fault $I$:
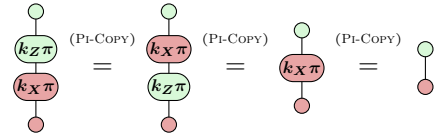
        \[ \input{A-missing-proofs/scalar-faulty-diagram-proof.tikz} \]
        Since $\weightfunc(I) = 0 \leq \weightfunc(F)$ for any fault $F$ on the \LHS, every fault has a correspondence with less or equal weight, providing fault equivalence by definition.

    \item[$\TextFEMerge$] Via repeated application of $\Fusion$, both $Z_{\LHS,1},Z_{\LHS,2}$ have a faulty diagram equivalent to that of as $Z_\RHS$, while their product and the trivial fault for the \LHS have an equivalent faulty diagram to the trivial fault on the \RHS.
        So every fault that can be generated on either side has at least one corresponding fault on the other, thus checking their induced weight remains.
        But everything corresponding to a trivial fault forms a straightforward case, and by definition of the minimum, it holds that
        \[ \forall i \in 1,2: \min(w_1,w_2) \leq w_i \quad\text{and}\quad \exists i \in 1,2: w_i \leq \min(w_1,w_2)\,, \]
        providing the other cases.

    \item[$\TextFEComb$] Via application of $\Copy$ and $\Fusion$, we see that both $Z_{\LHS,1},Z_{\LHS,2}$ have one-to-one correspond to $Z_{\RHS,1},Z_{\RHS,2}$, and so their product also has a correspondence to the product on the other side.
        We then consider $Z_{\RHS,12}$, and see that it forms an additional possibility to correspond to the product on the \LHS:
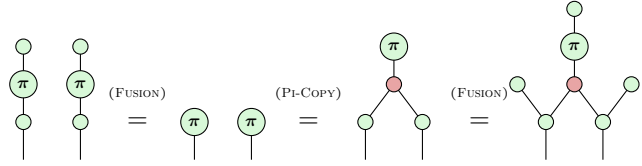
        \[ \input{A-missing-proofs/comb-faulty-diagram-proof.tikz} \]
        Thus, any combination on the \RHS is reproducible on the \LHS, so arguing about their induced weight remains.

        On both sides, the atomic faults weighted with $w_1,w_2$ are already annotated with their induced weight.
        Additionally, the atomic fault weighted with $w_1 + w_2$ on the \RHS corresponding to the product on the \LHS is exactly weighted with the weight sum of said product.
        Thus, every atomic fault on both sides has a correspondence of exactly equal weight.
        Any non-atomic fault $F$ has at least one atomic fault with the same faulty diagram, and it is straightforward to check that the weight of $F$ is at least that of the atomic fault.
        Thus, every fault on either side has a correspondence of at most the same weight.

    \item[$\TextFEDetect$] First, we see that the faults annotated with $w_2,\dots,w_n$ on the \LHS have a direct correspondence to the faults on the \RHS, and vice versa.
        For $Z_{\LHS,1}$, observe that the free-floating spider on either side forms a detecting region with the annotated edges:
        \[ \input{A-missing-proofs/detect-pauli-webs.tikz} \]
        Thus, $Z_{\LHS,1}$ in particular is detectable.
        Further, observe that the faults annotated with sums on the \RHS generate the faulty diagrams of all undetectable combinations of $Z_{\LHS,1}$ with other atomic faults.
        This implies that removing/introducing $Z_{\LHS,1}$ connected to the free-floating spider does not remove/introduce any new class of possible faulty diagrams.
        The reasoning about weights is directly analogous to the case of $\TextFEDetect$, and so fault equivalence follows.
\end{proofparts}
\end{proof}

\addingspiderstabilisers*
\begin{proof}
    In the main body, we have already shown how to add a $ZZ$ stabiliser to a green spider.
    The other cases follow similarly:
    For a green spider with a phase of $k \in \{0, \pi\}$, we can add an $X$ stabiliser as follows:
    \[ \input{A-missing-proofs/stab-fe-opposite-color-k-pi.tikz} \]

    For a green spider with a phase of $\frac{\pi}{2}$, we can add an all $X$ and a single $Y$ stabiliser as follows:
    \begin{equation}
        \label{eq:stab-fe-proof-opposite-color-pi-2}
         \vcenter{\hbox{\input{A-missing-proofs/stab-fe-opposite-color-pi-2.tikz}}} 
    \end{equation}

    For a green spider with a phase of $-\frac{\pi}{2}$, we can add an all $X$ and a single $Y$ stabiliser as follows:
    \[ \input{A-missing-proofs/stab-fe-opposite-color-minu-pi-2.tikz} \]

    As these generate all spider stabilisers, we can now add any generating set of spider stabilisers to the fault gadget.
    Finally, to complete the proof, we have to merge all the targets as shown in~\cref{prop:changing-fault-gadgets-1}.
\end{proof}

\wsoundnessprop*
\begin{proof}
    For any two diagrams $D_1,D_2$ under any two noise models, both $D_1 \wFaultEq{w} D_2$ and $D_1 \FaultBnd D_2$ imply $D_1 \wFaultBnd{w} D_2$.
    So if both rules are sound \wrt their respective properties, both are also sound \wrt $w$-fault boundedness individually.
    We show their individual soundness:
\begin{proofparts}
    \item[$\TextFELimit$] The \RHS is fully idealised and thus trivially fault-bound by the \LHS.
        As the single non-trivial atomic fault on the \LHS is by the side condition at or above weight $w$, checking $w$-fault equivalence reduces to finding a correspondence for the trivial fault $I_\LHS$, which is naturally $I_\RHS$ with a straightforward weight check.

    \item[$\TextFBndBound$] We distinguish the cases $w = -$ and $w \geq w'$.
        In the first case, $w = -$ is just a more explicit version of an idealised edge, yielding a fully idealised \LHS and thus leading to the claim trivially.

        In the second case, the only non-trivial faults on both sides are $Z_\LHS$ and $Z_\RHS$, and they clearly have the same faulty diagram.
        Via the assumed side condition, it holds that $\weightfunc(Z_\LHS) = w \geq w' = \weightfunc(Z_\RHS)$.
        Thus, for all $\tilde(w) \in \mathbb{N}^+$, all conditions for $\widetilde{w}$-fault boundedness are fulfilled, yielding the claim.
\end{proofparts}
\end{proof}
    \section{Additional Derived ZX Calculus Rules}
We introduce three additional ZX calculus rules, derivable from our axiomatisation of the ZX calculus:
\begin{proposition}[Hopf Rule]
    \begin{equation}
        \label{eq:hopf-rule}
        \input{02-preliminaries/hopf.tikz}
    \end{equation}
\end{proposition}
\begin{proof}
    \[ \input{02-preliminaries/hopf-proof.tikz}\, \qedhere \]
\end{proof}

\begin{proposition}[Alternative Hadamard Decomposition]
    \begin{equation}
        \label{eq:hadamard-decomposition}
        \input{02-preliminaries/hadamard-decomposition.tikz}
    \end{equation}
\end{proposition}
\begin{proof}
    The first equality holds by definition of the Hadamard gate.
    The second equality is easily shown in the ZX calculus:
    \[ \input{02-preliminaries/hadamard-decomposition-proof.tikz} \qedhere \]
\end{proof}

\begin{proposition}[Transformation of one-legged spiders from~\autocite{backens2017AsimplifiedStabiliserZXCalculus}]
    \begin{equation}
        \label{eq:state-conversion}
        \begin{tikzpicture}
	\begin{pgfonlayer}{nodelayer}
		\node [style=X phase dot] (0) at (-1.5, 0.5) {$\frac{\pi}{2}$};
		\node [style=none] (1) at (-1.5, -0.5) {};
		\node [style=none] (2) at (-1.5, -0.5) {};
		\node [style=none] (3) at (0, 0) {$=$};
		\node [style=none] (38) at (1.25, -0.5) {};
		\node [style=Z phase dot] (39) at (1.25, 0.5) {$\minu\frac{\pi}{2}$};
	\end{pgfonlayer}
	\begin{pgfonlayer}{edgelayer}
		\draw (0) to (2.center);
		\draw (39) to (38.center);
	\end{pgfonlayer}
\end{tikzpicture}

    \end{equation}
\end{proposition}
\begin{proof}
    \[ \begin{tikzpicture}
	\begin{pgfonlayer}{nodelayer}
		\node [style=X phase dot] (0) at (-4, 1) {$\frac{\pi}{2}$};
		\node [style=none] (1) at (-4, 0) {};
		\node [style=none] (2) at (-4, 0) {};
		\node [style=none] (3) at (-2.5, 0.5) {$=$};
		\node [style=Z phase dot] (4) at (-1, 1.25) {$\frac{\pi}{2}$};
		\node [style=none] (5) at (-1, -0.25) {};
		\node [style=none] (6) at (-1, -0.25) {};
		\node [style=hadamard] (7) at (-1, 0.5) {};
		\node [style=none] (8) at (0.5, 0.5) {$=$};
		\node [style=Z phase dot] (9) at (2, 2.5) {$\frac{\pi}{2}$};
		\node [style=none] (11) at (2, -1.25) {};
		\node [style=none] (13) at (0.5, 1.25) {$\Eq{eq:hadamard-decomposition}$};
		\node [style=Z phase dot] (15) at (2, 1.5) {$\minu\frac{\pi}{2}$};
		\node [style=Z phase dot] (16) at (2, -0.5) {$\minu\frac{\pi}{2}$};
		\node [style=X phase dot] (17) at (2, 0.5) {$\minu\frac{\pi}{2}$};
		\node [style=none] (18) at (3.5, 0.5) {$=$};
		\node [style=Z dot] (19) at (5, 1.5) {};
		\node [style=none] (20) at (5, -1.25) {};
		\node [style=Z phase dot] (22) at (5, -0.5) {$\minu\frac{\pi}{2}$};
		\node [style=X phase dot] (23) at (5, 0.5) {$\minu\frac{\pi}{2}$};
		\node [style=none] (24) at (6.5, 0.5) {$=$};
		\node [style=Z dot] (25) at (7.75, 1.5) {};
		\node [style=none] (26) at (7.75, -1.25) {};
		\node [style=Z phase dot] (27) at (7.75, -0.5) {$\minu\frac{\pi}{2}$};
		\node [style=X phase dot] (28) at (8.75, 0.5) {$\minu\frac{\pi}{2}$};
		\node [style=X dot] (29) at (7.75, 0.5) {};
		\node [style=none] (30) at (10.25, 0.5) {$=$};
		\node [style=Z dot] (31) at (11.5, 0.25) {};
		\node [style=none] (32) at (11.5, -1.25) {};
		\node [style=Z phase dot] (33) at (11.5, -0.5) {$\minu\frac{\pi}{2}$};
		\node [style=X phase dot] (34) at (13, 0.5) {$\minu\frac{\pi}{2}$};
		\node [style=Z dot] (35) at (12.25, 0.5) {};
		\node [style=none] (36) at (14.5, 0.5) {$=$};
		\node [style=none] (38) at (15.75, 0) {};
		\node [style=Z phase dot] (39) at (15.75, 1) {$\minu\frac{\pi}{2}$};
		\node [style=none] (40) at (-2.5, 1.25) {$\Colour$};
		\node [style=none] (41) at (3.5, 1.25) {$\Fusion$};
		\node [style=none] (42) at (6.5, 1.25) {$\Fusion$};
		\node [style=none] (43) at (10.25, 1.25) {$\Copy$};
		\node [style=none] (44) at (14.5, 1.25) {$\Fusion$};
	\end{pgfonlayer}
	\begin{pgfonlayer}{edgelayer}
		\draw (0) to (2.center);
		\draw (4) to (6.center);
		\draw (9) to (11.center);
		\draw (19) to (20.center);
		\draw (25) to (26.center);
		\draw (29) to (28);
		\draw (31) to (32.center);
		\draw (35) to (34);
		\draw (39) to (38.center);
	\end{pgfonlayer}
\end{tikzpicture}
 \qedhere \]
\end{proof}

\end{document}